\documentclass[a4paper,11pt]{article}

\RequirePackage[OT1]{fontenc}

\usepackage{amsthm,amsmath,bbm,amssymb,graphicx,booktabs,array,fullpage,url,mathtools,wrapfig,lipsum,mathrsfs,dsfont,titling,epstopdf,bm,relsize,caption,xcolor,algorithm,enumitem,multirow,soul, makecell}

\usepackage{natbib}
\usepackage[noend]{algpseudocode}
\usepackage[colorlinks=true,
linkcolor=blue,
urlcolor=blue,
citecolor=blue]{hyperref}

\newtheorem{conj}{Conjecture}
\newtheorem{thm}[conj]{Theorem}
\newtheorem{cor}[conj]{Corollary}
\newtheorem{prop}[conj]{Proposition}
\newtheorem{lemma}[conj]{Lemma}
\newtheorem{ass}{Assumption}
\newtheorem{fact}{Fact}

\newtheorem{innercustomgeneric}{\customgenericname}
\providecommand{\customgenericname}{}
\newcommand{\newcustomtheorem}[2]{%
	\newenvironment{#1}[1]
	{%
		\renewcommand\customgenericname{#2}%
		\renewcommand\theinnercustomgeneric{##1}%
		\innercustomgeneric
	}
	{\endinnercustomgeneric}
}

\newcustomtheorem{customAss}{Assumption}

\theoremstyle{definition}\newtheorem{remark}{Remark}

\def\sumsum{\mathop{\sum\sum}}
\def\Cov{\text{Cov}}   
\def\PP{\mathbb{P}}
\def\EE{\mathbb{E}}

\def\E{\mathcal{E}}
\def\X{\mathbf{X}}
\def\Z{\mathbf{Z}}
\def\W{\mathbf{W}}

\def\R{\mathbb{R}}
\def\y{\mathbf{y}}

\def\wh{\widehat}
\def\eps{\varepsilon}
\def\I{\mathcal{I}}
\def\w{\gamma_w}
\def\z{\gamma_z}
\def\H{\mathcal{H}}
\def\t{\tau}

\def\i{\infty}
\def\bI{\bf{I}}
\def\e{{\bf{e}}}
\def\beps{\mathbf{\varepsilon}}
\def\D{\Delta}
\def\wt{\widetilde}
\let\emptyset\varnothing
\def\oW{\overline{\mathbf{W}}}
\def\oX{\overline{\mathbf{X}}}
\def\S{\mathcal{S}}
\def\C{\Sigma_Z}
\def\whC{\wh \Sigma_Z}
\def\cu{C_{\max}}
\def\cl{C_{\min}}
\def\0{{\ell_0}}
\def\1{\bm{1}}
\def\pn{(p\vee n)}
\def\bt{\overline{\tau}^2}
\def\N{\mathcal{N}}

\def\sz{\Sigma_Z}

\newcommand{\sbt}{\,\begin{picture}(-1,1)(-0.5,-2)\circle*{2.3}\end{picture}\ }

\begin{document}
	
	\title{Inference in latent factor regression with clusterable features}
	
	\author{Xin Bing\thanks{Department of Statistics and Data  Science, Cornell University, Ithaca, NY. E-mail: \texttt{xb43@cornell.edu}.}~~~~~Florentina Bunea\thanks{Department of Statistics and Data Science, Cornell University, Ithaca, NY. E-mail: \texttt{fb238@cornell.edu}.}~~~~~Marten Wegkamp\thanks{Department of Mathematics and Department of Statistics and Data Science, Cornell University, Ithaca, NY. E-mail: \texttt{mhw73@cornell.edu}. }~~~~~
	}
	\date{}
	\maketitle
	\vspace{-0.5in}

	\begin{abstract}
		
Regression models,  in which
the observed  features $X \in \R^p$ and the response  $Y \in \R$  depend, jointly, on a lower dimensional, unobserved, latent vector $Z \in \R^K$, with $K\ll p$, are popular in a large array of applications, and mainly used for predicting a response from correlated features.  In contrast, methodology and theory for inference on the  regression coefficient $\beta\in \R^K$ relating $Y$ to $Z$ are  scarce, since typically   the un-observable factor $Z$ 
is hard to interpret. Furthermore, the determination of the asymptotic variance of an  estimator of $\beta$ is a long-standing problem, with solutions known only in  a few particular cases. 

To address some of these outstanding questions, we develop inferential tools for $\beta$ in a class of factor regression models in which the observed features are signed mixtures of the latent factors. The model specifications are  practically desirable, in a large array of applications, render interpretability to the components of $Z$,  and are sufficient for parameter identifiability. 

Without assuming that the number of latent factors $K$ or the structure of the mixture is known in advance, we construct computationally efficient estimators of $\beta$, along with estimators of other important model parameters. We benchmark the rate of convergence of $\beta$ by first establishing its $\ell_2$-norm  minimax lower bound, and show that our proposed estimator $\wh \beta$ is minimax-rate adaptive. Our main contribution is the provision of a unified analysis of the component-wise Gaussian asymptotic distribution of $\wh \beta$ and, especially, the derivation of a  closed form expression of its asymptotic variance, together with consistent variance estimators.  The resulting inferential tools can be used when both $K$ and $p$ are independent of the sample size $n$, and also when both, or either, $p$ and $K$ vary with $n$, while allowing for $p > n$. This complements the only asymptotic normality results obtained for a particular case of the model under consideration, in the regime $K = O(1)$ and $p \rightarrow\infty$, but without a variance estimate. 

		As an application, we provide, within our model specifications, a statistical platform for inference in regression on latent cluster centers, thereby increasing the scope of our theoretical results. 
		
		We benchmark the   newly developed  methodology on a recently collected data  set for  the study of the effectiveness of a new SIV vaccine. Our analysis enables the determination of the top latent  antibody-centric mechanisms associated with the vaccine response.

		
		
	\end{abstract}


	\section{Introduction}
	
	Latent factor models have been used successfully for several decades for modeling data with embedded low dimensional structures. In particular, they provide a natural framework for regression problems in which the covariate vector $X \in \R^p$ and the response $Y \in \R$ are {\it jointly} low dimensional. In a  latent factor regression model, this is formalized by assuming that there exists a random vector $Z \in \R^K$, for some {\em unknown} $K < p$, that is connected to
	the observed pair $(X,Y)\in \R^p\times \R$ via the  model
	\begin{align}\label{model}
	Y& = Z^\top \beta + \eps \\
	X &= AZ + W. \label{love}
	\end{align}
	The dimension $K$,  matrix $A\in\R^{p\times K}$ and vector $\beta\in \R^K$ are unknown.
	The random vectors 
	$Z$ and $W$ and random variable
	$\eps$ are independent,
	with zero means, $\EE[Z]={\bf 0}$, $\EE[W]={\bf 0}$ and $\EE[\eps]=0$, and covariance matrices $\C:=\Cov{(Z)}$ and $\Gamma:=\Cov{(W)}$, and variance $\sigma^2:= \EE[\eps^2]$, respectively.

	Factor regression models, and their many variants \citep{Connor-Koraj,Forni1996,Forni2000,SW2002_JASA,SW2002_JB,Bai-Ng-CI,Bair_JASA, BoivinNg,Blei:2007,Bai-Ng-forecast, Partial_Factor_Modeling, Kelly-2015,rei2016nonasymptotic,fan2017,bunea2020interpolation}
	have been introduced  to motivate and analyze   prediction  schemes for  $Y\in \R$ from $X \in \R^p$, when $p$ is very large and the components of $X$ are highly correlated. Parameter identifiability is not required for prediction purposes, as unique predictors can still be constructed when the assignment matrix  $A$ and  the covariance matrix $\C$ of $Z$  are identifiable only up to orthogonal transformations. 
	
	Substantially   less work has been devoted to inference in factor models, the  problem treated in this work. Classical factor analysis for an observable  vector $U \in \R^d$ postulates the existence of factors $Z \in \R^K$ such that $U = BZ + E$, for some $d \times K$ factor loading matrix.  Factor regression models are  an instance of factor models, where one emphasizes the different roles, response and covariates, respectively,  of the  observable 
	variables $U = (X, Y)$, and $B$ consists in the matrix $A$ augmented by the vector $\beta$. 
	
	This paper proposes and analyses computationally efficient estimators for inference on the regression coefficient $\beta$ in identifiable 
	and interpretable factor regression models, an under-explored problem. 
	
	\subsection{A framework for regression on interpretable latent factors}\label{framew}
	
	
	We begin by summarizing the  model parameters, the nature of the data, as well as the relation between parameter dimensions and sample size. Throughout this work we assume that we have access to an i.i.d. sample $(X_1, Y_1),\dots,(X_n,Y_n)$ of $(X, Y) \in {\R}^p \times \R$, and that $(X, Y)$  have mean zero and satisfy (\ref{model}) and (\ref{love}).
	
	We allow for $p > n$, while  $K < p$.
	In this work, we consider the case of non-sparse $\beta$, and { $K<\sqrt{n}$}, but  allow $K$ to grow with the sample size $n$. The complementary cases of { $K>\sqrt{n}$} and $\beta$ sparse will be studied in a follow-up work.

	
	Our central interest is on valid inference for $\beta$, which  first requires establishing its identifiability. Restrictions on  generic factor models of the type $U = BZ + E$ under which the model parameters are identifiable can be traced back to \cite{lawley_1940}.  A very detailed exposition of possible identifiability restrictions was first collected in the seminal work of \cite{anderson1956}. They have been  revisited in several works, for instance, \cite{lawyley1971,bollen1989,yalcin2001,bai2012}. Of those restrictions on $B$, some are of purely  mathematical convenience \citep[Sections 5 and 6]{anderson1956}, whereas, as considered in this work,  others are practically  interpretable, \cite{anderson1956,anderson1988}.

	
	We focus on a class of identifiable factor regression models in which the observed covariates $X$ are signed mixtures of the latent vector's components $Z_k$, $1 \leq k \leq K$,   with  unknown $K$. 
	A  latent $Z_k$ can be interpreted as the representative of one of the mixtures. Inference on $\beta$ is thus inference for the mixture representatives.  The nature of a representative $Z_k$ is the nature of those few observed $X_j$'s that are connected only to that $Z_k$, justifying their name, pure variables, indexed by $I_k$, with their totality  indexed by $I := \cup_{k = 1}^{K}I_k \subseteq \{1, \ldots, p\}$. In Section \ref{sec_model_assump} we formalize this model class, and show that it is identifiable.

	
	
	Versions of this factor model class are  routinely used in educational and psychological testing, where the latent variables are viewed as aptitudes or psychological states \citep{thurstone,anderson1956,bollen1989}. The $X$-variables are test results, with some tests  
	specifically designed to measure only one single aptitude $Z_k$, for each aptitude, whereas others test mixtures of aptitudes. By experimental design, $I$ and $K$ {\it are  known} in  this classical literature. 
	
	Another  important  application of this factor regression model, in which both $I$ and $K$ {\it are unknown}, is to the analysis of biological data sets with hidden signatures.  The data sets that we discuss in Section \ref{sec_real_data} have, by design, what we termed pure variables. Furthermore, 
	because of the inherent biological redundancy built into the multi-omic screens that generate the  components of $X$, one expects at least two of the  $X$-variables effectively measuring the same biological signature, and only that one.  For instance, there can be two  or more paralogous genes with one very specific function, or  two  or more different subsets of immune cells carrying out the same specific niche immunological function. Such signatures are known to exist, but cannot be measured directly, and correspond to the components of the latent vector $Z$. Whereas some of these functions are known, it is one of the purposes of  the analysis to discover new ones, as well as new $X$-variables solely associated with them. Thus,  neither $I$ nor $K$ can be treated as known, nor can $K$ be treated as fixed, since the  number of functions $K$ can grow as  $p$ grows, which can in turn grow with $n$.

	Our first contribution is to propose this flexible and, in many applications, more realistic, framework  for estimation and inference on $\beta$ in factor models with pure variables, in which the index set $I$ and the number of factors $K$ are  not known and have to be estimated from the data. All the results of this paper are derived in this context, which is the first  point of departure from previous results for inference in factor models derived for models with $I$ and $K$ known, for instance in  \cite{anderson1956,anderson1988,yalcin2001,bai2012}. 
	
	To emphasize the specific usage of factor models for regression and inference on mixture representatives, under the model specifications formally given in Section \ref{sec_model_assump}, we refer to it as {\sl Essential Regression}.

	Figure \ref{ER} below gives an instance of Essential Regression. 
	A response $Y$ depends on three  latent factors $(Z_1, Z_2, Z_3)$, which in turn  are connected to $(X_1,\ldots,X_{10})$.  The measured variables $X_1$ and $X_2$ 
	have only (100\%)  function $Z_1$. The  $\pm$ 1  edge weights indicate that  this function activates $X_1$ and inhibits $X_2$. Variable $X_3$ has mixed functions, $50\%$ is devoted to function $Z_1$, and the sign indicates that $Z_1$ is an inhibitor, while  $30\%$ is devoted to function $Z_2$, an activator. The fact that the weights, in absolute value,  do not sum up to 1 increases the model flexibility, by allowing  free association between $X_3$ and  other functions that are not explained by this model.


	\begin{figure}[ht]
		\centering
		\includegraphics[width=.8\textwidth]{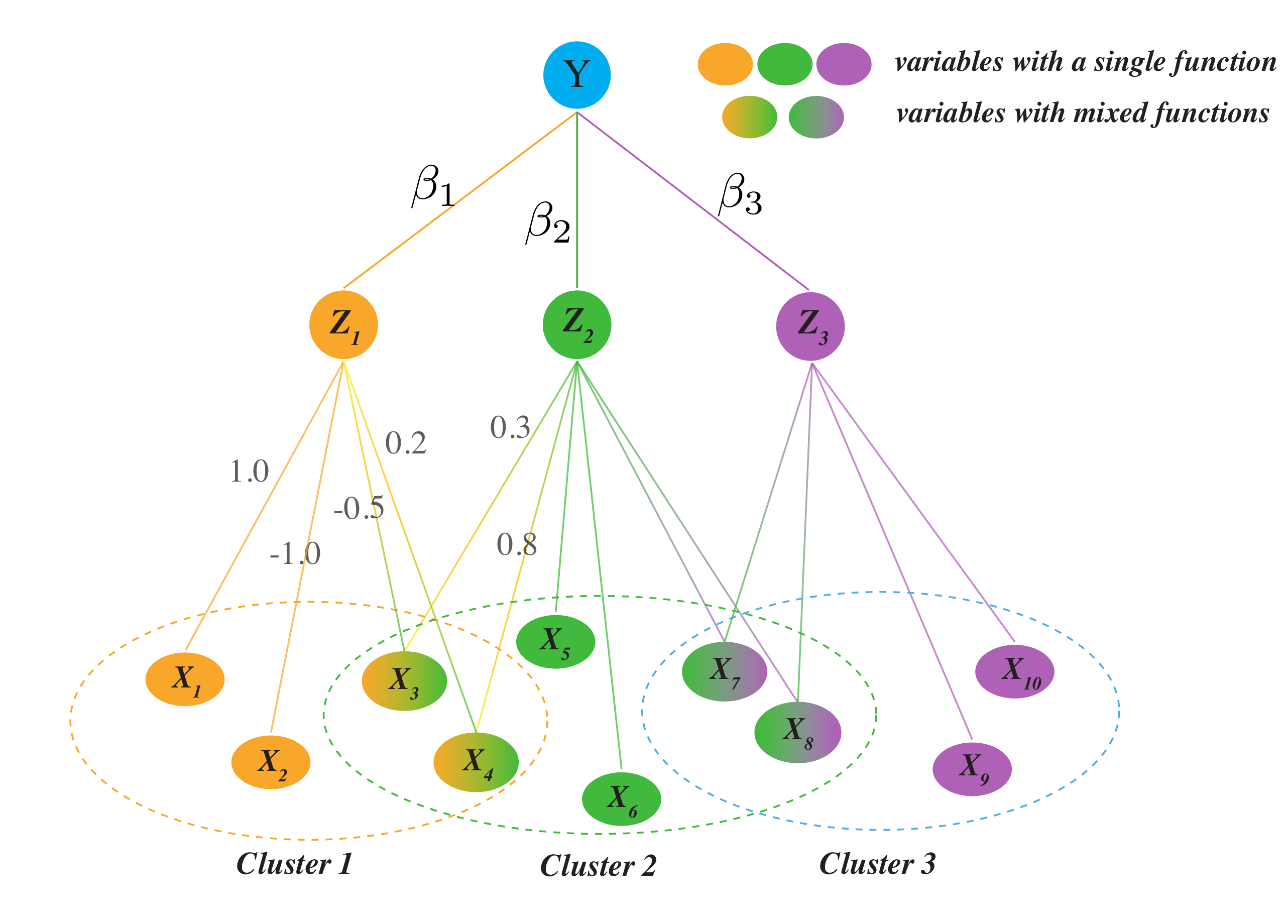}
		\caption{An illustrative example of Essential Regression}
		\label{ER}
	\end{figure}

	A similar, data-driven,  figure is presented in Section \ref{sec_real_data}, in which we   show that the  Essential Regression model fits the data collected in a new SIV-study (SIV is the non-human primate equivalent of HIV), and offers insights into immunological signatures   driving the vaccine response. This example illustrates a scientifically-desirable way of modeling a response $Y$ directly at the function ($Z$) level,  when the  observed $X$-variables  have either single or mixed functions.

	\subsection{Our contributions }

	\paragraph{Minimax adaptive estimation of $\beta$.} 	In Section \ref{sec_est_beta}, we  construct a computationally tractable estimator of $\beta$.
	As part of the estimation procedure, the following unknown quantities are also estimated from the data, under the Essential Regression model: $A$, $I$, $K$, $\C:= \Cov(Z)$, $\sigma^2 := \EE[\eps^2]$ and $\Gamma := \Cov(W)$. To benchmark the quality of our estimator we derive first, in  Theorem \ref{thm_beta_lower} of Section \ref{sec_beta_lower}  the minimax optimal rate of estimating $\beta$ in 
	$\ell_2$-norm in an Essential Regression model. We show in Theorem \ref{thm_beta} of Section \ref{sec_consistency_beta} that the proposed estimator $\wh \beta$ is minimax  rate adaptive, up to logarithmic factors in         $n$ and $p$.
	Our result uses the fact that  only the estimation of the $|I|\times K$ sub-matrix $A_{I \sbt}$,
	instead of the entire $p \times K$ matrix $A$, is 
	involved in the construction of $\wh \beta$. In Section \ref{sec_other_estimates} we 
	introduce and discuss  various competitors, including an estimator of $\beta$ that utilizes estimators of  the full matrix $A$. These estimators are natural consequences of rewriting the identity for $\beta=\C^{-1} \text{Cov}(Z,Y)$, see (\ref{def_beta}), (\ref{iden_beta_full}), (\ref{iden_beta_I_full}),
	(\ref{iden_beta_mm}) and (\ref{iden_beta_I}).
	We give insight into why these estimators are less efficient than the   estimator proposed and studied, and confirm this in our simulation study in Section \ref{sec_sims}.
	
	
	To estimate $I$ and $K$  we use the method proposed in  \cite{LOVE}, as it  guarantees that we can consistently estimate $K$,  without imposing any restrictions on our target for inference, $\beta$. Furthermore, this method also guarantees that $ I \subseteq \wh I \subseteq I \cup J_1$, where $J_1$ is an index  set of what we term quasi-pure variables, defined formally in Section \ref{sec_ass}. As the name suggests, a quasi-pure variable is a measured $X$-variable that is very strongly  associated with only one $Z_k$, while having very small, but non-zero, association with other latent factors. A signal strength assumption on the entries of $A$ would render $J_1 = \emptyset$, which would simplify the analysis of $\wh \beta$ considerably. 
	
	To maintain a flexible modeling framework, the proofs of all our results, rate optimality and asymptotic distribution, allow for the presence of quasi-pure variables, $J_1 \neq \emptyset$, while controlling their relative number via Assumptions \ref{ass_J1} and \ref{ass_J1_prime}. The price to pay for considering a more  realistic scenario  is an increase in the  technical difficulty of the proofs of Theorems \ref{thm_beta} and \ref{thm_distr} and Proposition \ref{prop_est_V}, for instance in Lemmas \ref{lem_quad}, \ref{lem_H_op}, \ref{lem_Rem_2}, \ref{lem_Z_wt_W_wt}, \ref{lem_wt_WW_D}, \ref{lem_1}.


	\paragraph{Inference for $\beta$:  component-wise limiting distribution, asymptotic variance and its estimates.} Within various classes of identifiable factor models, and in the classical set-up $K$ and $p$ fixed, \cite{anderson1956, anderson1988}
	proposed  MLE-based estimators of the rows of identifiable loading factors $B$, in a generic factor model $U = BZ + E$.  They  pointed out that the asymptotic covariance matrix of the Gaussian limit of their estimators has a very involved expression, and left its derivation open. 
	
	In the regime $K$ fixed  and $p \to \infty$, \cite{bai2012} offered a solution to this problem, two decades later.  They derived the asymptotic distribution, including the expression of the  limiting covariance, of  MLE-inspired estimators of the rows of $B$, under various identifiability restrictions on $B$, including a version of the conditions given in Section \ref{sec_inference_beta}  below, corresponding to $I$ and $K$ known. Their proof uses a linearization argument, and requires $p \rightarrow \infty$ to establish that the corresponding  remainder term converges in probability to zero. The practical implementation of the estimator involves  an EM-type algorithm that is very sensitive to initialization, and becomes problematic in high dimensions.  The estimation of the limiting covariance is not considered in their work. 
	
	Computationally feasible estimators of the rows of $B$, and in particular of the entries of $\beta$,  with closed form, estimable, asymptotic  variances continue to be lacking in the classical regime $K, p$ fixed, and also when both dimensions are allowed to grow. Furthermore, no results of this type have been established when $K$ and $I$ are unknown. 
	
	As our main contribution, we address these open questions in this work, via a unifying analysis, by studying the component-wise distribution of estimators of $\beta_k$, for $1 \leq k \leq K$
	
	Theorem \ref{thm_distr} of Section \ref{sec_inference_beta} shows that the  computationally tractable estimator proposed in Section \ref{sec_est_beta} is asymptotically normal, with consistently estimable variance, under all scenarios of interest. A consistent  estimator of this variance is given in Section \ref{sec_inference_beta} and its consistency is proved in Proposition \ref{prop_est_V}.  Table \ref{tab_infer} below offers  a snap shot of our asymptotic normality results, relative to  existing results.

	\begin{table}[H]
		\centering
		{\renewcommand{\arraystretch}{2}{
				\resizebox{\textwidth}{!}{
					\begin{tabular}{|c!{\vrule width 1.5pt}c|c|c!{\vrule width 1.5pt}c|c|c|}
						\hline
						& \multicolumn{3}{c!{\vrule width 1.5pt}}{$K$ and $I$ both known} & \multicolumn{3}{c|}{$K$ and $I$ both unknown}\\\cline{2-7}
						& $K, p$ fixed & $K$ fixed, $p\to \i$ & $K,p\to \i$ & $K, p$ fixed & $K$ fixed, $p\to \i$ & $K, p\to \i$\\\cline{1-7}
						\thead{Existing\\  results} & \thead{\cite{anderson1956}, \\ MLE estimator,  \\ no closed form of the \\ asymptotic variance.}
						&  \thead{\cite{bai2012}, \\ MLE-inspired estimator,  \\
							computationally involved; \\ closed form  asymptotic\\ variance $Q_k$
							when $\lambda_K \asymp p$.}
						& NA & NA & NA & NA\\
						\hline 
						\multirow{2}{*}{\thead{\\
								Theorem \ref{thm_distr},\\
								Section  \ref{sec_inference_beta}
						}} & 
						\multicolumn{6}{c|}{Computationally tractable $\wh\beta_k$ and  asymptotic variance $V_k$}\\\cline{2-7}
						& 
						\checkmark  & \thead{$V_k$ reduces to $Q_k$ when\\ $\lambda_K \to\i$ and
							\\
							$\lambda_K\gtrsim ({p}/{\sqrt{n}}) \log (p \vee n)$.} & \checkmark & \checkmark& \checkmark & \checkmark\\
						\hline 
					\end{tabular}
				}
		}}
		\caption{Asymptotically normal estimators of $\beta_k$ in a class of factor models: existing  and new results (Theorem \ref{thm_distr}).}
		\label{tab_infer}
	\end{table}

	The quantity  $\lambda_K :=\lambda_K(A \C A^{\top})$  quantifies  the size of the signal in $X=AZ+W$. 
	Theoretical analyses under the regime $p>n$ are performed under a  conservative signal strength assumption, $\lambda_K \asymp p$, in the 
	existing literature on factor models.
	See,  for instance, 	\citep{Chamberlain-1983,Connor-Koraj,Bai-Ng-forecast, Bai-factor-model-03, fan2011, fan2013large}. This includes results pertaining to inference on $\beta$ of \cite{bai2012}, which are most closely related to  our work.  
	
	In Section \ref{sec_inference_beta} we prove that our proposed  estimators  of $\beta$   attain a Gaussian limit under a considerably relaxed condition, $ \lambda_K \gtrsim{p/\sqrt{n}} $
	(up to  logarithmic $\log\pn$ factors), within a framework in which  $K$ can  grow  as fast as $O(\sqrt n/\log\pn)$. A technical discussion of this condition is provided in Section \ref{sec_ass}, and in Remark \ref{rem_inference_results} of Section \ref{sec_inference_beta}.
	
	Table 1 offers a complete picture, to the best of our knowledge, of the existing literature on   inference for $\beta$, under the modelling framework considered in this work.  
	For completeness, we summarize in Remark \ref{rem_K} of Section \ref{sec_consistency_beta},  other  approaches proposed in the literature for the selection of $K$, in other identifiable factor models.    We summarize them, and state sufficient conditions for their consistency in Table \ref{tab_K}.\\
	

	For clarity of presentation, we give below the expressions of the limiting variances in the particular case when  $\Cov(W)=\tau^2\bI_p$. By letting $\sigma^2 = \EE[\eps^2]$, $\Cov(W)=\tau^2\bI_p$, $\Theta := A\C$ and $\Theta^+ = (\Theta^\top \Theta)^{-1}\Theta^{\top}$, the asymptotic variance derived in \cite{bai2012}, in the regime  $K = O(1)$ and $\lambda_K \asymp p \rightarrow \infty$,  is 	
	$${Q_k = \left(\sigma^2 + \t^2\|\beta\|_2^2 \right)\left[\C^{-1}\right]_{kk}}.$$	 The assumption that  $
	\lambda_K \asymp p 
	$ is made in \cite{bai2012} indirectly, 
	as a consequence of their assumption (A)
	($\C$ is positive definite and $K$ is fixed) and of their   Assumption (C) ($\| A_{i\sbt}\|_2 \le C$, $c\le \Gamma_{ii} \le C$ and $p^{-1} A^\top \Gamma^{-1} A    $ converges to some positive definite matrix as $p\to\i$), see page 438 of \cite{bai2012}. 
	
	The asymptotic variance of our proposed estimator, valid for all the regimes represented in Table \ref{tab_infer} above, has the  formula  derived  in Theorem \ref{thm_distr} of Section \ref{sec_inference_beta},  	
	$${V_k = \left(\sigma^2 + {\t^2\over m}\|\beta\|_2^2 \right)\left[\left[\C^{-1}\right]_{kk} + \tau^2\e_k^\top \left(\Theta^\top \Theta\right)^{-1}\e_k\right]+ {\tau^4\over m(m-1)}\sum_{a=1}^K \beta_a^2 \sum_{i\in I_a}\left[\e_k^\top \Theta^+ \e_i\right]^2,}	$$
	where $\{\e_1, \ldots, \e_K\}$ is the canonical basis of $\R^K$.

	To contrast the two asymptotic variance expressions, we consider the common regime $K = O(1)$  and $p \rightarrow \infty$, in which case we note that  the signal strength requirement under which our Theorem \ref{thm_distr} is established reduces to  $\lambda_K \gtrsim (p/\sqrt{n}) \log \pn$. Theorem \ref{thm_distr} shows that, in this case, if we further assume that  $\lambda_K\to\i$, the asymptotic variance of our estimator reduces to 
	\[
	V_k = \left(\sigma^2 + {\t^2\over m}\|\beta\|_2^2 \right)\Omega_{kk},
	\] 
	and thus 	 $\lim_{n\to \i} Q_k / V_k \geq 1$.  The two asymptotic variances coincide when $m = 1$, which is the minimum  identifiability requirement for a factor model with pure variables in which the pure variable set $I$ is  known. Hence
	we recover, in this regime,  the asymptotic variance derived  in   \cite{bai2012},  while at the same time relaxing  their signal strength conditions required for this derivation.

	As noted above, and as we prove formally in Theorem \ref{thm_distr}, the expression of the asymptotic variance $V_k$ is valid in all the regimes presented in this table. In particular, in the classical regime in which $K$ and $p$ do not vary with $n$, its derivation requires only $\lambda_K \gtrsim n^{-1/2}$. \\

	In order to provide a unifying analysis, valid for both fixed and growing dimensions, we  use the classical Lyapunov CLT for triangular arrays. The verification of the third moment condition of this theorem requires the lengthy, technical, derivations  in Lemmas \ref{lem_third_moment} and \ref{lem_var_UI}.
	Finally, although the expression of the asymptotic variance $V_k$ is involved, it  can be estimated consistently for each $1\le k\le K$, by a computationally efficient estimator. This result is given in Proposition \ref{prop_est_V} of Section \ref{sec_inference_beta} and its proof, which  requires considerable attention,  is presented in Appendix \ref{sec_proofs_prop_V}, followed by a list of the many technical lemmas used in this proof,  Lemmas \ref{lem_tau} -- \ref{lem_1}.

	\paragraph{An application to regression on latent cluster centers}
	
	The identifiable factor model $X = AZ + W$ satisfying Assumption \ref{ass_model} in Section \ref{sec_model_assump} below can be used to define, uniquely,  overlapping clusters of the coordinates of $X$. The clusters are  centered around the components of the latent vector  $Z$,  and $X$-variables  in cluster $k$ have indices in the set $G_k := \{ j\in [p] :  | A_{jk}| > 0\},\ \text{for $ 1\le k \le K$}.$ A procedure for estimating consistently $K$ and the corresponding clusters has been developed recently  in \cite{LOVE}.  With this interpretation, the  Essential Regression framework can be employed for inference on the latent cluster centers. We show in Section \ref{sec_app_to_clustering} that although it may be tempting to  replace the components of $Z$ by weighted  averages of  variables within a cluster, and subsequently regress $Y$ onto them, this procedure would not estimate $\beta$ in  (\ref{model}). 
	However, we further show that this can be immediately corrected by regressing on the best linear predictor of $Z$ from these weighted averages. With this correction, we obtain 
	exactly the estimator  of $\beta$  constructed in Section \ref{sec_est_beta}, 
	and  the inferential tools developed in Section \ref{sec_inference_beta} can be used  for inference in regression on unobserved, latent, cluster centers. In the context of the applications to biological data sets mentioned in Section \ref{framew}, this will be inference at the biological signature level, as illustrated in Section \ref{sec_real_data}.

	\subsection{Organization of the paper}
	
	The rest of the paper is organized as follows. 
	
	Section \ref{sec_est} gives a set of  modeling assumptions under which the model  given by (\ref{model}) and (\ref{love}) is identifiable. Section \ref{sec_model_assump} introduces and discusses these assumptions, including parameter interpretability. Section \ref{sec_ident_beta} shows that our central parameter, $\beta$, along with other important parameters, is identifiable.  
	
	Section \ref{sec_est_beta} introduces our proposed estimator $\wh \beta$ of $\beta$. Section \ref{sec_other_estimates} discusses other natural estimators, and  explains why they should be expected to have inferior theoretical and practical  perfomance relative to $\wh \beta$. The  numerical  performance of these alternate estimators is presented in Section \ref{sec_sims}. 
	
	The  performance of estimators of $\beta$ in factor regression models satisfying  Assumption \ref{ass_model} is benchmarked in  Section \ref{sec_beta_lower}.  Theorem \ref{thm_beta_lower} provides the minimax lower bound for estimating $\beta$ in this class of models, with respect to the $\ell_2$ loss. 
	
	Section \ref{sec_consistency_beta} shows that the estimator $\wh \beta$ proposed in Section \ref{sec_est_beta} is $\ell_2$-norm  consistent, and  minimax-rate  adaptive, under assumptions collected in Section \ref{sec_ass}. 
	
	Section \ref{sec_inference_beta} is devoted to the component-wise asymptotic normality of $\wh \beta$ and to the estimation of the asymptotic variance, as well as to a comparison with existing literature. 
	
	Section \ref{sec_app_to_clustering} presents an application of the framework, methodology and theory developed in previous sections to regression on latent cluster centers, when the clusters are allowed to overlap. 
	
	Section \ref{sec_sims} verifies numerically our theoretical results.  Section \ref{sec_real_data} shows how our methodology can be used to make inference on unobserved, latent,  immunological modules, using a data set collected during a study on the effectiveness of a new SIV-type vaccine. 
	
	
	All  proofs are deferred to the supplement.
	Appendix \ref{sec_proofs_ident_lower} gives the proofs of Proposition \ref{prop_iden} and Theorem \ref{thm_beta_lower}, on identification and minimax lower bounds, respectively.
	Appendix \ref{sec_auxiliary_lemma} provides necessary preliminary results, and could be skipped at first reading.
	The proof of Theorem \ref{thm_beta} concerning the convergence rates of $\wh \beta$ is given in Appendix \ref{sec_proof_rate_beta}. The proof of Theorem  \ref{thm_distr} on  the asymptotic normality of $\wh \beta$ is given in Appendix \ref{sec_proof_asn_beta}, while Proposition \ref{prop_est_V} on  consistent estimation of the asymptotic variance $V_k$ is proved in Appendix \ref{sec_proofs_prop_V}.

	\subsection{Notation}
	For any positive integer $q$, we let $[q] = \{1, 2, \ldots, q\}$. For two numbers $a$ and $b$, we write $a\vee b := \max\{a, b\}$ and $a\wedge b:= \min\{a, b\}$.
	For a set $S$, we use $|S|$ to denote its cardinality. 
	We use $\H_d$ to denote the set of all $d\times d$ signed permutation matrices and $\S^{d-1}$ to represent the space of the unit vectors in $\R^d$. We denote by ${\bI}_d$ the $d\times d$ identity matrix, by $\bm{1}_d$ the $d$-dimensional vector with entries equal to $1$ and by $\{\e_j\}_{1\le j\le d}$ the canonical basis in $\R^d$. 
	For a generic vector $v$, we let $\|v\|_q = \left (\sum_i|v_i|^q\right)^{1/q}$ denote its $\ell_q$ norm for $1\le q< \i$. 
	We also write $\|v\|_\i = \max_i |v_i|$ and $\|v\|_0= |\textrm{supp}(v)|$. Let $Q$ be any matrix. We use $\|Q\|_{{\rm op}} = \sup_{v\in \S^{d-1}}\|Qv\|$ 
	and $\|Q\|_\i = \max_{i,j}|Q_{ij}|$ for its operator norm 
	and element-wise maximum norm, respectively. For a symmetric matrix $Q\in \R^{d\times d}$,  we denote by $\lambda_k(Q)$ its $k$th largest eigenvalue for $k\in[d]$. For a positive semi-definite symmetric matrix, we will frequently use the fact that $\lambda_1(Q) = \|Q\|_{{\rm op}}$. For an arbitrary real valued matrix $M$, we let $\sigma_K(M)$ denote its $k$th singular value (in decreasing order).

	For any two sequences $a_n$ and $b_n$, $a_n \lesssim b_n$ stands for there exists constant $C>0$ such that $a_n \le Cb_n$. We write $a_n \asymp b_n$ if $a_n \lesssim b_n$ and $b_n \lesssim a_n$. We also use $a_n = o(b_n)$ to denote $a_n / b_n \to 0$ as $n\to \i$.

	\section{Modeling assumptions and identifiability}\label{sec_est}
	
	\subsection{Modeling assumptions} \label{sec_model_assump}
	
	We begin by formalizing and explaining the  set of model identifiability assumptions that will be used in this work.

	\begin{ass}\label{ass_model}\mbox{}
		\begin{enumerate}
			\setlength\itemsep{0mm}
			\item[(A0)] $\|A_{j\sbt}\|_1 \le 1$ for all $j\in [p]$.
			\item[(A1)] For every $k\in [K]$, there exists at least two $j\ne \ell \in [p]$, such that $|A_{j\sbt}| = |A_{\ell \sbt}| = \e_k$.
			\item [(A2)] $\C := \Cov(Z)$ is positive definite. There exists a constant $\nu>0 $ such that 
			$$
			\min_{1\le a<b\le K} \Bigl(\left[\C\right]_{aa}\wedge \left[\C\right]_{bb} - \left|\left[\C\right]_{ab}\right| \Bigr)> \nu.
			$$  
		\end{enumerate}
	\end{ass}
	\noindent In \textit{(A1)}, the absolute value is taken entry-wise and we use $\{\e_1,\ldots, \e_K\}$ to denote the canonical basis in $\R^K$. 
	For future reference, we denote the index set corresponding to pure variables as 
	\begin{equation}\label{def_I}
	I = \bigcup_{k=1}^KI_k,\qquad I_k = \left\{i\in [p]: |A_{i\sbt}|=\e_k\right\}.
	\end{equation}
	Its complement set is called the non-pure variable set $J := [p] \setminus I$.

	Assumption \ref{ass_model} guarantees that $A$ and $\C$ are identifiable, up to 
	signed permutations \citep[Theorem 2]{LOVE}. 
	We refer to the  second assumption \textit{(A1)} as  {\it the pure 
		variable assumption}.  It states that   every $Z_k$, $1\le k \le K$, must have at least two components of $X$, the \textit{pure variables}, solely associated with it, up to additive noise with possibly different variance levels.
	An in-depth comparison with the rich literature on factor models of type (\ref{love}) and a detailed explanation of assumptions\textit{ (A0) -- (A2)} can be found in \cite{LOVE}, and thus we only offer a brief set of comments here.

	If  $A\sz A^{T}$ is identifiable,  we show in Corollary \ref{cor_ident} in Appendix \ref{app_cor}
	that $A$ and $\sz$ are identifiable up to signed permutations under Assumption \ref{ass_model}, but  when {\em (A1)} is relaxed to
	\begin{flalign*}
	\text{{\em (A1')} {\it For each $k\in [K]$, there exists at least one index  $i\in [p]$ such that $A_{i\cdot}  = \e_k$.}}
	\end{flalign*}
	However, the existing conditions under which $A\sz A^{T}$ can be identified from the decomposition $\Sigma = A\sz A^{T} + \Gamma$  can be incompatible with factor models with pure variables, for instance the incoherence condition in \cite{Chandrasekaran}, or can be  very stringent growth conditions on the eigenvalues of $A\sz A^{T}$. For an instance of the latter we refer to \cite{fan2013large,Bai-Ng-K} and  also to  Table \ref{tab_K} below. These difficulties can be bypassed under {\em (A1)} of Assumption \ref{ass_model}, using  the approach taken in  \cite{LOVE}, which does not rely on  separating out $A\sz A^{T}$ from $\Sigma$ in the  first step. 
	In Assumption \ref{ass_model}, {\em (A0)}, we set the scale to 1 to aid the interpretation of  matrix $A$ as a cluster membership matrix, and thus view the model $X = AZ + W$  as a latent clustering model, following \cite{LOVE}. The equality between the weights of  two pure variables, $|A_{j\cdot}| = |A_{\ell\cdot}| = \e_k$ has been relaxed in a recent work, \citep{bing2020detecting}, but  under a  slightly different scaling condition than {\it (A0)}, and we do not pursue that approach here. 
	
	Furthermore, we mention, for completeness, that a more rigid form of assumption 	{\em (A1')}, specifically 
	\[
	\textrm{{\em (A1'')} {\it For each $k\in [K]$, there exists a {\bf known} index $i\in [p]$ such that $A_{i\cdot}  = \e_k$,}}
	\]
	has had a long history,   as it is one of the few ``user-interpretable" parametrizations of $A$ that  eliminates the rotation ambiguity of the latent factors.
	In psychology, the ``pure" variables induced by the parametrization are called factorially simple items \cite{mcdonaldbook}. A similar condition on $A$ can be traced back to the econometrics literature, and an early reference is \cite{koopmans1950}, further discussed in \cite{anderson1956}, who called it ``zero elements (of $A$) in specified positions''. We refer to \cite{koopmans1950,anderson1956,thurstone} for more examples in psychology, sociology, etc. This parametrization is also called the errors-in-variable parametrization and has wide applications in  structural equation models, see, \cite{Joreskog1970,Joreskog1973}. The more recent review paper \cite{yalcin2001}, and  the references therein, provide  a nice overview of many other concrete applications that support interest in factor models under a parametrization of this type. We provide another example below.

	In the context of Assumption \ref{ass_model}, we interpret the entries in $A$ as (signed) mixture weights. Under model (\ref{model}), each $X_j$ is a signed mixture of $Z_1, \ldots, Z_K$, according to these weights. This assumption, which is sufficient for identifiability, is also a a desirable modelling assumption.

	As an illustration, assume that  $X \in \R^p$ contains gene-level measurements, and that  $Z \in \R^K$ corresponds to their biological functions. Then, \textit{(A0)} enables, 
	in this example, 
	to associate a gene with multiple biological functions, in different proportions per function. The inequality sign in \textit{(A0)} further allows some genes not to be associated with any of the functions captured by this model, thereby increasing the robustness of the model. The second requirement, \textit{(A1)},  simply says 
	that the measured $X_j$ and $X_\ell$ have  the same biological function $Z_k$, and only that function.  We considered signed mixtures to increase the flexibility of the model. In this example,  signs correspond to the nature of the function. For instance, if gene $X_j$ activates a signaling pathway, and $Z_k$ has positive sign, then $Z_k$ has a function associated with the activation of this pathway, whereas a negative sign indicates a function associated with its inhibition.

	
	
	Assumption \textit{(A2)}  allows us to depart from the widely used, and restrictive, assumption of independence among the  latent factors. We require the variability of the factors to be strictly larger than that between factors. This implies the minimal desideratum  that the factors  be different.

	We first discuss the identifiability of $\beta$
	in Section \ref{sec_ident_beta}. Then in Section \ref{sec_est_beta},  we propose our estimator of $\beta$ which uses its identifiability constructively.
	
	\subsection{Identifiability of $\beta$: a constructive approach}\label{sec_ident_beta}
	
	Under model (\ref{model}), we have  $Y = Z^\top \beta + \eps$, and thus the  coefficient $\beta$ satisfies
	\begin{equation}\label{def_beta}
	\beta = [\Cov(Z)]^{-1}\Cov(Z,Y) = \C^{-1}\Cov(Z,Y).
	\end{equation}
	Since model (\ref{love}) and Assumption \ref{ass_model} imply $\Cov(Z,Y) = (A^\top A)^{-1}A^\top \Cov(X,Y)$, we have 
	\begin{align}\label{iden_beta_full}
	\beta &~=~ 	\C^{-1}(A^\top A)^{-1}A^\top \Cov(X,Y)\\\label{iden_beta_I_full} 
	&~=~ (\Theta^\top \Theta)^{-1}\Theta^\top \Cov(X,Y)
	\end{align}
	with $\Theta = A\C$. 
	Therefore, $\beta$ is uniquely defined  whenever $\Theta$ is unique. By partitioning the $p\times K$ matrix $A$ as 
	$A_{I \sbt} \in \R^{|I|\times K}$ and $A_{J \sbt}\in \R^{|J|\times K}$ corresponding to $I$ and $J$, respectively, model (\ref{love}) and Assumption \ref{ass_model} imply the following decomposition of $\Sigma$, 
	\begin{equation*}
	\Sigma = \begin{bmatrix}
	\Sigma_{II} & \Sigma_{IJ} \\ \Sigma_{JI}  & \Sigma_{JJ} 
	\end{bmatrix} = \begin{bmatrix}
	A_{I \sbt}\C A_{I \sbt}^\top  & A_{I \sbt}\C A_{J \sbt}^\top \vspace{2mm}\\
	A_{J \sbt}\C A_{I \sbt}^\top  & A_{J \sbt}\C A_{J \sbt}^\top 
	\end{bmatrix} + \begin{bmatrix}
	\Gamma_{II}  & \\ & \Gamma_{JJ}
	\end{bmatrix}.
	\end{equation*}
	In particular, we have $\Sigma_{II} = A_{I \sbt}\C A_{I \sbt}^\top  + \Gamma_{II}$ and
	\begin{equation}\label{iden_AC}
	\Theta= A \C = \left(\Sigma_{\sbt I} - \Gamma_{\sbt I}\right)A_{I \sbt}^\top (A_{I \sbt}^\top A_{I \sbt})^{-1}.
	\end{equation}
	The uniqueness of $\Theta$ is thus implied by  that of $A_{I \sbt}$ and $\C$. 
	Theorem 1 in \cite{LOVE} shows that, under Assumption \ref{ass_model}, the matrices  $A_{I \sbt}$  and $\C$ can be uniquely determined, up to a signed permutation matrix $P$, from $\Sigma := \Cov(X)$. As a result, $\Theta$ can also be recovered from (\ref{iden_AC}) up to $P^\top $, hence $\beta$ is identifiable from (\ref{iden_beta_I_full}) up to $P^\top $. We remark that the permutation matrix $P$ will not affect either inference  or prediction. Indeed, writing $\wt A = AP$, $\wt Z = P^\top Z$ and $\wt \beta = P^\top \beta$, one still has $Y = \wt Z^\top \wt \beta+\eps$ and $X = \wt A\wt Z+W$. We summarize the identifiability of $\beta$ in the proposition below.  Its proof can be found in Appendix \ref{sec_proofs_prop_iden}.
	\begin{prop}\label{prop_iden}
		Under models (\ref{model}) -- (\ref{love}) and Assumption \ref{ass_model}, the quantities $\Sigma$ and  $\Cov(X, Y)$ define $\beta$ uniquely, via (\ref{iden_beta_I_full}) and (\ref{iden_AC}),  up to a signed permutation matrix.
	\end{prop}

	Our estimator, given  in the next section,  is based on the representations (\ref{iden_beta_I_full}) and (\ref{iden_AC}), followed by appropriate plug-in  estimators. 

	\section{Estimation of $\beta$}\label{sec_est_beta}
	We assume that the data consists of $n$ independent observations $(X_1,Y_1),\ldots,(X_n,Y_n)$ 
	that satisfy model (\ref{model}) and (\ref{love}). We write $\X:=(X_1,\ldots,X_n)^\top$ for the observed $n\times p$ data matrix and $\y:=(Y_1,\ldots,Y_n)^\top$ for the observed response vector.
	Let $\wh \Sigma = n^{-1}\sum_{i = 1}^n X_i X_i^\top
	$ denote the sample covariance matrix. Motivated by equations (\ref{iden_beta_I_full}) and (\ref{iden_AC}), we consider the plug-in estimator of $\beta$ via the following steps:
	\begin{enumerate}
		\item[(1)] Obtain estimates $\wh K$ and  $\{\wh I_1,\cdots, \wh I_{\wh K}\}$ from $\wh \Sigma$ with tuning parameter $\delta$ by using Algorithm 1 in \cite{LOVE}. For the reader's convenience, we state the procedure in Algorithm \ref{alg1} below. 
		
		\item[(2)] 
		Next, for each $a\in[\wh K]$ and $b\in [\wh K]\setminus\{a\}$, we compute 
		\begin{equation}\label{Chat}
		[\whC]_{aa} = 
		\frac{1}{|\wh I_a|(|\wh I_a|-1)}\sum_{i\ne j\in \wh I_a}|\wh \Sigma_{ij}|,
		\quad 
		~ [\whC]_{ab} = 
		\frac{1}{|\wh I_a||\wh I_b|}\sum_{i\in \wh I_a, j\in \wh I_b} \wh A_{ia}\wh A_{ib}\wh \Sigma_{ij},
		\end{equation}
		to form the estimator  $\whC$ of $\C$. Furthermore, the estimation of $A_{I \sbt}$ follows the procedure in \cite{LOVE}. For each $k\in [\wh K]$ and the estimated pure variable set $\wh I_k$,  
		\begin{align}\label{est_AI_a}
		&\textrm{Pick an element $i\in \wh I_k$ at random, and set $\wh A_{i\sbt}=\e_k$;}\\\label{est_AI_b}
		&\textrm{For the remaining $j\in \wh I_k\setminus\{i\}$, set $\wh A_{j\sbt} = \textrm{sign}(\wh \Sigma_{ij})\cdot \e_k$.}
		\end{align}
		
		\item[(3)] Estimate $\Gamma_{\sbt I}$ by $\wh \Gamma_{{\sbt \wh I}}$ with 
		\begin{equation}\label{est_Gamma_I}
		\wh \Gamma_{ii} = \wh\Sigma_{ii} - \wh A_{i\sbt}^\top \whC \wh A_{i\sbt},\quad \forall\ i\in \wh I,\qquad \wh \Gamma_{ji} = 0,\quad \forall\ j\ne i.
		\end{equation}
		\item[(4)] Compute
		\begin{equation}\label{est_Theta}
		\wh \Theta = \left(
		\wh\Sigma_{\sbt \wh I} - \wh \Gamma_{{\sbt \wh I}}
		\right)\wh A_{{\wh I\sbt}}\left(\wh A_{{\wh I\sbt}}^\top \wh A_{{\wh I\sbt}}\right)^{-1}.
		\end{equation}
		Provided that $\wh \Theta^\top \wh \Theta$ is non-singular, estimate $\beta$ by 
		\begin{equation}\label{est_beta}
		\wh \beta = \left(\wh \Theta^\top  \wh \Theta\right)^{-1}\wh \Theta^\top  {1\over n}\X^\top \y.
		\end{equation}
	\end{enumerate}
	{\begin{algorithm}[ht]
			\caption{Estimate the partition of the pure variables $\I$ by $\wh \I$}\label{alg1}
			\begin{algorithmic}[1]
				\Procedure {PureVar}{$\wh \Sigma$, $\delta$}
				\State $\wh \I \gets \emptyset$.
				\ForAll {$i\in [p]$} 
				\State $\wh I^{(i)} \gets \bigl\{l\in [p]\setminus\{i\}: \max_{j\in [p]\setminus\{i\}}|\wh \Sigma_{ij}| \le |\wh \Sigma_{il}|+2\delta\bigr\}$
				\State $Pure(i) \gets True$.
				\ForAll {$j \in \wh I^{(i)}$}
				\If {$\bigl||\wh \Sigma_{ij}|- \max_{k\in [p]\setminus\{j\}}|\wh\Sigma_{jk}|\bigr| > 2\delta$}   
				\State $Pure(i) \gets False$,
				\State \textbf{break}
				\EndIf	
				\EndFor
				\If {$Pure(i)$}
				\State $\wh I^{(i)} \gets \wh I^{(i)}\cup \{i\}$
				\State $\wh\I \gets$ \textsc{Merge($\wh I^{(i)},\ \wh \I$)}
				\EndIf
				\EndFor
				\State \Return $\wh \I$ and $\wh K$ as the number of sets in $\wh \I$
				\EndProcedure
				\Statex
				\Function {Merge}{$\wh I^{(i)}$, $\wh \I$}\label{alg2}
				\ForAll {$G \in \wh \I$}
				\Comment $\wh \I$ is a collection of sets
				\If {$G \cap \wh I^{(i)}\ne \emptyset$} 
				\State  $G\gets G\cap \wh I^{(i)}$
				\Comment Replace $G\in \wh \I$ by $G\cap \wh I^{(i)}$
				\State \Return $\wh \I$
				\EndIf
				\EndFor
				\State $\wh I^{(i)} \in \wh \I$
				\Comment add $\wh I^{(i)}$  in $\wh \I$
				\State \Return $\wh \I$
				\EndFunction
			\end{algorithmic}
	\end{algorithm}}
	
	The above procedure requires a single tuning parameter $\delta$, and that $K<n$. The theoretical order of $\delta$ is given in (\ref{def_delta}) of Section \ref{sec_ass} under the sub-Gaussian distributional assumptions. A practical data-driven procedure of selecting $\delta$ is stated in Appendix \ref{app_cv}.
	We   show in Section \ref{sec_app_to_clustering} that $\wh{\beta}$ coincides with the ordinary least squares estimator  that minimizes $\|\y - \wh \Z\beta \|_2^2$ over $\beta$, based on  an appropriately constructed  predictor $\wh \Z$ of the latent data matrix $\Z:=(Z_1,\ldots,Z_n)^\top$.   We prove in Theorem \ref{thm_beta} of Section \ref{sec_consistency_beta} that $\wh \Theta^\top  \wh \Theta$ is non-singular with high probability. In practice, in case that $\wh\Theta^\top\wh\Theta$ is singular or ill-conditioned, we propose to invert $\wh\Theta^\top\wh\Theta + t\cdot {\bf I}_{\wh K}$ instead, for any small $t>0$.  
	
	In Section \ref{sec_other_estimates} 
	we discuss other possible estimators based on the alternative representations of $\beta$ given in 
	(\ref{def_beta}), (\ref{iden_beta_full}) and (\ref{iden_beta_I_full}).

	\section{Statistical guarantees}

	\subsection{Minimax lower bounds for estimators of $\beta$ in Essential Regression}\label{sec_beta_lower}
	To benchmark our estimator of $\beta$, we derive the minimax optimal rate of $\|\wh\beta -\beta\|_2$ over the parameter space 
	$(\beta, \C, A) \in \S(R, m)$ with 
	\begin{align*}
	\S(R, m) :=\big\{ (\beta, \C, A): \ 
	&\|\beta\|_2\le R,\ \cl \le \lambda_{\min}(\C)\le \lambda_{\max}(\C)\le \cu,\\
	&A \text{ satisfies Assumption \ref{ass_model} with }\min_k|I_k| = m
	\big\},
	\end{align*}
	where $I_k$ is defined in (\ref{def_I}). 
	For the purpose of the minimax result, it suffices to consider the joint distribution of $(X, Y)$ as 
	\begin{equation}\label{distr_normal}
	\begin{bmatrix}
	X \\ Y 
	\end{bmatrix} \sim  N_{p+1}\left(0, \begin{bmatrix}
	A\C A^\top +\tau^2{\bI}_p & A\C\beta\\
	\beta^\top \C A^\top  & \beta^\top \C \beta + \sigma^2
	\end{bmatrix}\right)
	\end{equation}
	for $(\beta, \C, A)\in \S(R, m)$ and some positive constants $\sigma^2$ and $\tau^2$.
	\begin{thm}\label{thm_beta_lower}
		Let $K\le \bar{c}(R^2\vee m)n$ for some positive constant $\bar{c}$. 
		Let $(X_1, Y_1),\ldots,(X_n,Y_n)$ be i.i.d. random variables from the normal distribution in (\ref{distr_normal}).
		Then, there exist constants $c'>0$, $c''\in (0,1]$ depending only on $\bar{c}$, $\cu$, $\cl$, $\sigma^2$ and $\tau^2$, such that 
		\begin{equation}\label{ondergrens}
		\inf_{\wh\beta}\sup_{(\beta, \C, A) \in \S(R, m)}\PP\left\{\|\wh \beta - \beta\|_2 \ge c'\left(1\vee {R\over \sqrt m}\right)\cdot \sqrt {K\over n}\right\} \ge c''.
		\end{equation}
		The inf is taken over all estimators $\wh \beta$ of $\beta$.
	\end{thm}
	\begin{proof}
		The proof is deferred to Appendix \ref{sec_proofs_thm_beta_lower}. It uses the classical technique in \cite{Intro_non_para} for proving minimax lower bounds. After carefully constructing a set of ``hypotheses'' of $\beta$, we observe that the Kullback-Leibler (KL) divergence between joint distributions of $(X, Y)$ parametrized by two hypotheses of $\beta$ can be calculated from the log ratio of corresponding conditional densities of $Y|X$. This observation greatly simplifies the proof. 
	\end{proof}
	
	The factor $\sqrt{K/n}$ in (\ref{ondergrens}) 
	is the standard minimax rate of estimation in linear regression with observed $Z$ and sub-Gaussian errors. The factor multiplying it can be viewed as the price to pay for not observing $Z$. It quantifies the trade-off between not observing $Z$, with strength $\|\beta\|_2$, and the number of times, given by $m=\min_k|I_k|$, each component of $Z$ is partially observed, up to additive error. The ratio $\|\beta\|_2/\sqrt{m}$ indicates that, under the Essential Regression framework, the fact that $Z$ is not observed can be alleviated by the existence of pure variables, and the quality of estimation is expected to increase as $m$ increases. Theorem \ref{thm_beta_lower} above shows that, from the point of view of estimating $\beta$ consistently in Essential Regression, the number of factors $K$ can grow with $n$, a scenario not treated in the more classical factor regression  literature.
	It also reveals that, as in the classical regression set-up, although $K$ can grow with $n$ in Essential Regression,  consistent estimation  of unstructured $\beta$ cannot be guaranteed when $K > n$.  
	This will be treated in follow-up work.
	
	To the best of our knowledge, the minimax lower bound established above  is a new result in the factor regression model literature and it is interesting to place our results in a broader, related, context. For this, note that under the Essential Regression framework, {\it  if  $I$ and $A_{I \sbt}$ were known}, the pure variable assumption implies
	\begin{align}\label{bel}
	Y = Z^\top  \beta + \eps,\qquad 
	\bar X_I = Z + \bar W_I
	\end{align} 
	with $\bar X_I := (A_{I \sbt}^\top A_{I \sbt})^{-1}A_{I \sbt}^\top X_{I}$ and $\bar W_I :=(A_{I \sbt}^\top A_{I \sbt})^{-1}A_{I \sbt}^\top W_I$. Model (\ref{bel}) becomes an instance of an errors in variables model where the covariance structure of the error term $\bar W_I$ is diagonal.
	The minimax optimal  lower bound for estimating  $\beta$ in such models has been derived recently in  \cite{ Belloni17}, under sparsity assumptions on $\beta$. In the particular case of non-sparse $\beta$, which  we treat here, their lower bound agrees with that given by our Theorem \ref{thm_beta_lower}, although their bound is derived over a larger class, and can only be compared with (\ref{ondergrens}) when $I$ is known. The closest result to that of Theorem \ref{thm_beta_lower} is the minimax lower bound on rows of $A$  bounded in $\ell_1$ norm, and has been derived in \cite{LOVE}. We complement this here, in the latent factor regression context, by providing a minimax lower bound on $\beta$ with a  $\ell_2$ norm allowed to increase with $n$.  

	\subsection{Assumptions}\label{sec_ass}
	
	In this section we collect the assumptions under which we evaluate the performance of our proposed estimator $\wh \beta$.
	
	We first  make	the following distributional specifications  for $\eps$, $W$ and $Z$ defined in model (\ref{model}):
	\begin{ass}\label{ass_subg}
		Let $\gamma_{\eps}, \w, \z$ and $B_z$ be positive finite constants.
		Assume $\eps$ is $\gamma_{\eps}$-sub-Gaussian\footnotemark\footnotetext{A mean zero random variable $x$ is called $\gamma$-sub-Gaussian if $\EE[\exp(tx)]\le \exp(t^2\gamma^2/2)$ for all $t\in \R$.} and $W$ has independent $\w$-sub-Gaussian entries. Further assume $\|\C\|_\i \le B_z$ and the random vector $\C^{-{1/2}}Z$ is $\z$-sub-Gaussian\footnotemark\footnotetext{A mean zero random vector $x$ is called $\gamma$-sub-Gaussian if $v^\top x$ is $\gamma$-sub-Gaussian for any $\|v\|_2=1$.}.	\end{ass}
	
	The quality of our estimator $\wh \beta$ given by (\ref{est_beta}) depends on how well we estimate $K$, $I$, its partition $\{I_k\}_ {1\leq k \leq K}$, as well as $\Theta$. Our goal is to estimate $K$ consistently, under minimal assumptions. However, consistent estimation of  the partition  requires a stronger set of assumptions that we would like to avoid. We introduce and discuss below a set of assumptions under which the partition is recovered sufficiently well for the purpose of inference on $\wh \beta$.
	
	Assumption \ref{ass_subg} implies that $X_j$ is $\gamma_x$-sub-Gaussian with $\gamma_x=(\z\sqrt{B_z}+\w)$, as shown in Lemma \ref{lem_Z} in Appendix \ref{sec_proof_prelim_lemmas}, and it is well known (see, for instance, \cite[Lemma 1]{bien2016convex}) that in this case,
	\begin{equation}\label{def_delta}
	\PP\left\{
	\max_{1 \le j< \ell \le p} |\wh \Sigma_{j\ell} - \Sigma_{j\ell}| \le \delta
	\right\} \ge 1-(p\vee n)^{-c'}
	\end{equation}
	with $\delta = c\sqrt{\log (p\vee n)/n}$, for some constant $c'>0$ and $c = c(\gamma_x)>0$ sufficiently large.
	
	Under Assumptions \ref{ass_model} and \ref{ass_subg}, and when $\log p  \le c'' n$ for some constant $c''>0$, \cite{LOVE} provides an algorithm for estimating $K$ and $I$,  and prove in their Theorem 3 the following: 
	\begin{enumerate}\label{thm_I}
		\item[(1)] $\wh K = K$;
		\item[(2)] $I_k \subseteq \wh I_{\pi(k)} \subseteq I_k \cup J_1^k$, for all $k\in [K]$, 
	\end{enumerate}
	where  $\pi: [K] \rightarrow [K]$ is a permutation and $J_1^k := \{ j\in J: |A_{jk}| \ge 1-4\delta/\nu \}$ with  constant $\nu$ defined in Assumption \ref{ass_model} of Section \ref{sec_model_assump} above.
	
	Since we do not impose any separation condition between the pure variable rows $A_{I \sbt}$ and the remaining rows in $A_{J \sbt}$, the sets $\{J_1^k\}_{k=1}^K$ are typically not empty, and as formalized in (2) above, we cannot expect to recover $I$ perfectly in the presence of \emph{quasi-pure} variables with indices belonging  to  the set $J_1\coloneqq \bigcup_{k=1}^K J_1^k$. Indeed, when $\log p = o(n)$, for any $j\in J_1^k$ we have $|A_{jk}| \approx 1$ and $A_{jk'}\approx 0$ for any $k'\ne k$, so variables corresponding to $J_1^k$ are very close to the pure variables in $I_k$,   and possibly indistinguishable from one another, in finite samples. 
	
	While allowing for $J_1 \neq \emptyset$ increases the flexibility of the model, it also  poses significant technical difficulties in the analysis of the asymptotic distribution of $\wh \beta_k$, for each $k$, evidenced in the proofs of Theorem \ref{thm_beta}, Theorem \ref{thm_distr} and Proposition \ref{prop_est_V}. Nevertheless, the limiting distribution can still be derived when  $|J_1|$ is small relative to $|I|$. 
	The influence of the misclassified $X$-variables with entries in  $J_1$ becomes 
	negligible in  both the finite sample rate analysis of $\wh \beta$ provided in Section \ref{sec_consistency_beta} and the  asymptotic analysis  of Section \ref{sec_inference_beta} under the condition given below. We introduce  
	\begin{equation}\label{def_rho}
	\bar\rho^2=  
	\sum_{k=1}^K  \left(\frac{|J_1^k|}{|I_k| + |J_1^k|}\right)^2
	\end{equation}
	to quantify the influence of quasi-pure variables on the quality of our estimation.  Theorem \ref{thm_beta} shows that optimal estimation of $\beta$ is possible in the presence of quasi-pure variables as long as their number is negligible relative to the number of  pure variables in the same group, in that the following assumption holds.

	\begin{ass}{3}\label{ass_J1}
		The overall proportion	$\bar{\rho}^2$ satisfies   $m \bar{\rho}^2 =O(1)$ with $m:=\min_{k\in[K]} |I_k|$.
	\end{ass}
	
	We briefly discuss  this assumption below. 
	Let $s\le K$ be the number of factors that have quasi-pure variables, that is,
	$$	        s = |S|, \quad S = \{1\le k\le K: |J_1^k| \ge 1\}.
	$$
	
	\begin{enumerate}
		\item  Note that, by (\ref{def_rho}),  we  have  $\bar \rho^2  \le s$, and therefore  $m\bar \rho^2  \le ms$. Thus,  if both $s$ (and in particular $K$, when $s = K$) and  $m$ remain bounded, as $n\to\i$, Assumption \ref{ass_J1} holds. 
		In other words, in factor models with a possibly large, but fixed,  number of factors $K$, such that one of the factors has very few $X$-variables solely associated with it, the assumption holds.

		\item    Otherwise, if $m s\to\i$, we note that if  we assume that 
		\begin{eqnarray}\label{restr}
		|J_1^k| \le c( | J_1^k| + | I_k|) /\sqrt{ms},\quad \end{eqnarray}
		holds  for all $k\in S$
		where $c>0$ is some universal constant,  then 
		$\bar \rho^2 \le c^2/m$,  and Assumption \ref{ass_J1} holds. 
		To gain insight into  (\ref{restr}),  note that it also  implies that,
		for each cluster $k\in S$, 
		$| J_1^k|  = o (|I_k|)$,
		and therefore  $| J_1|= o(|I|) $. 
		Thus,  in factor models with a growing number of factors and a growing number of pure variables per factor,  (\ref{restr})  prevents 
		the   number of quasi-pure variables  to  grow  faster than the number of pure variables.
		
		\item  In support of the above discussion, we also offer a calculation of  $\bar \rho^2$ in a  particular case.  Assume that  each $I_k$ has the same size $m$, and each $J_1^k$ has the same size $m'$. Then $\bar \rho^2 = s (m'/(m+m'))^2 $, and we can verify Assumption 3 explicitly in terms of $m, m'$ and $s$.  Consider
		\begin{enumerate}
			\item  $|I_k|=m$ with $m\ge 1$, $m$ fixed, and $|J_1^k|=m'\ge 1$, for all $k\in S$. Then $\bar \rho^2 = s (m'/(m+m'))^2$, and Assumption 3 is   met  if and only if $s$ remains bounded. 
			
			
			\item $|I_k|= m$ and $| J_1^k|= m'= O( m^\alpha)$,   for all $k\in S$, with $0\le \alpha\le 1$ and $m\to \i$. A simple calculation shows that  Assumption 3 is  met only if $s=O(m^{1-2\alpha})$. In particular,  $\alpha<1/2$ allows $s\to\i$;    $\alpha=1/2$   requires $s$ to be bounded, and the case $\alpha>1/2$ forces $s=0$ (all the stated limits are in terms of  $n\to\i$).\\
		\end{enumerate} 
	\end{enumerate}

	Another quantity that needs to be controlled is the covariance matrix $\C$. It plays the same role as the Gram matrix in  classical linear regression with random design.
	
	\begin{ass}\label{ass_C}
		The smallest eigenvalue $\lambda_{\min}(\C)  > C_{\min}$ for some constant $C_{\min}$ bounded away from $0$.
	\end{ass}
	
	Assumptions \ref{ass_subg} -- \ref{ass_C}   allow for a cleaner presentation of our results.
	We can  trace  explicitly the  dependency of the estimation  rate  for $\beta$ on $\bar\rho$ and $\cl$ in the proofs. 
	An important feature of this framework is that under Assumptions \ref{ass_model} -- \ref{ass_C}
	and   
	$	K = O\left( n / \log \pn \right),$
	the matrix
	$\wh\Theta^\top \wh\Theta$   
	can be inverted,   with probability tending to 1.
	
	We require one more  condition, 
	which measures the strength of the signal $A\C A^\top$ retained by the low rank approximation of $\Sigma=A \C A ^\top +\Gamma$.

	\begin{ass}\label{ass_H}
		The $K$th eigenvalue $\lambda_K := \lambda_K(A\C A^\top )$ of the signal $A\C A^\top$ satisfies
		\begin{equation}\label{cond_2}
		\lambda_K \geq c \  p\sqrt{\log\pn \over n}
		\end{equation} 
		for some sufficiently small constant $c>0$.
		This implies $K\lesssim \sqrt{n/\log\pn}$.
	\end{ass}

	The implication $K\lesssim \sqrt{n/\log\pn}$ follows immediately from (\ref{cond_2}) and
	the bound $\lambda_K\le B_z(p/K)$ in
	Lemma \ref{lem_signal}.

	We recall that the quantity  $\lambda_K :=\lambda_K(A \C A^{\top})$  quantifies  the size of the signal in $X=AZ+W$. It is a key quantity in the well studied problem of signal recovery from a  $n \times p$ matrix of noisy observations ${\bf X}$, 
	with rows corresponding to $n$ i.i.d. copies of $X$. The signal can be recovered from $n^{-1}\X^\top\X  $
	as soon as $\lambda_K$ is above the noise level 
	\citep{bunea2011,bunea2015,giraud-book,wegkamp2016,wainwright_2019}.
	The noise level is quantified by the largest eigen-value $\lambda_{1}(n^{-1} \W^\top\W )$, based on the $n \times p$ data matrix  
	${\W}$ with rows corresponding to $n$ i.i.d. copies of $W$.
	Standard random matrix theory shows that
	$\lambda_1 (n^{-1} \W^\top\W )$  concentrates with overwhelming probability  around its mean, which is of order $(n+p) /n$, see \cite{vershynin_2012}.
	Therefore, 
	one  needs at least $\lambda_K \gtrsim (n+ p) /n$
	to distinguish the signal from the noise. For the more specific task of optimal estimation of $\beta$,  we require Assumption \ref{ass_H} (see Lemma \ref{lem_Rem2} in Appendix \ref{sec_proof_rate_beta}). The investigation of its optimality is beyond the scope of this paper, and will be studied in a follow up work. However, we emphasize that, as mentioned in the Introduction, the consistent estimation of rows of the factor loadings in factor models, especially when  $p > n$, has only been established under the  stricter condition  $\lambda_K \asymp p$ \citep{Bai-factor-model-03, fan2011, fan2013large, bai2012}. The intuition behind this more restrictive  assumption is as follows. If $\C$ is positive definite, with finite eigenvalues, and the rows of $A$ are $p$ i.i.d. draws of a $K$-dimensional sub-Gaussian random vector, $p > K$, then reasoning as above and using the results in \cite{vershynin_2012},   $ \lambda_K \asymp p$, with high probability. However, for a generic, deterministic,  $A$, it would be an assumption, that we show can be considerably relaxed to Assumption \ref{ass_H}.   More details are provided  in Remark \ref{rem_inference_results} of Section \ref{sec_inference_beta}.

	\subsection{Consistency of $\wh{\beta}$ in $\ell_2$-norm: Rates of convergence and optimality} \label{sec_consistency_beta}
	
	The following theorem states the convergence rate of $\min_{P\in \H_K}\|\wh\beta-P\beta\|_2$. 
	\begin{thm}\label{thm_beta}
		Suppose 
		Assumptions \ref{ass_model} -- \ref{ass_C} hold and assume $K\log\pn \le cn$ for some sufficiently small constant $c>0$. 
		Then,
		with probability greater than $1-\pn^{-c'}$ for some constant $c'>0$, $\wh K=K$,
		the matrix	$\wh \Theta^\top  \wh \Theta$ is non-singular and  the estimator
		$\wh \beta$ given by (\ref{est_beta}) satisfies:
		\begin{equation}\label{rate_init_beta}
		\min_{P\in \H_K}\left\|\wh \beta - P\beta\right\|_2 \lesssim   \left(1\vee {\|\beta\|_2\over \sqrt m}\right)\sqrt{K\log \pn \over n}
		\left(1\vee {p\over \lambda_K}\sqrt{\log\pn \over n}\right)
		\end{equation}
		If additionally Assumption \ref{ass_H} holds, then with the same probability, $\wh \beta$ given by (\ref{est_beta}) satisfies
		\begin{equation}\label{rate_beta}
		\min_{P\in \H_K}\left\|\wh \beta - P\beta\right\|_2 \lesssim   \left(1\vee {\|\beta\|_2\over \sqrt m}\right)\sqrt{K\log \pn \over n}.
		\end{equation}
	\end{thm}
	\begin{proof}
		The proof is given in Appendix \ref{sec_proof_rate_beta}.
	\end{proof}
	
	\begin{remark}{\rm
			The estimator $\wh\beta$ achieves the minimax rate in Theorem \ref{thm_beta_lower} up to a logarithmic  $\log\pn$ term. 
			Inspection of the proof, when $K$ grows with $n$, shows  that the $\log\pn$ terms appearing in the condition $K\log\pn \le cn$ and in the upper bound 
			(\ref{rate_beta}) can be improved to $\log K$, but in that case the probability tail in Theorem \ref{thm_beta} will change to $1-K^{-c}$. This additional $\log K$ is the price to pay for not observing $Z$. On the other hand, comparing (\ref{rate_beta}) with the convergence rate of the oracle least squares estimator (OLS) {\em when $Z$ is observable}, the extra factor  $\|\beta\|_2/\sqrt{m}$ is due to the   error term $W$ in $X=AZ+W$. This extra factor becomes negligible
			when  the coefficients of $\beta$ are uniformly bounded ($\|\beta\|_\i \lesssim 1$) and 
			the number of latent factors cannot grow much faster than the cardinality of the smallest subgroup   of pure variables 
			($K\lesssim m$). In the worst case scenario the rate of $\wh\beta$ is slower than the aforementioned OLS by a factor of the order  of $\|\beta\|_2$. }
	\end{remark}
	
	\begin{remark}\label{rem_K}{\rm 
			The selection of $K$, the number of latent factors, has been thoroughly studied, in general factor models.  For completeness, we provide Table \ref{tab_K} below  summarizing the existing approaches of selecting $K$,  as well as the conditions under which the resulting estimates are consistent.   As the table shows, consistent estimation of $K$  via the existing methods is proved under the assumption that there exists a large  gap between  the eigenvalues of $A\sz A^T$ and $\Gamma$, respectively.  Although in the  inference results derived in this manuscript  we also need
			$\lambda_K = \lambda_K(A\sz A^T)$ to satisfy Assumption \ref{ass_J1}, that is,
			\[
			{\lambda_K \over p} ~ \gtrsim  ~ {\log (p\vee n) \over \sqrt n},
			\]
			this is {\em not} used to guarantee the correct  selection of  $K$ (this can be readily seen from the discussion in the middle of page 15), hence it is much milder than the conditions in the existing literature  \cite{Ahn-2013,Bai-Ng-K,Onatski09}, 
			as seen from the table.  We do, however, also establish that $K$ can be consistently estimated by our procedure. Leveraging  the particular class of factor models treated in this work, we show that our method is consistent, for  both growing $K$  and fixed $K$, but under  a  set of conditions that is very different than those previously considered  (please see Table \ref{tab_K} below). In particular, we do not require $\lambda_1$ and $\lambda_K$ to be of equal order $p$, but we do consider only diagonal error structures $\Gamma$.  
			\begin{table}[H]
				\centering
				{\renewcommand{\arraystretch}{2}{
						\resizebox{\textwidth}{!}{
							\begin{tabular}{|c|c|c|}
								\hline 
								$\wh K = K$ w.h.p.   & $K$ is fixed & $K$ grows with $n$\\\hline
								\multirow{1}{*}{\thead{Existing literature: \cite{Ahn-2013},\\
										\cite{Bai-Ng-K,Onatski09}}
									} &  $\lambda_1\asymp \lambda_K \asymp  p$, $\|\Gamma\|_{\i,1}=O(1)$ & \multirow{1}{*}{NA}\\\cline{2-2}
								\hline 
								Proposed method  &   \multicolumn{2}{c|}{Assumption 1, $\Gamma$ is diagonal with bounded entries, $\log p = o(n)$}\\
								\hline 
							\end{tabular}
				}}}
				\caption{\small
					Selection of the number of factors $K$: methods and sufficient conditions for consistency.  We write $\lambda_1,\ldots, \lambda_K$ for the  top $K$  eigenvalues of $A\sz A^T$.}
				\label{tab_K} 
			\end{table}
		}
	\end{remark}

	\subsection{Other possible estimators}
	\label{sec_other_estimates}
	We discuss other natural estimators that could be considered in this model,  based on equivalent representations of $\beta$. Each of these representations offer  a valid basis for estimating $\beta$ via plug-in estimation of the unknown quantities.   However, we recommend the  estimator $\wh\beta$ in (\ref{est_beta}) above, as it  has several advantages over these other candidates,  both theoretically and numerically.

	\begin{enumerate}
		\item[(a)] Recall that $\beta$ can be uniquely defined via identity  (\ref{iden_beta_full}) as long as $A$ and $\C$ are unique  up to signed permutations.
		The  expression (\ref{iden_beta_full}) suggests the following estimator 
		\begin{equation}\label{est_beta_full}
		\wt\beta^{(A)} = \whC^{-1}\left(\wh A^\top \wh A\right)^{-1}\wh A^\top {1\over n}\X^\top   \y
		\end{equation}
		based on  some estimates $\wh A$ of $A$ and $\whC$ of $\C$. For instance, one can estimate $A$ and $\C$ by the procedure given in (\ref{Chat}), (\ref{est_AI_a}), (\ref{est_AI_b}) and (\ref{est_AJ}),  as these estimates have optimal convergence rates \citep{LOVE}.

		\item[(b)] Building on identity (\ref{iden_beta_full}),  using $\Sigma - \Gamma = A\C A^\top $ and writing $B = A(A^\top A)^{-1}$, we can show that
		\begin{equation}\label{iden_beta_mm}
		\beta = \left[
		B^\top (\Sigma - \Gamma) B
		\right]^{-1}B^\top \Cov(X,Y).
		\end{equation}
		If $A$ and $\C$ are unique up to signed permutations, then so is $\Gamma = \Sigma - A\C A^\top $.   Equipped with $\wh A$, one can further estimate $\Gamma$ via (\ref{est_Gamma}) and estimate $B$ by $\wh B = \wh A (\wh A^\top \wh A)^{-1}$.  Then expression (\ref{iden_beta_mm})  provides another way of estimating $\beta$ via
		\begin{equation}\label{est_beta_mm}
		\wh\beta^{(A)} = \left[
		\wh B^\top (\wh\Sigma - \wh\Gamma) \wh B
		\right]^{-1}\wh B^\top {1\over n}\X^\top \y.
		\end{equation}
	\end{enumerate}
	The two estimates are the same, $\wt \beta^{(A)} = \wh\beta^{(A)}$, if  $\wh\Gamma = \wh\Sigma - \wh A\whC \wh A^\top $ in (\ref{est_beta_mm}). Since $\wh\beta^{(A)}$ uses the diagonal structure of $\Gamma$, it is expected to have better performance than $\wt \beta_{(A)}$.  However, both (\ref{est_beta_full}) and (\ref{est_beta_mm})  
	require  {\it separate} estimation of  $A$ and $\C$ by   $\wh A$ and $\whC$, respectively.  In contrast, our estimator $\wh\beta$ in (\ref{est_beta}) estimates $\Theta = A\C$ {\it as a whole}, leading to better rate performance as evidenced  in the simulation section. Furthermore, as we mentioned in Section \ref{sec_est_beta}, the proposed $\wh \beta$  has a simple interpretation, as a  ordinary least squares estimator, relative to appropriately constructed predictors of $Z$,  which makes it more appealing in practice.  We give the details in Section \ref{sec_app_to_clustering}.
	
	\begin{enumerate}
		\item[(c)] Recall that $X_I = A_{I \sbt}Z + W_I$ and $A_{I \sbt}$ has full rank under model  (\ref{love}) and Assumption \ref{ass_model}.  By using (\ref{def_beta}) and $\Cov(Z,Y) = (A_{I \sbt}^\top A_{I \sbt})^{-1}A_{I \sbt}^\top \Cov(X_I, Y)$, we find the identity 
		\begin{equation}\label{iden_beta_I}
		\beta = \C^{-1}(A_{I \sbt}^\top A_{I \sbt})^{-1}A_{I \sbt}^\top \text{Cov}(X_I, Y) = \left(\C A_{I \sbt}^\top A_{I \sbt}\C\right)^{-1} \C A_{I \sbt}^\top \text{Cov}(X_I, Y).
		\end{equation}
		This expression of $\beta$ relies  only on $A_{I \sbt}\C$ rather than the full matrix $\Theta = A\C$ as in (\ref{iden_beta_full}).  This is yet a different way of estimating  $\beta$ and
		identity (\ref{iden_beta_I}) suggests the following estimator of $\beta$
		\begin{equation}\label{est_beta_I}
		\wh\beta^{(I)} = \whC^{-1}\left(\wh A_{{\wh I\sbt}}^\top \wh A_{{\wh I\sbt}} \right)^{-1}\wh A_{{\wh I\sbt}}^\top {1\over n}\X^\top   \y.
		\end{equation}
	\end{enumerate}
	In Appendix \ref{sec_thm_beta_I}, we  state, without its lengthy proof due to space restrictions,
	that $\wh\beta^{(I)}$ has the same convergence rate as (\ref{rate_beta}) under Assumptions \ref{ass_model} -- \ref{ass_C}.
	However, we still recommend $\wh\beta$ over  $\wh\beta^{(I)}$ as $\wh\beta$ has better numerical  performance and has  a smaller asymptotic variance (see Theorem \ref{thm_distr} and Theorem \ref{thm_distr_I} in the Appendix). 
	We also verify these points in the simulation study presented in Section \ref{sec_sims}.

	\subsection{Component-wise asymptotic normality of $\wh \beta$}\label{sec_inference_beta}
	In this section, for ease of  presentation, we assume that the signed permutation matrix $P$ is identity and we  consider $\Gamma = \tau^2{\bI}_p$ and $|I_k| = m$ for all  $k\in [K]$, but our proof holds for the general case and the corresponding explicit, general,  expression of the asymptotic variance is given in display (\ref{def_V_k_general}) of the Appendix. 
	
	The component-wise asymptotic normality of $\wh \beta$ is proved under the challenging, but realistic, scenario in which some of the non-pure variables are very close to the pure variables, justifying their name,  quasi-pure variables,  introduced in Section \ref{sec_ass}
	above. Allowing for this situation is similar to relaxing the signal strength conditions used in the literature on support recovery. In our  context, they would  correspond to requiring that the  pure and non-pure variables are well separated, in that 
	\[
	\min_{j\in J}\min_{P\in \H_K}\left\|A_{j\sbt} - P \e_1\right\|_1 \ge c\sqrt{\log\pn / n}
	\]
	for some universal constant $c>0$, {\it an assumption that we do not make}.  We note that such assumption would be  equivalent 
	to requiring $\bar{\rho}  = 0$, 
	for $\bar{\rho} $ defined in (\ref{def_rho}).
	
	In Section \ref{sec_ass} we  established the convergence rate of $\wh \beta$  when $\bar{\rho} \neq 0$, but satisfies Assumption \ref{ass_J1}.  For the asymptotic normality result, we can still allow for $\bar{\rho}  \neq 0$, but require it to be of the smaller size stated in Assumption \ref{ass_J1_prime}.
	\begin{customAss}{3$\bm{^\prime}$}\label{ass_J1_prime}
		The overall proportion 	$\bar\rho$ satisfies	$\bar\rho^2\log\pn= o(1 / m)$ as $n\to\infty$.
	\end{customAss}

	
	\begin{customAss}{5$\bm{^\prime}$}\label{ass_H_prime}
		The $K$th eigenvalue $\lambda_K = \lambda_K(A\C A^\top )$ of $A\C A^\top$ satisfies
		\begin{equation*}
		\lambda_K \Bigg / {p\log\pn \over  \sqrt n} \to \i, \ \text{as} \ n \rightarrow \infty. 
		\end{equation*} 
	\end{customAss}

	\begin{thm}\label{thm_distr}
		Under 
		Assumptions \ref{ass_model}, \ref{ass_subg}, \ref{ass_J1_prime}, \ref{ass_C} and \ref{ass_H_prime},  assume 
		$\gamma_w / \tau = O(1)$ and $\gamma_{\eps} / \sigma = O(1)$. 
		Then, with probability tending to one, $\wh K=K$,
		$\wh\Theta^\top \wh \Theta$ is non-singular and 
		for any $1\le k\le K$,
		\[
		\sqrt{n/V_k}(\wh \beta_k - \beta_k) \overset{d}{\to} N\left(0,1\right),\qquad \text{as}\quad n\to \i,
		\]
		where, with $\Theta^+ = (\Theta^\top\Theta)^{-1}\Theta^\top$, 
		\begin{equation}\label{def_Vk}
		V_k = \left(\sigma^2 + {\t^2\over m}\|\beta\|_2^2 \right)\left[\Omega_{kk} + \tau^2\e_k^\top \left(\Theta^\top \Theta\right)^{-1}\e_k\right] + {\tau^4\over m(m-1)}\sum_{a=1}^K \beta_a^2 \sum_{i\in I_a}\left[\e_k^\top \Theta^+ \e_i\right]^2 .
		\end{equation}
		Furthermore, if additionally $\lambda_K/ \tau^2 \to \i$ holds, then 
		\begin{equation}\label{def_Vk_simp}
		V_k = \left(\sigma^2 + {\t^2\over m}\|\beta\|_2^2 \right)\Omega_{kk} .
		\end{equation}
	\end{thm}
	\begin{proof}
		The proof is deferred to Appendix \ref{sec_proof_asn_beta}, and we offer insights into its main steps below.
	\end{proof}

	\paragraph{Outline of the proof of Theorem \ref{thm_distr}:}
	There are four main steps in the proof. We briefly explain them below and highlight the difficulties in each step. 
	
	In the classical approach of establishing the asymptotic normality for $\wh \beta_k - \beta_k$, a crucial step is to decompose the expression of $\wh \beta_k - \beta_k$ as a sum of independent mean-zero random variables, that serves as the leading term,  plus a remainder  term of smaller order. The asymptotic variance of the main term determines the asymptotic variance of $\wh\beta_k - \beta_k$.

	Under our setting (\ref{model}) and (\ref{love}), the first step of proving Theorem \ref{thm_distr} is to establish such a decomposition on the event that the dimensions   of $\wh\beta$ and $\beta$ are equal. This event holds with overwhelming probability tending to one. In display (\ref{disp_clt_term}) of the proof, we show that indeed, on this event, 
	\begin{align*}
	\sqrt{n}(\wh \beta_k - \beta_k) & = {1\over \sqrt n} \sum_{i=1}^n {\xi_{ik}} + \sqrt{n}([{\rm Rem}_1]_k +[{\rm Rem}_2]_k), 
	\end{align*}
	with ${\rm Rem}_1$ and ${\rm Rem}_2$ defined in display (\ref{def_Rem_2}). 
	Each summand $\xi_{ik}$ is in turn a sum of four terms that are   
	bi-linear combinations of  $\eps$, $Z$ and $W$. 
	The $\xi_{ik}$ are independent and form a triangular array  since $K$ and $p$ may grow in $n$.
	Interestingly, if $X=AZ$ and $W = 0$, $\xi_{ik}$ reduces to $\e_k^\top \C^{-1} \Z_{i\sbt} \eps$, the usual error term in the analysis of  the ordinary least squares estimator based on observed  $Z$. 
	For the two remainder terms,  we note that ${\rm Rem}_1$ depends on the error of estimating $\Theta^+:= (\Theta^\top \Theta)^{-1}\Theta^\top$,  while ${\rm Rem}_2$ is induced by the existence of quasi-pure variables indexed by  $J_1$.
	
	The second step of the proof is to calculate the first two moments of $\xi_{ik}$ via Lemmas \ref{lem_moment} and \ref{lem_WW}, which is relatively straightforward algebra.
	
	In the third step, we apply  Lyapunov's central limit theorem to $\sum_{i=1}^n \xi_{ik}$
	for triangular arrays.
	Verification of the Lyapunov condition requires calculation of the  third moments of $\xi_{ik}$.
	We rely on  Rosenthal's inequality and  a careful analysis to accomplish this in Lemma \ref{lem_third_moment}.  
	
	The final, fourth step is to show that both $[{\rm Rem}_1]_k$ and $[{\rm Rem}_2]_k$ are negligible as $n\to\i$. This requires a fair amount of work. 
	
	To control the remainder term ${\rm Rem}_1$, a key step is to provide upper bound for
	the quantity $(\wh \Theta - \Theta)^\top \wh \Theta (\wh \Theta^\top \wh\Theta)^{-1}$ and even establishing the existence of $(\wh \Theta^\top \wh \Theta)^{-1}$ requires a delicate analysis. Lemmas \ref{lem_H_op} and \ref{lem_Rem2} are devoted to this goal.
	In order to ensure that $[{\rm Rem_1}]_k/\sqrt{V_k} = o_p(1)$, we need the  signal strength condition in 
	Assumption \ref{ass_H_prime}, which is slightly stronger relative  to Assumption \ref{ass_H}.  
	Assumption \ref{ass_H_prime} implies
	$K\log\pn = o(\sqrt{n})$, needed for analyzing the estimator
	of $(\Theta^\top\Theta)^{-1}$,  as we recall that we allow for $\Theta^\top \Theta$ to be general, in particular  we  do not impose any sparsity assumption on it.
	
	The remainder term 
	$[{\rm Rem}_2]_k$ is a complicated function of random quantities that involve sums or maxima over the quasi-pure variable index set $J_1$. If no such variable exists  ($\bar\rho = 0$), then ${\rm Rem}_2 = 0$, and the proof ends.
	However, since we allow for  $\bar\rho \neq 0$,  it turns out to be challenging  to show that 
	$\sqrt{n/V_k}[{\rm Rem}_2]_k $ 
	still vanishes asymptotically under Assumption \ref{ass_J1_prime}. This is done in Lemmas \ref{lem_Rem_2} and \ref{lem_diag_gamma_hat_td}.
	\bigskip
	
	
	In practice, estimation of  $V_k$ is required   to construct   valid confidence intervals for individual coordinates of $\beta$. We propose a simple  plug-in estimator $\wh V_k$ by substituting  $\sigma^2$, $\t_i^2$, $|I_k|$, $\beta$ and $\Omega$ by their estimates. Specifically, we use $\wh \Omega = \whC^{-1}$ and estimate
	$\sigma^2$ and  $\tau_i^2$ by
	\begin{align}\label{est_Gamma}
	\wh\tau_i^2 &= \wh \Sigma_{ii} - \wh A_{i\sbt}^\top \whC\wh A_{i\sbt},\quad \text{for all }i\in [p]; \\\label{est_sigma}
	\wh\sigma^2  &= {1\over n} \y^\top \y - 2\wh\beta^\top \wh h + \wh \beta^\top \whC\wh \beta
	\end{align} 
	with 
	\begin{equation}\label{est_h}
	{\wh h} = {1\over n}\left(\wh A_{{\wh I\sbt}}^\top \wh A_{{\wh I\sbt}}\right)^{-1}\wh A_{{\wh I\sbt}}^\top \X_{\sbt \wh I}^\top \y.
	\end{equation}
	If either $\wh\tau_i^2$ or $\wh\sigma^2$ is negative, we set it to 0. We estimate $A_{I\sbt}$ according to  (\ref{est_AI_a}) -- (\ref{est_AI_b}) and estimate $A_{J \sbt}$ by using the Dantzig-type estimator $ \wh A_D$ proposed in \cite{LOVE} given by
	\begin{equation}\label{est_AJ}
	\wh A_{j\sbt} = \arg\min_{\beta^j\in \R^K}\left\{\|\beta^j\|_1:\  \left\|\whC \beta^j - (\wh A_{{\wh I\sbt}}^\top \wh A_{{\wh I\sbt}})^{-1}\wh A_{{\wh I\sbt}}^\top \wh\Sigma_{\wh Ij}\right\|_\i \le c\sqrt{\log\pn / n}\right\}
	\end{equation}
	for any $j\in \wh J$, with some constant $c>0$.
	The estimator $\wh A$ enjoys the optimal convergence rate of $\max_{j\in [p]}\|\wh A_{j\sbt} -A_{j\sbt}\|_q$ for any $1\le q\le \i$ \citep[Theorem 5]{LOVE}.
	Finally, $\Theta$ is estimated by (\ref{est_Theta}).  
	
	The next proposition shows that the plug-in estimator $\wh V_k$ consistently estimates 
	the asymptotic variance $V_k$ of $\wh \beta_k$.
	\begin{prop}\label{prop_est_V}
		Under the same conditions of Theorem \ref{thm_distr}, we have with probability tending to one, $\wh K=K$ and
		\[
		\left|{\wh V_k^{1/2}/ V_k^{1/2}} - 1\right|  = o_p(1).
		\]
		Consequently, we have with probability tending to one, $\wh K=K$ and
		\[
		\sqrt{n/\wh V_k}(\wh \beta_k - \beta_k )\overset{d}{\to} N(0, 1),\qquad \text{as }n\to\i, \quad k\in [K].
		\]
	\end{prop}
	\begin{proof}
		The proof of the consistency of $\wh V_k$ is given in Appendix \ref{sec_proofs_prop_V} and the rest of the proof follows from  Theorem \ref{thm_distr} and an application of the Slutsky's theorem.
	\end{proof}
	
	\begin{remark} \label{rem_inference_results}
		{\rm	If we treat model (\ref{model}) and (\ref{love}) as an augmented factor model, the vector $\beta$ simply corresponds to a particular row of the augmented matrix $\wt A = [A^\top , \beta^\top ]^\top $. As mentioned in the Introduction, and to the best of our knowledge, the only asymptotic normality results, with explicit asymptotic variance, have been derived in 
			\cite{bai2012}, when $I$ and $K$ {\it are known}, $K$ is a fixed constant, and $p \rightarrow \infty$.  In this framework, 	\cite{bai2012} show  
			\begin{equation}\label{uk}
			\sqrt{n/Q_k}(\wt \beta_k - \beta_k)  \overset{d}{\to} N(0, 1),\qquad \text{with}\qquad 
			Q_k  = \left(\sigma^2 + \t^2\|\beta\|_2^2 \right)\Omega_{kk},
			\end{equation} as $n\to\i$,
			for $1\le k\le K$, for an MLE-type estimator $\wt \beta$ of $\beta$. 
			
			We discuss below the relative computational and theoretical enhancements  offered  by  Theorem \ref{thm_distr} above, first within their framework, and then beyond it.

			At the computational level, \cite{bai2012} propose to estimate $\beta$ via an alternating EM-algorithm, but offer no guarantees that this estimator coincides with the MLE-type estimator $\wt \beta$ that is  theoretically studied. In contrast, the estimator $\wh \beta$ constructed in Section \ref{sec_est_beta} is also the estimator  analyzed theoretically in this work. 
			
			At the theoretical level, the estimator $\wt \beta$  of \cite{bai2012} is analyzed under
			
			\begin{equation}\label{conditions_Bai}
			p\lesssim \lambda_K(A\C A^\top ) \le  \lambda_1(A\C A^\top ) \lesssim p,\qquad c \le  \lambda_K(\C) \le \lambda_1(\C) \le C
			\end{equation}
			for some constants $0<c< C<\infty$, among other technical assumptions.
			Our estimator $\wh \beta$ is analyzed under considerably weaker assumptions. For instance, our condition on $\lambda_K(A\C A^\top )$ in Assumption \ref{ass_H} considerably relaxes the above condition (\ref{conditions_Bai}),  we don't require $\lambda_1(A\C A^\top ) \asymp \lambda_K(A\C A^\top )$, but  we allow the condition number of $A\C A^\top$  to grow as fast as   ${n}^{1/2} / \log\pn$. The latter follows  from Assumption \ref{ass_H_prime} in conjunction with the fact that $\lambda_1(A\C A^\top) \lesssim p $.
			
			Furthermore, when (\ref{conditions_Bai}) holds, $\|\beta\|_2$ is bounded and $p \rightarrow \infty$, as assumed in \cite{bai2012}, display (\ref{def_Vk_simp}) shows that the asymptotic variance $V_k$ of our estimator $\wh \beta_k$ reduces to $Q_k$ given in (\ref{uk}) above, when $m = 1$ and $I$ is {\it known}, as considered in \cite{bai2012}.

			In summary, Theorem \ref{thm_distr} holds uniformly over $\beta\in \R^K$, and for both fixed and growing dimensions $K$ and $p$, and whether $I$ is known or not.  While unifying these cases requires a much more  complicated analysis, the reward is that  our asymptotic analysis, and in particular the asymptotic variance $V_k$ can be derived simultaneously for all cases of interest. 	Furthermore, we provide a consistent estimator of $V_k$, which to the best of our knowledge has not been considered elsewhere. 
		}

	\end{remark}

	\section{ Essential Regression as regression on latent cluster centers}\label{sec_app_to_clustering}
	
	The decomposition (\ref{love}) in our model formulation can be used as a model for possibly overlapping clustering. For this, we interpret $A$ as an allocation  matrix that assigns the $X$-variables to possibly overlapping groups $G_k$ corresponding to the components of $Z$ via  
	\[ G_k = \{j\in [p] : A_{jk} \neq 0 \}. \] 
	This approach was first proposed in \cite{LOVE}, and their algorithm, called LOVE, was shown to estimate clusters consistently. With this interpretation, 
	the quantity (at the population level)
	$$
	\bar X := (A^\top A)^{-1}A^\top X,$$
	can be viewed as weighted cluster averages of all variables.
	As discussed in the Introduction, 
	Essential Regression provides a framework within which we can analyze when 
	the commonly used cluster averages can be used for downstream analysis, with statistical guarantees. For inference on  $\beta$, we remark that, at the population level, 
	\begin{align}\label{disp_beta}
	\beta = 
	\arg \min_\alpha \EE\left[\left( Y-\alpha^\top  \wt Z\right)^2\right] 
	\neq ~ \arg \min_\alpha \EE\left[\left( Y-\alpha^\top  \bar X\right)^2\right] 
	\end{align}
	where $\wt Z$ is the best linear predictor (BLP) of $Z$ from $	\bar X$, given by 
	\begin{align}\label{def_barZ_I2}
	\wt Z = ~\Cov(Z, \bar X)[\Cov(\bar X)]^{-1}\bar X= \Theta^\top \Theta\left(\Theta^\top \Sigma \Theta\right)^{-1}\Theta^\top X.
	\end{align}
	Display (\ref{disp_beta}) suggests that estimation of $\beta$ should be based on the BLP of $Z$ rather than the weighted cluster averages. This is indeed true. We let $\wh \Z = \X\wh\Theta(\wh\Theta^\top  \wh\Sigma \wh \Theta)^{-1}\wh\Theta^\top \wh\Theta$, which is well defined  provided that  $(\wh\Theta^\top  \wh\Sigma \wh \Theta)^{-1}$ exists. The latter is met with high probability under Assumption \ref{ass_H}. Consider the least squares  estimator $\wt \beta$ corresponding to regressing $\y$ onto $\Z$. Then 
	\begin{align*}
	\wt \beta &= \left(\wh \Z^\top \wh \Z\right)^{-1}\wh \Z^\top  \y = \left(\wh\Theta^\top \wh\Theta \right)^{-1}\left(\wh\Theta^\top  \wh\Sigma \wh \Theta\right)\left(\wh\Theta^\top  \wh\Sigma \wh \Theta\right)^{-1}\wh\Theta^\top  \X^\top \y = \wh\beta,
	\end{align*}
	by using $\wh\Sigma = n^{-1}\X^\top \X$. This natural approach yields exactly the same estimator $\wh{\beta}$  we introduced in Section \ref{sec_est_beta}. We can thus view $\wh \beta$ as a post-clustering estimator, and the results of Theorems \ref{thm_beta} and \ref{thm_distr} as pertaining to post-clustering inference, at the factor level. 
	

	\section{Simulations}\label{sec_sims}
	In this section, we complement and support our theoretical findings with simulations, focusing on the $\ell_2$ convergence rate of $\wh\beta$ and on the attained  coverage of the  corresponding 95\% confidence interval (CI) for $\beta$. Additional simulation results are provided in Appendix \ref{app_sim} to investigate the impact of  a potential inconsistent estimation of the number of factors  on subsequent estimation steps, and on inference for  $\beta$.
	
	\paragraph{Data generating mechanism:} We first describe how we generate $A$, $\C$, $\Gamma$, and $\beta$. Recall that $A$ can be partitioned into $A_{I \sbt}$ and $A_{J \sbt}$. To generate $A_{I \sbt}$, we set $|I_k| = m$ for each $k \in [K]$ and choose $A_{I \sbt} = {\bI}_K \otimes {\bm 1}_m$, where $\otimes$ denotes the Kronecker product. Each row of $A_{J \sbt}$ is generated by first randomly selecting its support with cardinality drawn from $\{2,3\ldots,K\}$ and then by sampling its non-zero entries from $\textrm{Uniform}(0,1)$. 
	In the end, we rescale $A_{J \sbt}$ such that the $\ell_1$-norm of each row is no greater than 1. 
	We multiply all entries in $A_{J \sbt}$    by independent random signs.
	To generate $\C$, we range its $K$ diagonal   
	entries  from 2.5 to 3 with equal increments. The  off-diagonal elements of $\C$ are then chosen as 
	$[\C]_{ij} = (-1)^{(i+j)}([\C]_{ii} \wedge [\C]_{jj})(0.3)^{|i-j|}$ for any $i\ne j \in [K]$. Finally, $\Gamma$ is chosen by randomly sampling its diagonal elements from $\textrm{Unif}(1, 3)$.
	
	Next, we generate the $n\times K$ matrix $\Z$ and the $n\times p$ noise matrix $\W$ by generating   i.i.d. rows  from $N_K(0, \C)$ and $N_p(0, \Gamma)$, respectively. Finally,  we set $\X = \Z A^\top  + \W$ and $\y = \Z\beta + \eps$ where the $n$ components of $\eps$ are i.i.d.\ $N(0, 1)$. For each setting, we repeat generating pairs $(\X, \y)$ 200 times and record the corresponding results.
	
	\subsection*{Convergence rate and the coverage of the 95\% confidence interval}\label{sec_sim_beta}
	
	We consider the following four settings: 
	\begin{enumerate}
		\setlength\itemsep{0mm}
		\item[(1)] Fix $p = 400$, $K = 10$, $m = 5$, and vary $n\in\{200,400, 600, 800\}$;
		\item[(2)] Fix $n = 300$, $K = 10$, $m = 5$, and vary $p\in\{100,300, 500, 700\}$;
		\item[(3)] Fix $n = 300$, $p = 400$, $m = 5 $ and vary $K \in \{5, 10, 15, 20\}$;
		\item[(4)] Fix $n = 300$, $p = 400$, $K = 10 $ and vary $m \in \{2, 5, 10, 15, 20\}.$
	\end{enumerate}
	The entries of $\beta$ are independently sampled from $\textrm{Unif}(1,3)$. For each setting, we calculate the averaged $\ell_2$ errors $\|\wh\beta - \beta\|_2$ of the following four estimators, the proposed estimator and three other estimates,  and report them in Table : 
	\begin{itemize}[leftmargin = 8mm]
		\setlength\itemsep{0mm}
		
		\item Our proposed estimator: $\wh \beta$ constructed in (\ref{est_beta}); 
		\item $\wh\beta^{(A)}$ defined in (\ref{est_beta_mm});
		\item $\wh\beta^{(I)}$ defined in (\ref{est_beta_I}); 
		\item $\wh\beta_{\rm naive}$ obtained by naively regressing $\y$ on $\bar \X = \X\wh A [\wh A^\top \wh A]^{-1}$;
		\item $\wh\beta_{\rm oracle} = (\Z^\top \Z )^{-1}\Z^\top  \y$, the oracle least squares estimator based on  the true   matrix $\Z$.
	\end{itemize}
	
	We focus on the best feasible estimators, with respect to their respective $\ell_2$ error, $\wh \beta$ and 
	$\wh \beta^{(I)}$, and check the coverage of their corresponding 95\% confidence intervals (CI). Table \ref{tab_1} shows the coverage and average length of the CI's for $\beta_1$, respectively constructed from $\wh\beta$, based on Theorem \ref{thm_distr} and Proposition \ref{prop_est_V} and  $\wh\beta^{(I)}$, based on Theorem \ref{thm_distr_I} and Proposition \ref{prop_est_Q} in the Appendix.


	\paragraph{Summary:}
	
	As expected, the oracle estimator  $\wh\beta_{\rm oracle}$ is the best performer, since it uses the true $\Z$, not available to the other estimates.
	Among the remaining estimators, our proposed estimator  $\wh\beta$ 
	outperforms the other three estimators in all cases.
	The gap between $\wh\beta$ and $\wh\beta_{\rm oracle}$ decreases as either $n$ or $m$ increases or $K$ decreases.
	We find that $\wh\beta^{(I)}$ has better performance than $\wh\beta_{\rm naive}$ and $\wh\beta^{(A)}$. While dominated by $\wh\beta_{\rm oracle}$ and $\wh\beta$, we find that for large $m$, $\wh\beta^{(I)}$ performs similarly to $\wh\beta$. Overall, the naive estimator $\wh\beta_{\rm naive}$ has the worst performance,  which supports our findings in Section \ref{sec_app_to_clustering}. 
	
	The estimation  errors of $\wh\beta$ and $\wh\beta^{(I)}$ decrease as $n$ and/or $m$ increase. The estimation errors increase in  $K$.
	It is worth mentioning  that increasing $p$ barely affects the  estimation errors. These findings support  Theorem \ref{thm_beta}. 
	
	Regarding the CIs of $\beta$ based on $\wh \beta$, the average
	coverage, over $200$ repetitions,  are close  to the nominal  95\% in most  settings, especially for  moderately large sample size $n$. This
	further supports the  results of  Section \ref{sec_inference_beta}. The coverage level of the intervals based on  $\wh\beta$ are closer to the  95\%  level than those based on  $\wh\beta^{(I)}$.  The averaged CI lengths corresponding to $\wh\beta$ are also smaller than those relative to  $\wh\beta^{(I)}$,  in most of the settings we considered. This suggests that $\wh\beta$ is more efficient than $\wh\beta^{(I)}$.  We further corroborate the validity of Theorem \ref{thm_distr} and Proposition \ref{prop_est_V}
	in Figure \ref{fig_hist_1} in Appendix \ref{sec_histograms}. This figure depicts histograms based on 200 values  of  $\sqrt{n/\wh V_k}(\wh\beta_1 - \beta_1)$.

	\begin{table}[ht]
		\centering
		{\renewcommand{\arraystretch}{1.5}
			\begin{tabular}{lccccc|cc|cc}
				\hline
				& $\wh\beta$ & $\wh\beta^{(I)}$ &  $\wh\beta_{\rm naive}$ &  $\wh\beta^{(A)}$   & $\wh\beta_{\rm oracle}$ & \multicolumn{2}{c|}{CIs of $\wh\beta$} &  \multicolumn{2}{c}{CIs of $\wh\beta^{(I)}$}\\\hline  
				\multicolumn{6}{l|}{
					{\bf	Vary $n$ with $p = 400$, $K = 10$, $m=5$}} & coverage & length & coverage & length\\
				$n = 200$  & 0.045 & 0.052 & 0.204 & 0.106 & 0.002 & 91.0 & 0.76 & 92.5 & 0.85\\ 
				$n = 400$ & 0.019 & 0.023 & 0.135 & 0.052 & 0.001 & 95.0 & 0.53 &  93.5 & 0.58 \\ 
				$n = 600$ & 0.013 & 0.016 & 0.112 & 0.036 & 0.001 & 93.5 & 0.44 & 94.0 & 0.47 \\ 
				$n = 800$ & 0.009 & 0.011 & 0.101 & 0.029 & 0.001 & 94.0 &  0.38 & 93.5 & 0.40\\ \hline
				\multicolumn{6}{l|}{{\bf Vary $p$ with $n = 300$, $K=10$, $m=5$}} & \multicolumn{2}{c|}{} \\
				$p=100$ & 0.029 & 0.035 & 0.183 & 0.029 & 0.002 & 94.5 & 0.69 & 92.5  & 0.72 \\ 
				$p=300$ & 0.029 & 0.035 & 0.181 & 0.067 & 0.002 & 94.0 & 0.64 & 92.0 & 0.69 \\ 
				$p=500$ & 0.031 & 0.039 & 0.179 & 0.088 & 0.002 & 95.0 & 0.65 & 95.0 & 0.72\\ 
				$p = 700$ & 0.030 & 0.037 & 0.176 & 0.103 & 0.001 & 94.5 & 0.63 & 94.0 & 0.70 \\ 
				\hline
				\multicolumn{6}{l|}{{\bf Vary $K$ with $n = 300$, $p = 400$, $m=5$}} & \multicolumn{2}{c|}{} \\
				$K=5$ & 0.016 & 0.020 & 0.064 & 0.041 & 0.001 & 91.0 & 0.44 & 91.0 & 0.45  \\ 
				$K=10$ & 0.026 & 0.032 & 0.177 & 0.080 & 0.001 & 94.5 & 0.63 & 92.0 & 0.68 \\ 
				$K=15$ & 0.052 & 0.065 & 0.282 & 0.099 & 0.002 & 93.0 & 0.88 & 94.0 & 0.96  \\ 
				$K=20$ & 0.046 & 0.057 & 0.192 & 0.064 & 0.002 & 97.5 & 0.81 & 96.5 & 0.89\\ 
				\hline
				\multicolumn{6}{l|}{{\bf Vary $m$ with $n = 300$, $p = 400$, $K=10$}} & \multicolumn{2}{c|}{}\\
				$m=2$ & 0.116 & 0.191 & 0.301 & 0.171 & 0.002 & 90.1 & 1.02 & 90.1 & 1.34 \\ 
				$m=5$ & 0.030 & 0.035 & 0.173 & 0.080 & 0.001 & 93.5 & 0.65 & 91.0  & 0.70 \\ 
				$m=10$ & 0.015 & 0.016 & 0.097 & 0.036 & 0.002 & 94.0 & 0.48  & 93.0 & 0.48 \\ 
				$m=15$ & 0.011 & 0.012 & 0.060 & 0.022 & 0.002 & 92.0 & 0.40 & 90.5 & 0.39 \\ 
				$m=20$ & 0.008 & 0.008 & 0.035 & 0.012 & 0.001 & 96.5 & 0.35 & 96.0 & 0.34 \\ 
				\hline 
		\end{tabular}}\vspace{2mm}
		\caption{$\ell_2$ error  of different estimators and the coverages and\\ the averaged lengths of the 95\% CIs of $\beta_1$.}
		\label{tab_1}
	\end{table}


	\section{Analysis of SIV-vaccine induced humoral immune responses} \label{sec_real_data}

	We tested Esential Regression on a high-dimensional dataset of vaccine-induced humoral immune responses, from a recently published study that demonstrated multiple antibody-centric mechanisms of vaccine-induced protection against SIV \citep{Jishnu}, the non-human primate equivalent of HIV. The dataset comprised $p = 191$ antibody functional and biophysical properties, including Fc effector functions, glycosylation profiles and binding to Fc receptors. The properties were measured for $n = 60$ non-human primates (NHPs). For each NHP, the level of protection offered by the vaccine (number of intra-rectal SIV challenges after which the NHP got infected or whether the NHP remained uninfected after the maximum number of challenges for the study, 12, normalized by the total number of challenges)  was used as the outcome $Y \in [0, 1]$ we regressed to.
	
	\begin{figure}[ht]
		\centering
		\begin{tabular}{c}
			\includegraphics[width = .7\textwidth]{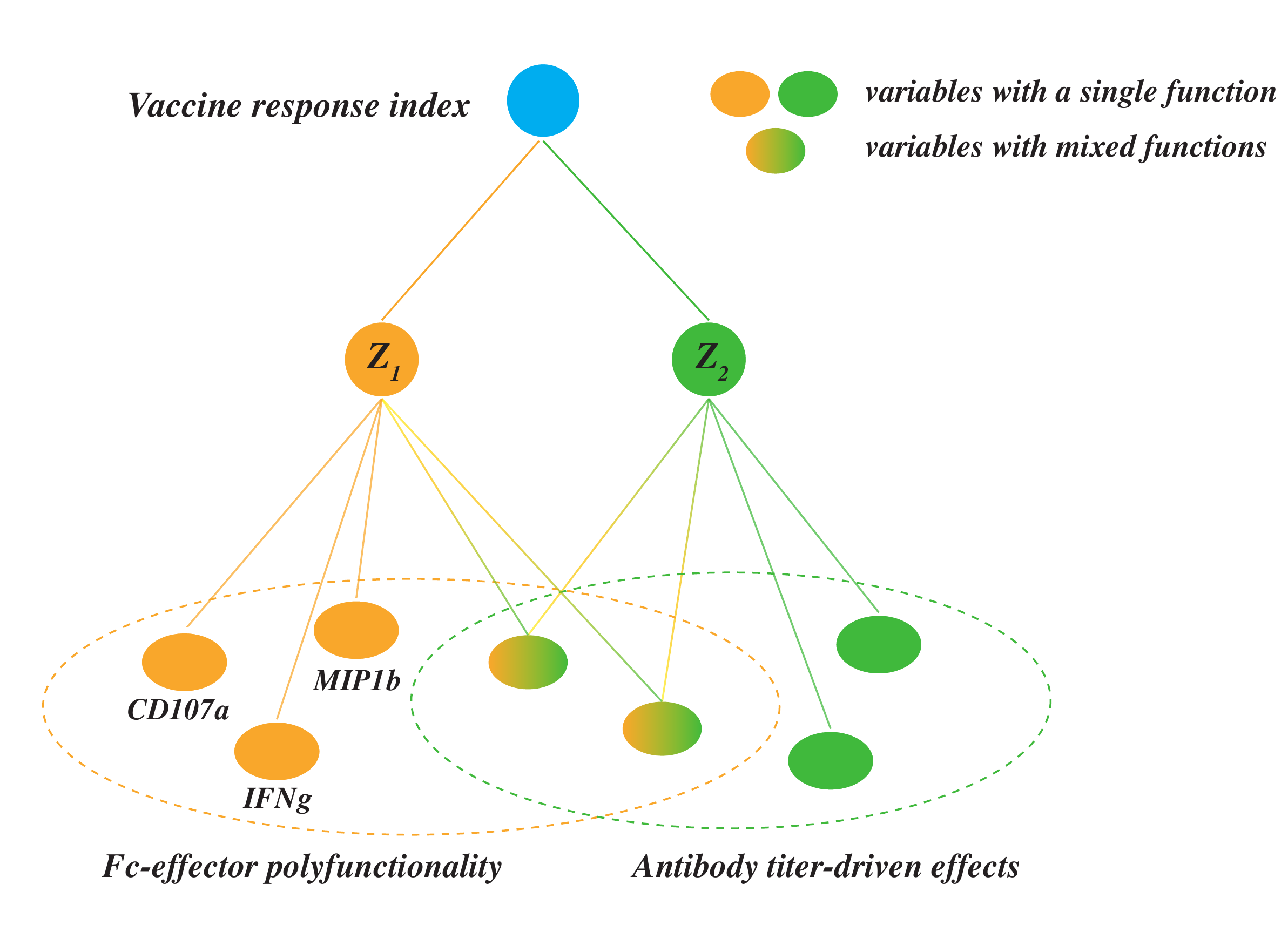}
		\end{tabular}
		\caption{
			Two representative clusters with their pure variables. (The overlapping variables between these two clusters are more than the plot shows.)}
		\label{fig-real-data}
	\end{figure}
	
	One goal of the study was to determine the un-observed  humoral signatures associated with the level of protection offered by the vaccine $Y$,  and suggests therefore a latent factor regression framework.  In particular, the Essential Regression model is ideally suited for this data set, in light of prior biological knowledge on the measured $X$-variables: some of the measured antibody properties work in tandem with several other properties (mixed-function variables), while others are part of individual immunological signatures (pure/single-function variables) \citep{Bournazos-Ravetch, Nimmerjahn}. For this data set we used the algorithm developed in \cite{LOVE} to obtain  an estimator $\wh K =10$ of the number of factors. 
	
	We used the asymptotic normality  of $\wh\beta$, established in Section \ref{sec_inference_beta} above, to determine the strength of association between $Y$  and the biologically interpretable  immunological signatures. 
	This  task  is difficult to  accomplish, with theoretical guarantees,  outside a latent factor regression framework. Common existing approaches include standard regularized regression at the observed bio-marker level, followed by an ad-hoc re-creation of clusters and cluster centers \citep{Jishnu}. Although subsequent regression of $Y$ onto cluster centers, appropriately defined, can be easily performed, theoretical justifications of such procedure is lacking. 
	In contrast, Essential Regression coupled with Theorem \ref{thm_distr} and Proposition \ref{prop_est_V} provides a principled way of regressing directly onto the latent cluster centers.  Figure \ref{fig-real-data} depicts the top two biological functions associated with the level of protection offered by the vaccine, under FDR-control. The estimated coefficients  are $\wh\beta_1 = 0.104$ with asymptotic $90\%$ confidence interval $[0, 0.21]$,  and $\wh \beta_2 = 0.105$ with 90\% asymptotic confidence interval $[0.02, 0.19]$, corresponding respectively to $Z_1$  on the left of Figure \ref{fig-real-data}, and  to $Z_2$,  on the right. On the basis of the pure and mixed variables in the two associated clusters,  $Z_1$ and $Z_2$ can be broadly defined as Polyfunctionality involving multiple Fc effector functions and Enhanced IgG titers and FcR2A binding, respectively. These findings are in excellent alignment with  biological expectations, providing strong support for the applicability of the methods and theory developed in this work even in data sets of  modest sample size.

	\section*{Acknowledgements} 
	We are grateful to Jishnu Das for help with the interpretation of our data analysis results. 
	Bunea and Wegkamp are supported in part by NSF grant DMS-1712709. Bing is supported in part by NSF grant DMS-1407600.

	\section*{Supplementary Material}\label{supp}
	The supplementary document includes all the proofs, the data-driven selection of the tuning parameter and auxiliary results.

	{\setlength{\bibsep}{0.85pt}{
			\bibliographystyle{ims}
			\bibliography{ref}
	}}

	\newpage
	\appendix

	\section{The identifiability of $A$ when $\Gamma$ is known}\label{app_cor}
	
	We state and prove the identifiability of $A$ when $\Gamma$ is known under Assumption \ref{ass_model} but when {\it (A1)} is replaced by 
	\begin{flalign*}
	\text{{\em (A1')} {\it For each $k\in [K]$, there exists at least one index  $i\in [p]$ such that $A_{i\sbt}  = \e_k$.}}
	\end{flalign*}
	
	\begin{cor}\label{cor_ident}
		Under {\it (A0)}, {\it (A2)} of Assumption \ref{ass_model} and {\it (A1')}, suppose $\Gamma$ is known. Then the matrix $A$ is identifiable from $\Sigma$, up to a $K\times K$ signed permuation matrix. 
	\end{cor}
	\begin{proof}
		When $\Gamma$ is known, one can identify $A\C A^\top = \Sigma - \Gamma$.  If $I$ can be identified, then $A$ is identifiable by the proof of Theorem 2 in \cite{LOVE}. It remains to show that $I$ is identifiable from $A\C A^{\top}$. This can be shown by repeating the same arguments of proving \cite[Theorem 1]{LOVE} except that we replace the definitions in (2.2) and (2.3) of \cite{LOVE} by
		\[
		M_i := \arg\max_{j\in[p]}|\Sigma_{ij}|, \quad  S_i := \{j\in [p]: |\Sigma_{ij}| = M_i\}
		\]
		for each $1\le i\le p$.
	\end{proof}

	\section{Proofs of Proposition \ref{prop_iden} and Theorem \ref{thm_beta_lower}}\label{sec_proofs_ident_lower}
	
	\subsection{Proof of Proposition \ref{prop_iden}: the identifiability of $\beta$}\label{sec_proofs_prop_iden} 
	From the structure of $\Sigma$ together with Assumption \ref{ass_model}, Theorem 1 in \cite{LOVE} can be directly invoked to show that $I$ and its partition $\I$ are identifiable up to a label permutation. In addition,  $A$ is identifiable up to a signed permutation. Suppose we identify $\wt{A} = AP$ for some signed permutation $P$ and, in particular, $\wt{A}_{I\sbt} = A_{I \sbt}P$. 
	
	First observe that for any $a,b\in [K]$, $[\C]_{ab}$ is recovered by 
	\begin{equation}\label{eqn:C_recovery}
	[\C]_{ab} = A_{ia}A_{jb}\Sigma_{ij}, \ \ \text{for } i\in I_a,\ j\in I_b,\ i\neq j.
	\end{equation}
	We prove this as follows. Since $\Sigma = A\C A^\top  + \Gamma$ with $\Gamma$ diagonal, 
	\[\Sigma_{ij} = \sum_{c,d=1}^KA_{ic}[\C]_{cd} A_{jc} =  A_{ia}A_{jb} [\C]_{ab},\]
	where in the second step we use that $i\in I_a$, $j\in I_b$. Then (\ref{eqn:C_recovery}) follows from $|A_{ia}|=|A_{jb}| =1$.
	
	From $\wt A$ and the corresponding partition $\{\wt I_a\}_{a\in [K]}$, we can define $\wt{\Sigma}_Z$ to be the matrix with elements $[\wt{\Sigma}_Z]_{ab} = \wt A_{i_a a}\wt A_{j_b b}\Sigma_{i_a j_b}$ for some $i_a\in \wt I_a$, $j_b\in \wt I_b$, and $i_a\neq j_b$. Let $\pi: [K] \to [K]$ be the permutation mapping corresponding to $P$. Specifically, $\pi(a)$ equals the unique $b\in [K]$ such that $|P_{ba}|=1$. Then for any $i\in [p]$ and $a\in [K]$, $\wt A = AP$ implies $\wt A_{ia} = P_{\pi(a)a}A_{i\pi(a)}$ and $\wt I_a = I_{\pi(a)}$, so
	\[[\wt{\Sigma}_Z]_{ab} =  P_{\pi(a)a} P_{\pi(b)b}A_{i_a\pi(a)} A_{j_b\pi(b)}\Sigma_{i_aj_b} = P_{\pi(a)a} P_{\pi(b)b} [\C]_{\pi(a)\pi(b)} = [P^\top \C P]_{ab},\]
	where we use (\ref{eqn:C_recovery}) in the second equality since $i_a\in I_{\pi(a)}$, $j_b\in I_{\pi(b)}$, $i_a\neq j_b$. This shows $\wt{\Sigma}_Z = P^\top \C P$. Finally, we have
	\begin{equation*}
	\wt{\beta} = \wt{\Sigma}_Z^{-1}(\wt{A}_{I\sbt}^\top \wt{A}_{I\sbt})^{-1}\wt{A}_{I\sbt}^\top \Cov(X_I, Y)=  \wt{\Sigma}_Z^{-1}(\wt{A}_{I\sbt}^\top \wt{A}_{I\sbt})^{-1}\wt{A}_{I\sbt}^\top A_{I\sbt}\C\beta = P^\top \beta
	\end{equation*}
	by using $\wt{A}_{I\sbt} = A_{I \sbt}P$ and $\wt{\Sigma}_Z = P^\top \C P$ in the last step. This concludes the proof.\qed

	\subsection{Proof of Theorem \ref{thm_beta_lower}: the minimax lower bounds for estimators of $\beta$}\label{sec_proofs_thm_beta_lower}
	
	Let $\PP_{\beta}$ and $\PP_{\beta'}$ denote the joint distribution of $(X_i, Y_i)$ for $i = 1,\ldots, n$, parametrized by the same $(A,\C)$ but different $\beta$ and $\beta'$, respectively. Denote by $\text{KL}(\PP_{\beta}, \PP_{\beta'})$ the Kullback-Leibler divergence between these two distributions. 	Since 
	\begin{align*}\sup_{(\beta, A, \C)\in \S(R, m)}&\PP_{A, \C, \beta}\left\{\|\wh \beta - \beta\| \ge c'\left(1\vee {R\over \sqrt m}\right)\cdot \sqrt {K\over n}\right\} \\
	&\ge \sup_{\|\beta\|\le R}\PP_{A^*, \C^*, \beta}\left\{\|\wh \beta - \beta\| 
	\ge c'\left(1\vee {R\over \sqrt m}\right)\cdot \sqrt {K\over n}\right\},
	\end{align*}
	for any fixed $A^*$ and $\C^*$, we let $\PP_\beta := \PP_{A^*, \C^*, \beta}$ for $A^*$ and $\C^*$ defined below and aim to prove 
	\[\inf_{\wh \beta}\sup_{\|\beta\|\le R}\PP_{\beta}\left\{\|\wh \beta - \beta\| \ge c'\left(1\vee {R\over \sqrt m}\right)\cdot \sqrt {K\over n}\right\}\ge c''.\]
	We choose $\C^*$ such that $0<\cl\le \lambda_{\min}(\C^*) \le \lambda_{\max}(\C^*) \le \cu<\i$ and 
	\begin{equation}\label{def_Anull}
	A^* := \begin{bmatrix}
	\1_m\otimes {\bI}_K \vspace{2mm}\\ 0
	\end{bmatrix}
	\end{equation}
	with $\otimes$ denoting the kronecker product and $\1_d$ denoting the vector in $\R^d$ with all ones.
	
	We start by constructing a set of hypothesis $S$ for $\beta$.  From Lemma \ref{lem_klop}, stated in Section \ref{sec_lem_proof_beta_lower}. with $s=k = K-1$, we can find a subset $S_0$ of the set of binary sequences $\{0, 1\}^{K-1}$ such that 
	\begin{enumerate}
		\item[(i)] $\log|S_0| \ge c_1(K-1)$, 
		\item[(ii)] $c_2(K-1)\le \|a\|_0 \le (K-1)$, for all $a\in S_0$.
		\item[(iii)] $\|a-b\|^2 \ge c_3(K-1)$, for all $a,b\in S_0$ and $a\ne b$,
	\end{enumerate}
	where $c_1, c_2, c_3>0$ are absolute constants. Let $v^{(0)} = (1, 0, \ldots, 0)\in \R^K$ and $v^{(j)} = (0, a)\in \R^K$ for all $a\in S_0$ so that $j\in \{1,\ldots, |S_0|\}$. We then define $\beta^{(0)} = (R, 0, \ldots, 0)$ and 
	\begin{equation}\label{def_beta_J}
	\beta^{(j)} := {R\over \sqrt{1+ \eta^2(K-1)}}\left(v^{(0)} + \eta v^{(j)}\right)\quad \text{for all }j\in \{1,\ldots, |S_0|\},
	\end{equation}
	with $\eta$ to be chosen later.
	
	It is easy to verify that $\|\beta^{(j)}\| \le R$ so that $(\beta^{(j)}, \C^*, A^*)\in \S(R, m)$ for $j\in \{0,1,\ldots, |S_0|\}$. Moreover, (iii) above implies that, for any $j,\ell \ge 1$ with $j\ne \ell$,
	\begin{equation}\label{disp_betaA_Jell}
	\|\beta^{(j)} - \beta^{(\ell)}\|^2 = {R^2\eta^2 \over 1+\eta^2(K-1)}\|v^{(j)}-v^{(\ell)}\|^2\ge c_3{R^2\eta^2(K-1) \over 1+\eta^2(K-1)},
	\end{equation}
	and (ii) above guarantees that, for any $j\ge 1$,
	\begin{align}\label{disp_betaA_J0'}
	\|\beta^{(j)} - \beta^{(0)}\|^2 &={R^2 (\sqrt{1+\eta^2(K-1)} - 1)^2 \over 1+ \eta^2(K-1)}  + {R^2\eta^2 \over 1+\eta^2(K-1)}\|v^{(j)}\|^2\\\nonumber
	&\ge{R^2\eta^2 \over 1+\eta^2(K-1)}\|v^{(j)}\|^2\\\label{disp_betA_J0}
	&\ge c_2{R^2\eta^2(K-1) \over 1+\eta^2(K-1)}.
	\end{align}
	
	On the other hand, for any $j\in \{1, \ldots, |S_0|\}$, Lemma \ref{lem_KL} in Section \ref{sec_lem_proof_beta_lower} implies 
	\begin{align*}
	&{1\over n}\text{KL}(\PP_{\beta^{(j)}}, \PP_{\beta^{(0)}}) \\
	&\le  {|(\beta^{(j)})^\top G^{-1}\beta^{(j)} - (\beta^{(0)})^\top G^{-1}\beta^{(0)}| + \|\C-G^{-1}\|_{{\rm op}}\|\beta^{(j)} - \beta^{(0)}\|^2 \over \sigma^2 + \min(\|\beta^{(j)}\|^2, \|\beta^{(0)}\|^2)/\|G\|_{{\rm op}}}
	\end{align*}
	with $G= \C^{-1} + \tau^{-2}A^\top A$.
	By using (\ref{def_beta_J}), the definition of $\beta^{(0)}$ and (ii), we further have
	\begin{align*}
	&|(\beta^{(j)})^\top G^{-1}\beta^{(j)} - (\beta^{(0)})^\top G^{-1}\beta^{(0)}|\\ 
	&= {R^2\eta\over 1+\eta^2(K-1)}\left| 2(v^{(0)})^\top G^{-1}v^{(j)} + \eta(v^{(j)})^\top G^{-1}v^{(j)} - \eta(K-1)(v^{(0)})^\top G^{-1}v^{(0)} \right|\\
	&= {R^2\eta^2\over 1+\eta^2(K-1)}\left| (v^{(j)})^\top G^{-1}v^{(j)} - (K-1)(G^{-1})_{11} \right|\\
	&\le {R^2\eta^2(K-1)\over 1+\eta^2(K-1)}\|G^{-1}\|_{{\rm op}}.
	\end{align*}
	Also note that $\|\beta^{(0)}\|^2 = R^2$ and 
	\[
	\|\beta^{(j)}\|^2 = {R^2 \over 1+\eta^2(K-1)}\left\| v^{(0)} + \eta v^{(j)}\right\|^2 =  {R^2(1 + \eta^2\|v^{(j)}\|^2) \over 1+\eta^2(K-1)} \overset{(ii)}{\ge}  cR^2.
	\]
	Together with
	\[
	\|\beta^{(j)} - \beta^{(0)}\|^2 \le {R^2\eta^2(K-1) \over 4[1+\eta^2(K-1)]}+ {R^2\eta^2 \over 1+\eta^2(K-1)}\|v^{(j)}\|^2 \overset{(ii)}{\le} {5R^2\eta^2(K-1)\over 4[1+\eta^2(K-1)]},
	\]
	from (\ref{disp_betaA_J0'}) and the fact that $f(x) = \sqrt{x}$ is concave for $x>0$,
	we obtain
	\begin{align*}
	\text{KL}(\PP_{\beta^{(j)}}, \PP_{\beta^{(0)}}) &\le {5\over 4}\cdot {nR^2\eta^2(K-1)\over 1+\eta^2(K-1)}\cdot  { \|G^{-1}\|_{{\rm op}}+ \|\C-G^{-1}\|_{{\rm op}}\over \sigma^2 + \min(\|\beta^{(j)}\|^2, \|\beta^{(0)}\|^2)/\|G\|_{{\rm op}}}\\
	&\le  3n\eta^2(K-1)\cdot  {R^2\cu\over \sigma^2 + cR^2/\|G\|_{{\rm op}}}.
	\end{align*}
	where in the second line we use 
	\begin{align*}
	\|G^{-1}\|_{{\rm op}} &\le  \|\C\|_{{\rm op}} \left\|\left[{\bI}_K + (\C)^{1/2}A^\top \Gamma^{-1} A(\C)^{1/2}\right]^{-1}\right\|_{{\rm op}} \le \|\C\|_{{\rm op}}.
	\end{align*}
	Further note that
	\begin{equation}\label{eq_G_0}
	\|G\|_{{\rm op}} \le \|\C^{-1}\|_{{\rm op}} + \tau^{-2}\|A^\top  A\|_{{\rm op}} \le \cl^{-1}+m/\tau^2 \le c'm/\tau^2.
	\end{equation} 
	Choosing 
	\begin{equation}\label{def_eta_lower}
	\eta^2 = c\cdot {\sigma^2 + c \tau^2R^2/(c'm) \over nR^2\cu}
	\end{equation}
	yields 
	\[
	\text{KL}(\PP_{\beta^{(j)}}, \PP_{\beta^{(0)}}) \le c\log |S_0|
	\]
	for any $j\ge 1$, and
	\[
	\|\beta^{(j)} - \beta^{(\ell)}\|^2 \ge c\left({\sigma^2 \over \cu} \vee {\tau^2R^2\over c'\cu m}\right){(K-1) \over n}\cdot {1\over 1+\eta^2(K-1)}
	\]
	for any $j\ne \ell$, from (\ref{disp_betaA_Jell}) and (\ref{disp_betA_J0}). 
	Since condition $K\le \bar{c}(R^2\vee m)n$ guarantees that $\eta^2(K-1) \le c = c(\bar{c})$, invoking Theorem 2.5 in \cite{Intro_non_para}  concludes
	\[
	\inf_{\wh\beta}\sup_{\|\beta\|\le R}\PP_{\beta}\left\{\|\wh \beta - \beta\| \ge c\left(1\vee {R \over \sqrt m}\right)\sqrt{K-1\over n}\right\} \ge c',
	\]
	which completes the proof.\qed
	
	\subsection{Lemmas used in the proof of Theorem \ref{thm_beta_lower}}\label{sec_lem_proof_beta_lower}

	\begin{lemma}\label{lem_klop}
		Let $k \ge 2$ and $s \ge 1$ be integers, $s \le k$. There exists a subset $S_0$ of the set of binary sequences $\{0, 1\}^k$ such that
		\begin{itemize}
			\item[(i)]  $\log|S_0| \ge c_1^*s \log(ek/s)$, 
			\item[(ii)] $c_2^*s\le \|a\|_0 \le s$, for all $a\in S_0$, and all $\|a\|_0 = s$ for $a\in S_0$, if $s\le k/2$, 
			\item[(iii)] $\|a-b\|^2 \ge c_3^*s$, for all $a,b\in S_0$ and $a\ne b$,
		\end{itemize}
		where $c_j^*>0$, $j = 1,2,3$  are absolute constants.
	\end{lemma}
	\begin{proof}
		This lemma is proved in \cite[Lemma 16]{klopp_2017}. 
	\end{proof}
	
	\medskip
	
	\begin{lemma}\label{lem_KL}
		Suppose that $(X_i, Y_i)$ are i.i.d. Gaussian from model (\ref{model}). Then, for any $\beta, \beta' \in \R^K$,  
		\[
		{1\over n}\text{KL}(\PP_{\beta}, \PP_{\beta'}) \le  {|\beta^\top G^{-1}\beta - \beta'^\top G^{-1}\beta'| + \|\C-G^{-1}\|_{{\rm op}}\|\beta - \beta'\|^2 \over \sigma^2 + \min(\|\beta\|^2, \|\beta'\|^2)/\|G\|_{{\rm op}}}
		\]
		where $G = \C^{-1} + A^\top \Gamma^{-1}A$.
	\end{lemma}	
	\begin{proof}
		By the additivity of the Kullback-Leibler divergence, it suffices to consider one data pair $(X_i, Y_i)$. We remove the subscript $i$ to lighten the notation. Note that, for given $\beta_j$ with $j=1,2$, model (\ref{model}) implies  
		\[
		\begin{bmatrix}
		Y\\
		X
		\end{bmatrix} \sim N_{p+1}\left(0,  ~
		\begin{bmatrix}
		\beta_j^\top \C\beta_j+\sigma^2 & \beta_j^\top \C A^\top \\
		A\C\beta_j & \Sigma
		\end{bmatrix} \right)
		\]
		with $\Sigma = A\C A^\top +\Gamma$. 
		This further yields
		\begin{align}\label{disp_cond_y|x}\nonumber
		Y|X &\sim N(\beta_j^\top \C A^\top \Sigma^{-1}X, \sigma^2 +\beta_j^\top (\C-\C A^\top \Sigma^{-1}A\C)\beta_j)\\
		& := N(\mu_j, \sigma^2_j).
		\end{align}
		Since the marginal distribution of $X$ does not depend on $\beta$, we observe that 
		\begin{align*}
		{1\over n}&\text{KL}(\PP_{\beta_1}, \PP_{\beta_2}) \\&~= \EE_{\beta_1}\left[\log f_{\beta_1}(Y|X)\right] - \EE_{\beta_1}\left[\log f_{\beta_2}(Y|X)\right]\\
		&= -{1\over 2}\log\sigma_1^2 - {1\over 2\sigma_1^2}\EE_{\beta_1}\left[(Y-\mu_1)^2\right]+ {1\over 2}\log\sigma_2^2  + {1\over 2\sigma_2^2}\EE_{\beta_1}\left[(Y-\mu_2)^2\right]\\
		&= {1\over 2}(\log \sigma_2^2 - \log \sigma_1^2)+ {1\over 2\sigma_2^2}\EE_{\beta_1}\left[ (Y-\mu_2)^2 - (Y-\mu_1)^2\right]\\
		&\qquad + {\sigma_1^2 - \sigma_2^2\over 2\sigma_1^2\sigma_2^2}\cdot  \EE_{\beta_1}\left[(Y-\mu_1)^2\right]\\
		&= {1\over 2}(\log \sigma_2^2 - \log \sigma_1^2)+ {1\over 2\sigma_2^2}\EE_{\beta_1}\left[ (Y-\mu_2)^2 - (Y-\mu_1)^2\right] + {\sigma_1^2 - \sigma_2^2\over 2\sigma_2^2}
		\end{align*}
		where the last inequality uses $\EE_{\beta_1}[(Y-\mu_1)^2] = \sigma_1^2$ from (\ref{disp_cond_y|x}). To calculate the expectation, it follows from $\EE[XX^\top ] = \Sigma$ and (\ref{disp_cond_y|x}) that
		\begin{align*}
		\EE_{\beta_1}\left[ (Y-\mu_2)^2 - (Y-\mu_1)^2\right]&= \EE_{\beta_1}\left[ \mu_2^2  - \mu_1^2 +2Y(\mu_1 - \mu_2)\right]\\
		&= \beta_2^\top \C A^\top \Sigma^{-1}A\C\beta_2 - \beta_1^\top \C A^\top \Sigma^{-1}A\C\beta_1 \\
		&\qquad +2\EE_{\beta_1}\left[Y(\beta_1-\beta_2)^\top \C A^\top \Sigma^{-1}X\right]\\
		&=\beta_2^\top \C A^\top \Sigma^{-1}A\C\beta_2 - \beta_1^\top \C A^\top \Sigma^{-1}A\C\beta_1\\
		&\qquad+ 2\EE\left[(\beta_1-\beta_2)^\top \C A^\top \Sigma^{-1}AZZ^\top \beta_1\right]\\
		&= (\beta_1-\beta_2)^\top \C A^\top \Sigma^{-1}A\C(\beta_1-\beta_2),
		\end{align*}
		where in the fourth line we use model (\ref{model}).
		Plugging this into the KL-divergence yields 
		\begin{align*}
		{1\over n}\text{KL}(\PP_{\beta_1}, \PP_{\beta_2}) &= {1\over 2}(\log \sigma_2^2 - \log \sigma_1^2) + {\sigma_1^2- \sigma_2^2\over 2\sigma_2^2}\\
		&\qquad+ {1\over 2\sigma_2^2}(\beta_1-\beta_2)^\top \C A^\top \Sigma^{-1}A\C(\beta_1-\beta_2).
		\end{align*}
		Note that $\C-\C A^\top \Sigma^{-1}A\C = G^{-1}$ from Fact \ref{fact_1}.
		Recalling that $\sigma_j^2 = \sigma^2 + \beta_j^\top G^{-1}\beta_j$ from (\ref{disp_cond_y|x}) and Fact \ref{fact_1} and using the inequality 
		\[
		\left|\log (\sigma^2 + t_2) - \log (\sigma^2 + t_1)\right| \le {|t_2 - t_1| \over \min (\sigma^2 + t_1, \sigma^2 + t_2)}
		\]
		for $t_1, t_2>0$ gives 
		\begin{align*}
		{1\over n}\text{KL}(\PP_{\beta_1}, \PP_{\beta_2}) \le {|\beta_1^\top G^{-1}\beta_1 - \beta_2^\top G^{-1}\beta_2|\over \min(\sigma_1^2, \sigma_2^2)}+ {(\beta_1-\beta_2)^\top (\C-G^{-1})(\beta_1-\beta_2)\over 2\sigma_2^2}.
		\end{align*}
		Using $\sigma_j^2 \ge \sigma^2 + \|\beta_j\|^2\lambda_{\min}(G^{-1}) =  \sigma^2 + \|\beta_j\|^2/\|G\|_{{\rm op}}$ completes the proof.
	\end{proof}
	
	\bigskip
	
	The following fact is used in the proof of Lemma \ref{lem_KL}.
	
	\begin{fact}\label{fact_1}
		Let $\Sigma = A\C A^\top +\Gamma$, $G = \Omega+A^\top \Gamma^{-1} A$ with $\Omega = \C^{-1}$.  
		\[
		\Sigma^{-1}A\C = \Gamma^{-1} AG^{-1},\quad \C-\C A^\top \Sigma^{-1}A\C = G^{-1}.
		\]
	\end{fact}
	\begin{proof}
		The Sherman-Morrison-Woodbury formula gives 
		\begin{align*}
		\Sigma^{-1}A\C &= \left[\Gamma^{-1}  - \Gamma^{-1} A(\Omega + A^\top \Gamma^{-1} A)^{-1}A^\top \Gamma^{-1} \right]A\C\\
		&=\Gamma^{-1} A\C - \Gamma^{-1} A(\Omega+ A^\top \Gamma^{-1} A)^{-1}(\Omega+A^\top \Gamma^{-1} A-\Omega)\C\\
		&= \Gamma^{-1} A(\Omega+A^\top \Gamma^{-1} A)^{-1},
		\end{align*}
		which concludes the proof of the first statement. The second part follows immediately by noting that
		\[
		\C -\C A^\top \Sigma^{-1}A\C = \C - \C A^\top \Gamma^{-1} A(\Omega+A^\top \Gamma^{-1} A)^{-1} = (\Omega+A^\top \Gamma^{-1} A)^{-1}.
		\]
	\end{proof}

	\section{Preliminaries}
	\label{sec_auxiliary_lemma}
	\subsection{General principles}
	
	Throughout the proofs of the main results, we will work on the event    
	\begin{equation}\label{def_event}
	\E:= \left\{\max_{1\le k\le K}{1\over n}\sum_{t=1}^n \Z_{tk}^2 \le B_z
	\right\}\bigcap \left\{\max_{1 \le j< \ell \le p}|\wh \Sigma_{j\ell} - \Sigma_{j\ell}| \le \delta\right\}
	\end{equation}
	with $B_z$ defined in Assumption \ref{ass_subg} and $\delta := c\delta_n$ some constant $c>0$ and 
	\[
	\delta_n = \sqrt{\log\pn / n}.
	\]
	Provided $\log p \lesssim n$,  
	Lemma \ref{lem_Z} and Lemma \ref{lem_bernstein} below guarantee that $\PP(\E)\ge 1-\pn^{-\alpha}$ for some constant $\alpha>0$. This event plays an important role.   The proof of Theorem 3 in  \cite{LOVE} reveals   that  on this event $\E$, we have
	\begin{itemize}
		\item 
		$\wh K=K$, 
		\item
		$I_k\subseteq \wh I_{\pi(k)} \subseteq I_k\cup J_1^k$ with $J_1^k=\{ j\in J: \ |A_{jk}| \ge 1-4\delta/v\}$, for all $k\in [K],$
	\end{itemize}
	for some permutation $\pi: [K]\to [K]$.
	{\em To lighten the notation, we assume throughout the remainder of the appendix that $\pi$ is the identity group permutation     and $A_{ik} = 1$ for any $i\in I_k$ and $k\in [K]$.} In the general case, the signed permutation matrix $P$ can be traced throughout the proofs.\\
	
	This means in particular that  the dimensions of the parameter $\beta$ and its estimate $\wh \beta$ are the same. Unfortunately, $\wh I$ is only guaranteed to satisfy $\wh I \supseteq  I$. This is the main source of many technical challenges in the proofs.
	We emphasize that signal conditions on $A$
	could prevent this, and lead to guarantees of $\wh I=I$. However, an assumption that
	the  entries of $A$ are either zero or exceed a certain threshold in absolute value,
	is rather unnatural and instead  we only rely on a mild separation condition on the matrix $\C$ in Assumptions \ref{ass_model} and \ref{ass_C}.
	\\

	The next subsections contain notation and some auxiliary results, that may be skipped during the first reading of this manuscript.

	\subsection{Notation}
	We write 
	\begin{equation}
	\wh \Pi = \wh A_{{\wh I\sbt}}(\wh A_{{\wh I\sbt}}^\top \wh A_{{\wh I\sbt}})^{-1}    
	\end{equation}
	and 
	\begin{equation}\label{def_tildes}
	\wt \X := \X_{\sbt\wh I}\wh \Pi, \qquad \wt \Z :=  \Z A_{{\wh I\sbt}}^\top \wh \Pi,\qquad \wt \W:= \W_{\sbt\wh I}\wh \Pi,
	\end{equation}
	so that $\wt \X = \wt \Z + \wt \W$. Similarly, we write 
	\begin{equation}
	\Pi  = A_{I\sbt}(A_{I\sbt}^\top A_{I\sbt})^{-1}
	\end{equation}  and 
	\begin{equation}\label{def_bars}
	\oX := \X_{\sbt I}\Pi , \qquad \oW:= \W_{\sbt I} \Pi ,
	\end{equation}
	so that $\oX = \Z + \oW$. 
	For all $k\in [K]$, set
	\[
	m_k := |I_k|\quad \text{  and }\quad m = \min_{k\in [K]} m_k. 
	\]
	On the event $\E$,
	we further define $L_k := \wh I_k \setminus I_k \subseteq J_1^k$ so that $$\wh m_k := |\wh I_k| = m_k + |L_k|,$$ and 
	\begin{equation}\label{def_rho_k}
	D_{\rho} = \textrm{diag}(\rho_1, \ldots, \rho_K),\quad \text{with}\quad 
	\rho _k = {|L_k| \over \wh m_k}.
	\end{equation}
	On the event $\E$, we have $| L_k|\le | J_1^k|$, and
	we will frequently make use of the inequality
	\begin{equation}\label{disp_rho_rho_bar}
	\|\rho\|_2^2 
	= \sum_{k=1}^K 	\left( \frac{| L_k|}{ m_k + |L_k|} \right)^2 
	\le \sum_{k=1}^K 	\left( \frac{| J_1^k|}{ m_k + |J_1^k|} \right)^2
	:= \bar{\rho}^2.
	\end{equation}
	This inequality uses the fact that the function $x/(m_k+x)$ is increasing for $x\ge0$. 
	Finally, the   inequality  
	\begin{equation}\label{rate_Ajek}
	\max_{k\in [K],\ j\in J_1^k} \|A_{j\sbt} - \e_k\|_1\le {8\delta / \nu} \lesssim \delta_n
	\end{equation}
	is a direct consequence of  the definition of $J_1^k$, and will be repeatedly used in our proofs.  
	Here $\nu$ is the constant  defined in Assumption \ref{ass_model}. 
	Recall $\Theta = A\C$ and define the matrices  $$\Omega:=\C^{-1} \quad \text{ and } \quad H:= A(\C)^{1/2} = \Theta (\C)^{-1/2} = \Theta \Omega^{1/2}$$ 
	and $H^+ = (H^\top H)^{-1}H^\top $. 
	Further, on the event $\E$  to ensure that $\wh K = K$, write
	\begin{equation}\label{def_H_hat}
	\wh H := \wh \Theta \Omega^{1/2}
	\end{equation}
	with $\wh\Theta$ defined in (\ref{est_Theta}), and
	\begin{equation}\label{def_deltaA_I}
	\D  = A_{{\wh I\sbt}}^\top  \wh A_{{\wh I\sbt}}(\wh A_{{\wh I\sbt}}^\top \wh A_{{\wh I\sbt}})^{-1} - {\bI}_K =A_{{\wh I\sbt}}^\top  \wh \Pi - {\bI}_K.
	\end{equation}

	
	\bigskip
	
	\subsection{Preliminary lemmas}\label{sec_proof_prelim_lemmas}
	We now state three preliminary lemmas that will be used throughout the proofs that follow. Their proofs are relegated to Section \ref{sec_proof_lem_prelim}. 
	
	\begin{lemma}\label{lem_Z}
		Suppose Assumption \ref{ass_subg} holds.
		\begin{enumerate}
			\item[(1)] For any fixed $v\in \R^K$ and $t\in [n]$, $\langle \Z_{t\sbt}, v \rangle$ is $\z\sqrt{v^\top \C v}$--sub-Gaussian. In particular, $\Z_{tk}$ is $\z'$-sub-Gaussian with $\z' := \z \sqrt{B_z}$ for any $k\in [K]$. 
			\item[(2)] Provided  Assumption \ref{ass_model} holds as well,  $\X_{tj}$ is $(\z' + \w)$-sub-Gaussian for any  $t\in [n]$ and $j\in[p]$.
		\end{enumerate}
	\end{lemma}

	\begin{lemma}\label{lem_W}
		Suppose Assumption \ref{ass_subg} holds. Let $v\in \R^K$ and $\alpha\in \R^p$ be any fixed vectors. Then, for any $t\in [n]$, 
		\begin{itemize}
			\item[(1)]  $\alpha^\top  \W_{t\sbt}$ is $\w \|\alpha\|_2$-sub-Gaussian;
			\item[(2)]  $v^\top H^+ \W_{t\sbt}$ is $\w\sqrt{v^\top (H^\top H)^{-1}v}$-sub-Gaussian;
			\item[(3)] $\oW_{tk}$ is $\w/\sqrt{m}$-sub-Gaussian for any $k\in [K]$. 
		\end{itemize}
	\end{lemma}

	\begin{lemma}\label{lem_bernstein}
		Let $\{X_t\}_{t=1}^n$ and $\{Y_t\}_{t=1}^n$ be any two sequences, each with zero mean independent $\gamma_x$-sub-Gaussian and $\gamma_y$-sub-Gaussian elements. Then, for some constants $c, c'>0$, we have 
		\[
		\PP\left\{{1\over n}\left|\sum_{t=1}^n\left(X_t Y_t - \EE[X_t Y_t]\right)\right| \le c\gamma_x \gamma_y t \right\}\ge 1-2\exp\left\{-c'\min\left( t^2,t \right)n\right\}.
		\]
		In particular, when $\log p \le c''n$ for some constant $c''>0$, one has
		\[
		\PP\left\{{1\over n}\left|\sum_{t=1}^n\left(X_tY_t - \EE[X_tY_t]\right)\right| \le c\gamma_x \gamma_y \sqrt{\log\pn \over n} \right\}\ge 1-2\pn^{-c'}.
		\]
	\end{lemma}
	
	Finally we collect several results that control different random quantities of   interest. 
	
	\begin{lemma}\label{lem_quad}
		Let $u, v \in \R^K$ and $\alpha \in \R^p$ be fixed vectors. Under Assumption \ref{ass_subg}, each of the following statements holds with probability $1-\pn^{-c}$ for some constant $c>0$, 
		\begin{align}
		\label{rate_zeps}
		&	{1\over n}|u^\top \Z^\top \eps| \lesssim \delta_n\sqrt{u^\top \C u} ,\\ \label{rate_weps}
		&	{1\over n}|\alpha^\top \W^\top \eps| \lesssim  \delta_n\|\alpha\|_2 ,\\\label{rate_wz}
		&		{1\over n}|\alpha^\top \W^\top \Z u| \lesssim \delta_n \|\alpha\|_2 \sqrt{u^\top \C u},\\\label{rate_hweps}
		&	{1\over n}|u^\top H^+ \W^\top \eps| \lesssim \delta_n \sqrt{u^\top (H^\top H)^{-1}u}, \\\label{rate_hwz}
		&	{1\over n}|u^\top H^+ \W^\top \Z v| \lesssim  \delta_n \sqrt{v^\top \C v}\sqrt{u^\top (H^\top H)^{-1}u},\\\label{rate_zz_c}
		& \left| u^\top \left({1\over n}\Z^\top \Z - \C\right) v\right| \lesssim \delta_n\sqrt{u^\top \C u}\sqrt{v^\top \C v} ,\\ \label{rate_zz}
		& {1\over n}|u^\top \Z^\top \Z v| \lesssim |u^\top \C v| + \delta_n\sqrt{u^\top  \C u}\sqrt{v^\top \C v}.
		\end{align} 
		If, in addition, Assumption \ref{ass_model} holds,  we have, 
		on an event that holds with probability $1-C\pn^{-\alpha}$ for some constants $C,\alpha>0$, $\wh K=K$ and each of the following inequalities holds,
		\begin{align}\label{rate_zw_td}
		&{1\over n}|u^\top \Z^\top \wt \W v| \lesssim \delta_n \|v\|_2\sqrt{u^\top \C u} \left(
		{1 \over \sqrt m} + \bar{\rho}
		\right)\\\label{rate_hww}
		&\left|u^\top H^+\left(
		{1\over n}\W^\top \wt \W - \wh \Gamma_{{\sbt \wh I}}\wh \Pi\right)v\right| \lesssim  \delta_n {\|u\|_2\|v\|_2 \over \sigma_K(H)}  \left(
		{1 \over \sqrt{m}} + \bar{\rho}
		\right)\\\label{rate_ww}
		&\left|\alpha^\top \left(
		{1\over n}\W^\top \wt \W - \wh \Gamma_{{\sbt \wh I}}\wh \Pi\right)v\right| \lesssim  \delta_n \|\alpha\|_2\|v\|_2 \left(
		{1 \over \sqrt{m}} + \bar{\rho}
		\right).
		\end{align}
	\end{lemma}
	\begin{proof}
		This lemma is proved in Section \ref{sec_proof_lem_quad}.
		We emphasize that the expressions on the left hand side of (\ref{rate_zw_td}), (\ref{rate_hww}) and (\ref{rate_ww}) are well defined on the event  $\wh K=K$.
	\end{proof}

	\begin{lemma}\label{lem_op}
		Under Assumption \ref{ass_subg}, 
		on an event that holds with probability $1-C\pn^{-K\alpha}$ for some constants $C,\alpha>0$, $\wh K=K$ as well as
		\begin{align}\label{rate_op_hwz}
		& {1\over n}\|H^+ \W^\top \Z \Omega^{1/2}\|_{{\rm op}} \lesssim  {\delta_n\over \sigma_K(H)} \sqrt{K},\\\label{rate_op_zz_c}
		& \left\|\Omega^{1/2}\left({1\over n}\Z^\top \Z - \C \right) \Omega^{1/2}\right\|_{{\rm op}} \lesssim \delta_n\sqrt{K},\\\label{rate_op_zwtd}
		& {1\over n}\left\|\Omega^{1/2}\Z^\top \wt \W\right\|_{{\rm op}} \lesssim \left({1\over \sqrt m} + \bar{\rho}\right)\delta_n  \sqrt{K},\\\label{rate_op_hwwtd}
		& \left\| H^+ \left({1\over n}\W^\top \wt \W - \wh \Gamma_{{\sbt \wh I}}\wh \Pi\right)
		\right\|_{{\rm op}} \lesssim {1\over \sigma_K(H)} \left({1\over \sqrt m} + \bar{\rho}\right)\delta_n  \sqrt{K}.
		\end{align}
	\end{lemma}
	\begin{proof}
		This lemma is proved in Section \ref{sec_proof_lem_op}.
	\end{proof}

	\begin{lemma}\label{lem_deltaA_I}
		Under Assumptions \ref{ass_model} and \ref{ass_subg}, we have, for any fixed $v\in \R^K$,
		on an event that holds with probability at least $1-C\pn^{-\alpha}$ for some constants $C,\alpha>0$, $\wh K=K$ and
		\[
		\|\D v\|_1 \lesssim       \|v\|_2 \bar\rho \delta_n.			
		\]
		Furthermore, for any $K\times K$-matrix $Q$,  on an event that holds with probability $1-C\pn^{-\alpha}$ for some constants $C,\alpha>0$, $\wh K=K$ and
		\[ 
		\|Q\D\|_{{\rm op}} \le \|Q\|_{2,\i} \cdot   \bar{\rho} \cdot \delta_n.
		\]
	\end{lemma}
	\begin{proof}
		This lemma is proved in Section \ref{sec_proof_lem_deltaA_I}.
	\end{proof}

	\begin{lemma}\label{lem_H_op}
		Suppose Assumptions \ref{ass_model} and \ref{ass_subg} hold, and that there exists some sufficiently small constant $c_0>0$ such that 
		\begin{equation}\label{cond_1}
		\delta_n\sqrt{K}  \cdot T_n\le c_0
		\end{equation}
		with 
		\begin{equation}\label{def_Tn}
		T_n := 
		1 +\left. \left({1\over \sqrt{m}} + {\bar{\rho}} + {\bar{\rho} \delta_n}\right) \right/ \sqrt{\cl}
		\end{equation}
		Then,  for some constants $0<c_1<1  $ and $c_2>0$,  on an event that holds with probability at least $1-C\pn^{-\alpha}$ for some constants $C,\alpha>0$, $\wh K=K$ and the following inequalities hold
		\begin{enumerate}
			\item[(a)] $\|H^+(\wh H - H)\|_{{\rm op}} \le c'\delta_n T_n\sqrt{K} \le c_1$;
			\item[(b)] $\lambda_K(\wh H^\top  \wh H ) \ge (1-c_1)^2 \lambda_K\left( H^\top  H\right);$
			\item[(c)] $
			\lambda_K(\wh \Theta^\top \wh \Theta )  \ge (1-c_1)^2\cl\lambda_K(H^\top H);
			$
			\item[(d)] $|u^\top(\wh H - H)v| \lesssim \|u\|_2\|v\|_2 T_n \delta_n$ for any fixed $u\in \R^p$ and $v\in \R^K$;
			\item[(e)] $\|\wh H - H\|_{{\rm op}} \le c_2  T_n\delta_n\sqrt{pK}$.
		\end{enumerate}
	\end{lemma}
	\begin{proof}
		This lemma is proved in Section \ref{sec_proof_lem_H_op}.
	\end{proof}

	\subsection{Proofs of the preliminary lemmas}

	\subsubsection{Proof of lemmas \ref{lem_Z}, \ref{lem_W} and \ref{lem_bernstein}}\label{sec_proof_lem_prelim}
	
	\begin{proof}[Proof of Lemma \ref{lem_Z}]
		For any $t\in [n]$,  Assumption \ref{ass_subg} implies that
		\begin{align*}
		\EE\left[\exp\left(\lambda\left\langle \Z_{t\sbt}, v\right\rangle\right)\right] &= \EE\left[\exp\left(\lambda\langle \C^{-1/2}\Z_{t\sbt}, \C^{1/2}v\rangle\right)\right]\\
		& \le \EE\left[\exp\left(\lambda^2\z^2v^\top \C v/2\right)\right], \qquad \forall \ \lambda\in\R.
		\end{align*}
		This proves the statement of $\langle \Z_{t\sbt}, v \rangle$. Taking $v = \e_k$ and using $\|\C\|_\i \le B_z$ concludes the proof of (1). Part (1) and the independence between $Z$ and $W$ yield
		\begin{align*}
		\EE[\exp(\lambda \X_{tj})] &= \EE\left[\exp(\lambda \langle A_{j\sbt}, \Z_{t\sbt}\rangle) \right]\EE\left[\exp(\lambda \W_{tj}) \right]\\
		&\le \exp\left(\lambda^2\z^2 A_{j\sbt}^\top \C A_{j\sbt} / 2\right)\exp\left(\lambda^2\w^2/ 2\right)\\
		&\le \exp\left\{{\lambda^2(\z^2\|\C\|_\i+\w^2)/ 2}\right\}\qquad (\|A_{j\sbt}\|_1\le 1)
		\end{align*}
		for any $\lambda\in \R$ and any $1\le j\le p$. This completes the proof.
	\end{proof} 
	
	\bigskip
	
	\begin{proof}[Proof of Lemma \ref{lem_W}]
		By Assumption \ref{ass_subg}, $\W_{tj}$ is $\w$-sub-Gaussian for all $t\in[n]$ and $j\in[p]$, so
		\[
		\EE[\exp(\lambda \W_{ti})] \le \exp(\lambda^2\w^2/2),\quad \forall \lambda\in \R.
		\]
		Again by Assumption \ref{ass_subg}, the $\W_{ti}$ are independent across index $i$, whence
		\begin{align*}
		\EE[\exp(t\alpha^\top W_{t\sbt})] &= \prod_{j=1}^p\EE\left[\exp(t \alpha_{j  }  \W_{tj})\right] \\
		&\le  \prod_{j=1}^p\exp\left[t^2\w^2\alpha_j^2/2\right]\\
		&\le \exp\left[t^2\w^2\|\alpha\|_2^2/2\right],
		\end{align*}
		proving that
		$\sum_{j}\alpha_j \W_{tj}$ is $\w\|\alpha\|_2$-sub-Gaussian. The second statement follows by taking $\alpha = (H^+)^\top v$. Similarly, 
		\begin{align*}
		\EE[\exp(t\oW_{tk})] &= \prod_{i\in I_k}\EE\left[\exp(t\W_{ti}/m_k)\right] \\
		&\le \prod_{i\in I_k}\exp\left[t^2\w^2/( 2m_k^2)\right]\\
		&\le \exp\left[t^2\w^2/(2m)\right],
		\end{align*}
		proving that $\oW_{tk}$ is $\gamma_w/\sqrt{m}$-sub-Gaussian. 
		This concludes the proof. 
	\end{proof} 
	
	\bigskip
	
	\begin{proof}[Proof of Lemma \ref{lem_bernstein}]
		Let $\|\cdot\|_{\psi_1}$ and $\|\cdot\|_{\psi_2}$ denote, respectively, the sub-exponential norm and the sub-gaussian norm (Definitions 5.7 and 5.13 in \cite{vershynin_2012}) as
		\[
		\|X\|_{\psi_1} = \sup_{p\ge 1}p^{-1}\left(\EE[|X|^p]\right)^{1/p}, \quad   \|X\|_{\psi_2} = \sup_{p\ge 1}p^{-1/2}\left(\EE[|X|^p]\right)^{1/p}
		\]
		for any random variable $X$. 
		Then, $\|X_t\|_{\psi_2} \le c\gamma_x$ and $\|Y_t\|_{\psi_2} \le c\gamma_y$ and an application of the H\"{o}lder's inequality yields $\|X_t Y_t\|_{\psi_1}\le \|X_t\|_{\psi_2} \|Y_t\|_{\psi_2} \le c\gamma_x\gamma_y$. The proof of the first statement follows by Corollary 5.17 in \cite{vershynin_2012}. The second statement follows by taking $t = \max(\delta_n, \delta_n^2)$ and observing that $\min(t^2, t)n = \log\pn$.
	\end{proof}

	\subsubsection{Proof of Lemma \ref{lem_quad}}\label{sec_proof_lem_quad}
	
	The proofs of (\ref{rate_zeps}) -- (\ref{rate_zz}) are all based on applying Lemma \ref{lem_bernstein} in conjunction with Assumption \ref{ass_subg}, Lemmas \ref{lem_Z} and \ref{lem_W} and the fact that $\EE[\Z_{t\sbt}\Z_{t\sbt}^\top ] = \C$ for all $t\in [n]$. 
	
	We work on the event $\E$ defined in (\ref{def_event}) that has probability at least $1-\pn^{-c}$ for some $c>0$. Recall that it implies $\wh K = K$ and $I_k \subseteq \wh I_k \subseteq I_k \cup J_1^k$ for all $k\in [K]$.	
	
	To prove (\ref{rate_zw_td}), since
	\[
	\wt \W_{\sbt k} = \oW_{\sbt k} + {|L_k| \over \wh m_k}\left(
	{1\over |L_k|}\sum_{i\in L_k}\W_{\sbt i} - \oW_{\sbt k}
	\right)
	\]
	from (\ref{def_tildes}) and (\ref{def_bars}), we have
	\begin{equation}\label{eq_diff_W_td_W_bar}
	\wt \W -\oW = \W_{\sbt L}\wh A_{L\sbt}D_{\wh m} - \oW D_{\rho}
	\end{equation}
	where $D_{\wh m} = \textrm{diag}(1/\wh m_1, \ldots, 1/\wh m_K)$ and $D_\rho$ is defined in (\ref{def_rho_k}). 
	It then follows that
	\begin{align}\nonumber
	&{1\over n}|u^\top  \Z^\top  \wt \W v|\\\nonumber
	&\le 	{1\over n}|u^\top  \Z^\top  \oW v|+	{1\over n}|u^\top \Z^\top \W_{\sbt L}\wh A_{L\sbt}D_{\wh m} v|  + 	{1\over n}|u^\top \Z^\top \oW D_{\rho}v|\\\label{disp_zw_td}
	&\le {1\over n}|u^\top  \Z^\top  \oW v|+	\max_{i \in J_1}{1\over n}|u^\top \Z^\top \W_{\sbt i}|\cdot \|\wh A_{L\sbt}D_{\wh m} v\|_1  + \max_k{1\over n}|u^\top \Z^\top \oW_{\sbt k}|\cdot  \|D_{\rho}v\|_1.
	\end{align}	
	We find that 
	\begin{equation}\label{disp_D_rho}
	\|\wh A_{L\sbt}D_{\wh m}v\|_1\le \bar\rho \|v\|_2 \le \|v\|_2 \|\rho\|_2 \overset{(\ref{disp_rho_rho_bar})}{\le} \|v\|_2 \bar{\rho},
	\end{equation}
	after we apply   Lemmas \ref{lem_Z}, \ref{lem_W} and \ref{lem_bernstein} and take a union bound. This concludes the proof of (\ref{rate_zw_td}).

	To prove the results of (\ref{rate_hww}) and (\ref{rate_ww}), we first have 
	\begin{equation}\label{disp_WWtd_Sigma_1}
	{1\over n}\W^\top \wt \W - \wh \Gamma_{{\sbt \wh I}}\wh \Pi = {1\over n}\W^\top (\wt \W-\oW) + {1\over n}\W^\top \oW - \Gamma_{\sbt I}\Pi +  \Gamma_{\sbt I}\Pi - \wh \Gamma_{{\sbt \wh I}}\wh \Pi.
	\end{equation}
	Then by plugging (\ref{eq_diff_W_td_W_bar}) into (\ref{disp_WWtd_Sigma_1}), we obtain
	\begin{align}\label{disp_WWtd_Sigma_2}
	{1\over n}\W^\top \wt \W - \wh \Gamma_{{\sbt \wh I}}\wh \Pi &= \left({1\over n}\W^\top \W_{\sbt L} - \Gamma_{\sbt L}\right)\wh A_{L\sbt}D_{\wh m} - \left({1\over n}\W^\top \oW- \Gamma_{\sbt I}\Pi\right)D_{\rho}\\\nonumber
	& \quad + {1\over n}\W^\top \oW - \Gamma_{\sbt I}\Pi +  \D_d
	\end{align}
	where 
	\begin{equation}\label{def_D_d}
	\D_d = \Gamma_{\sbt I}\Pi - \wh \Gamma_{{\sbt \wh I}}\wh \Pi +\Gamma_{\sbt L}\wh A_{L\sbt}D_{\wh m} -\Gamma_{\sbt I}\Pi D_{\rho}.
	\end{equation}
	We first prove (\ref{rate_ww}) and fix any $\alpha \in \R^p$ and $v\in \R^K$. From display (\ref{disp_WWtd_Sigma_2}), we only need to upper bound 
	\begin{align*}
	&\left\|\alpha^\top \left({1\over n}\W^\top \W_{\sbt L} - \Gamma_{\sbt L}\right)\right\|_\i\|\wh A_{L\sbt}D_{\wh m}v\|_1 +	\left\|\alpha^\top \left({1\over n}\W^\top \oW- \Gamma_{\sbt I}\Pi\right)\right\|_\i \|D_{\rho}v\|_1\\\nonumber
	& \quad + 	\left|\alpha^\top \left({1\over n}\W^\top \oW - \Gamma_{\sbt I}\Pi\right)v\right| + |\alpha^\top \D_d v|.
	\end{align*}
	Invoke Lemmas \ref{lem_W} and \ref{lem_bernstein} and apply a union bound to derive
	\begin{align*}
	&\max_{i\in J_1} \left|\alpha^\top \left({1\over n}\W^\top \W_{\sbt i} - \Gamma_{\sbt i}\right)\right| \lesssim \|\alpha\|_2 \delta_n,\\
	&\left|\alpha^\top \left({1\over n}\W^\top \oW - \Gamma_{\sbt I}\Pi\right)v\right| \lesssim \|\alpha\|_2\|v\|_2 \delta_n / \sqrt{m}
	\end{align*} 
	Each inequality holds with probability $1-\pn^{-c}$. By further using (\ref{disp_D_rho}), we conclude that, with probability $1-c'\pn^{-c}$, 
	\begin{align}\label{rate_WWL}
	&\left\|\alpha^\top \left({1\over n}\W^\top \W_{\sbt L} - \Gamma_{\sbt L}\right)\right\|_\i\|\wh A_{L\sbt}D_{\wh m}v\|_1  \lesssim  \bar\rho \|v\|_2 \|\alpha\|_2 \delta_n,\\\label{rate_WWD}
	&\left\|\alpha^\top \left({1\over n}\W^\top \oW- \Gamma_{\sbt I}\Pi\right)\right\|_\i \|D_{\rho}v\|_1\lesssim  \bar\rho \|v\|_2 \|\alpha\|_2 \delta_n /\sqrt{m},\\\label{rate_WW}
	&	\left|\alpha^\top \left({1\over n}\W^\top \oW - \Gamma_{\sbt I}\Pi\right)v\right| \lesssim \|\alpha\|_2\|v\|_2 \delta_n / \sqrt{m}.
	\end{align}
	We upper bound the remaining term $|\alpha^\top \D_d v|$. First note that 
	\begin{align*}
	[\D_d]_{ik} = {\Gamma_{ii} - \wh \Gamma_{ii} \over \wh m_k},\quad \forall i\in \wh I_k, \qquad [\D_d]_{ik} = 0, \text{ otherwise }.
	\end{align*}
	Furthermore, since $\wh I \subseteq I \cup J_1$ on the event $\E$, we have
	\[
	\|\D_dv\|_2^2 = \sum_{k=1}^K\sum_{i \in \wh I_k} {v_k^2 \over \wh m_k^2}\left(\wh \Gamma_{ii} - \Gamma_{ii}\right)^2 \le {\|v\|_2^2 \over m}\max_{i\in I\cup J_1} \left(\wh \Gamma_{ii} - \Gamma_{ii}\right)^2.
	\]
	We thus obtain 
	\begin{align}\label{rate_D_d}
	|\alpha^\top \D_d v| &\le \|\alpha\|_2 \|\D_dv\|_2 \lesssim \|\alpha\|_2 \|v\|_2 \delta_n /\sqrt{m}
	\end{align}
	on the event $\E$ intersected with
	\begin{equation}\label{def_event_E_w}
	\E_W := \left\{\max_{i\in I\cup J_1} \left|\wh \Gamma_{ii} - \Gamma_{ii}\right| \lesssim \delta_n\right\}.
	\end{equation}
	Since $\E_W$ holds with probability $1-\pn^{-c}$ by Lemma \ref{lem_tau}, in conjunction with (\ref{rate_WWL}) -- (\ref{rate_WW}), we conclude (\ref{rate_ww}).  The result of (\ref{rate_hww}) follows immediately by taking $\alpha = u^\top H^+$ and noting that $\|\alpha\|_2 \le \|u\|_2/\sigma_K(H)$.\qed
	
	\subsubsection{Proof of Lemma \ref{lem_op}}\label{sec_proof_lem_op}
	
	All results (\ref{rate_op_hwz}) -- (\ref{rate_op_hwwtd}) can be proved based on Lemma \ref{lem_quad} together with a classical discretization method to prove the uniformity over the unit sphere. We only prove display  (\ref{rate_op_zwtd}) since  the other results can be shown by   similar arguments.   
	
	Again, we work on the event $\E$ defined in (\ref{def_event}) that has probability at least $1-\pn^{-c}$ for some $c>0$. Recall that it implies $\wh K = K$ and $I_k \subseteq \wh I_k \subseteq I_k \cup J_1^k$ for all $k\in [K]$. 
	
	Let $\mathcal{N}_\eps\subset \mathcal{S}^{K-1}$ be a minimal $\eps$-net of $\mathcal{S}^{K-1}$, i.e. a set with minimum cardinality such that the collection of $\eps$-balls centered at points in $\mathcal{N}_\eps$ covers $\mathcal{S}^{K-1}$. 
	
	To prove (\ref{rate_op_zwtd}), a standard discretization argument, for instance, see \citep[Proof of Proposition 2.4]{rv2010} gives, on the event $\E$, 
	\[
	{1\over n}\|\Omega^{1/2}\Z^\top \wt \W\|_{{\rm op}} = \sup_{u, v\in \S^{K-1}} {1\over n}|u^\top \Omega^{1/2}\Z^\top \wt \W v| \le 4 \max_{u, v\in \N_{1/2}} {1\over n}|u^\top \Omega^{1/2}\Z^\top \wt \W v|.
	\]
	Recall that   (\ref{disp_zw_td}) and (\ref{disp_D_rho}) imply that,  for any $u, v\in \S^{K-1}$,
	\begin{align*}
	{1\over n}|u^\top \Omega^{1/2}\Z^\top \wt \W v| &\le {1\over n}|u^\top  \Omega^{1/2}\Z^\top  \oW v|+	\max_{i \in J_1}{1\over n}|u^\top \Omega^{1/2}\Z^\top \W_{\sbt i}|\cdot \bar\rho \|v\|_2\\
	&\quad   + \max_{1\le k\le K}{1\over n}|u^\top \Omega^{1/2}\Z^\top \oW_{\sbt k}|\cdot  \bar\rho \|v\|_2,
	\end{align*}
	on the event $\E$.
	Use the classical bound 
	$|\mathcal{N}_{1/2}| \le 5^K$ by Lemma 5.2 in \cite{vershynin_2012},  apply Lemma \ref{lem_W} and Lemma \ref{lem_bernstein} with $t = \delta_n\sqrt{K}$ for each random term in the above display and take the union bound over $u, v\in \N_{1/2}$, $i\in J_1$ and $k\in [K]$, and use the restriction
	$K\log\pn \le cn$
	to obtain 
	\[
	\max_{u, v\in \N_{1/2}}{1\over n}|u^\top \Omega^{1/2}\Z^\top \wt \W v| \lesssim 
	\delta_n \sqrt{K}\left(
	{1 \over \sqrt m} + \bar{\rho}
	\right) 
	\]
	with probability $1 - 4\cdot 5^{2K} \pn^{-c K}$.
	This finishes the proof of (\ref{rate_op_zwtd}).
	\qed
	
	\subsubsection{Proof of Lemma \ref{lem_deltaA_I}}\label{sec_proof_lem_deltaA_I}
	We  work on the event $\E$. Recall that it  $\wh K=K$, but also $\wh A_{ik}=A_{ik}$ for all $i\in I$ and $k\in [K]$ by Theorem \ref{thm_I}.
	Recall the definition  $\D = A_{{\wh I\sbt}}^\top  \wh A_{{\wh I\sbt}}(\wh A_{{\wh I\sbt}}^\top \wh A_{{\wh I\sbt}})^{-1} - {\bI}_K.$ For any fixed vector $v\in \R^K$, we have 
	\begin{align*}
	\|\D v\|_1 &= \sum_{k=1}^K \left|
	\sum_{a=1}^K{v_a \over \wh m_a}\sum_{i\in \wh I_a}(\wh A_{ik}-A_{ia})
	\right|\\
	&\le\sum_{k=1}^K 
	\sum_{a=1}^K{|v_a| \over \wh m_a}\sum_{i\in L_a}
	\left|\wh A_{ik}-A_{ia}
	\right|\\
	& = 
	\sum_{a=1}^K{|v_a| \over \wh m_a}\sum_{i\in L_a}
	\left\|\wh A_{i\sbt}-A_{i\sbt}
	\right\|_1\\
	&\le \max_{i \in J_1}	\left\|\wh A_{i\sbt}-A_{i\sbt}
	\right\|_1 \sum_{a=1}^K \rho_a |v_a|\\
	&\lesssim \delta_n \bar{\rho}\|v\|_2.
	\end{align*}
	We used $\wh I_a \subseteq I_a\cup L_a$ on the event $\E$  in the second line and   (\ref{def_rho_k}), (\ref{disp_rho_rho_bar}) and (\ref{rate_Ajek}) in the last line. This proves the first result. 
	
	To prove the second claim, we argue that for any matrix $Q\in \R^{d\times K}$, by using the dual norm inequality and the Cauchy-Schwarz inequality, we have
	\begin{align*}
	\|Q\D\|_{{\rm op}} &= \sup_{\alpha \in \S^{d-1}, u\in \S^{K-1}} \alpha^\top Q\D u\\
	& \le \sup_{\alpha \in \S^{d-1}, u\in \S^{K-1}}\| \alpha^\top Q\|_\i \|\D u\|_1\\
	&\lesssim \sup_{u\in \S^{K-1}}\|Q\|_{2,\i}\cdot  \delta_n \|D_\rho u\|_1.
	\end{align*}
	The result then follows from $\|D_\rho u\|_1 \le \|\rho\|_2  \le \bar{\rho}$ by using (\ref{disp_rho_rho_bar}). 
	\qed
	
	\subsubsection{Proof of Lemma \ref{lem_H_op}}\label{sec_proof_lem_H_op}
	We work on the event $\E$ defined in (\ref{def_event}) that has probability at least $1-\pn^{-c}$ for some $c>0$. Recall that it implies $\wh K = K$ and $I_k \subseteq \wh I_k \subseteq I_k \cup J_1^k$ for all $k\in [K]$. 
	
	We first prove $(a)$. The definition of $\wh \Theta$ in (\ref{est_Theta}) gives
	\begin{align}\nonumber
	\wh \Theta - \Theta &= {1\over n}\X^\top \wt \X - \wh \Gamma_{{\sbt \wh I}}\wh \Pi - A
	\C\\\label{disp_delta_theta}
	&= {1\over n}\X^\top \Z + {1\over n}\X^\top (\wt \X - \Z)- \wh \Gamma_{{\sbt \wh I}}\wh \Pi - A
	\C\\\nonumber
	&= A\left({1\over n}\Z^\top \Z  - \C \right) + A{1\over n}\Z^\top \wt \W + A{1\over n}\Z^\top \Z \D\\\nonumber
	&\quad + {1\over n}\W^\top  \wt \Z + {1\over n}\W^\top \wt \W - \wh \Gamma_{{\sbt \wh I}}\wh \Pi
	\end{align}
	where we use $\X = \Z A^\top + \W$ and the definition (\ref{def_deltaA_I}) in the third equality. By writing 
	\begin{align}\label{def_D_Z}
	\D_1 &:= \Omega^{1/2}\left({1\over n}\Z^\top \Z  -\C + {1\over n}\Z^\top \wt \W + {1\over n}\Z^\top \Z  \D\right)\Omega^{1/2},\\\label{def_D_W}
	\D_2 & := \left({1\over n}\W^\top  \wt \Z + {1\over n}\W^\top \wt \W - \wh \Gamma_{{\sbt \wh I}}\wh \Pi\right)\Omega^{1/2},
	\end{align}
	we have $(\wh \Theta - \Theta)\Omega^{1/2} = \wh H - H=  H\D_1 + \D_2$ such that 
	\[
	\|H^+(\wh H - H)\|_{{\rm op}} \le \|\D_1\|_{{\rm op}} + \| H^+\D_2\|_{{\rm op}}.
	\]
	By triangle inequality and the definition of operator norm, we can upper bound $\|\D_1\|_{{\rm op}}$ by 
	\begin{align*}
	&\left\|\Omega^{1/2}\left({1\over n}\Z^\top \Z  -\C\right)\Omega^{1/2}\right\|_{{\rm op}} + {1\over n}\|\Omega^{1/2}\Z^\top \wt \W\|_{{\rm op}}\|\Omega^{1/2}\|_{{\rm op}}+ {1\over n}\|\Omega^{1/2}\Z^\top \Z \D\|_{{\rm op}}\|\Omega^{1/2}\|_{{\rm op}}.
	\end{align*}
	By adding and subtracting $Z$, $\|H^+\D_2\|_{{\rm op}}$ is upper bounded by 
	\begin{align*}
	&{1\over n}\| H^+\W^\top \Z\Omega^{1/2}\|_{{\rm op}}+ {1\over n}\| H^+\W^\top \Z\D\|_{{\rm op}}\|\Omega^{1/2}\|_{{\rm op}}+ \left\| H^+\left({1\over n}\W^\top \wt \W - \wh \Gamma_{{\sbt \wh I}}\wh \Pi\right)\right\|_{{\rm op}}\|\Omega^{1/2}\|_{{\rm op}}.
	\end{align*}
	Note that Lemma \ref{lem_deltaA_I} gives
	\[
	{1\over n}\|\Omega^{1/2}\Z^\top \Z \D\|_{{\rm op}} \lesssim  {1\over n}\max_k\|\Omega^{1/2}\Z^\top \Z _k\|_2 \cdot  \bar{\rho} \delta_n\le \sqrt{K}{1\over n}\max_{k, k'} |\e_{k'}^\top \Omega^{1/2}\Z^\top \Z \e_k|\cdot \bar{\rho}\delta_n
	\]
	with probability $1-\pn^{-c}$, where in the second inequality we used $\|v\|_2\le \sqrt{K}\|v\|_\i$ for any $v\in \R^K$. And the term $H^+ \W^\top  \Z \D$ can be bounded by the same way. 
	By invoking results (\ref{rate_zz_c}), (\ref{rate_zw_td}), (\ref{rate_zz}), (\ref{rate_wz}), and (\ref{rate_ww}) in Lemma \ref{lem_op} and Lemma \ref{lem_deltaA_I} and taking the union bounds over $k,k' \in [K]$, one has 
	\begin{align}\label{rate_DZ_op}
	\|\D_1\|_{{\rm op}} &  \lesssim \delta_n\sqrt{K}\cdot T_n,\\\label{rate_HDW_op}
	\|H^+\D_2\|_{{\rm op}} &\lesssim  {\delta_n\sqrt{K} \over \sigma_K(H)}\cdot T_n
	\end{align}
	with probability $1-c\pn^{-c'}$ and $T_n$ defined in (\ref{def_Tn}).
	These two displays and condition (\ref{cond_1}) yield the result in $(a)$.

	For part $(b)$, to lower bound the $K$th eigenvalue of $\wh H^\top \wh H$, observe that 
	\begin{align*}
	\lambda_K\left(\wh H^\top \wh H\right) & \ge \lambda_K\left(\wh H^\top H(H^\top H)^{-1}H^\top \wh H\right)\\
	&\ge \lambda_K( H^\top   H)\cdot \lambda_K\left(
	\wh H^\top  H( H^\top   H)^{-2} H^\top \wh H\
	\right)\\
	&=  \lambda_K( H^\top   H)\cdot \sigma_K^2\left(
	( H^\top   H)^{-1} H^\top \wh H\
	\right).
	\end{align*}
	Since Weyl's inequality gives 
	\[
	\sigma_K\left(
	H^+ \wh H\
	\right) \ge 1 - \|H^+(\wh H- H)\|_{{\rm op}},
	\]
	invoking $(a)$ finishes the proof of $(b)$. Part $(c)$ follows immediately from part $(b)$ and  noting that 
	\[
	\lambda_K(\wh H^\top \wh H) = \lambda_K(\Omega^{1/2}\wh \Theta^\top  \wh\Theta \Omega^{1/2}) \le {\lambda_K(\wh \Theta^\top \wh\Theta) / \cl}.
	\]
	
	To prove $(d)$, starting from $\wh H - H= H\D_1 + \D_2$, it suffices to upper bound $|u^\top H\D_1 v|$ and $|u^\top\D_2v|$. For the first term, recalling that $H = A\C^{1/2}$, we have 
	\begin{align*}
	|u^\top H\D_1 v| &\le  \|u^\top A\|_1 \cdot \max_k 
	|\e_k^\top(\C)^{1/2} \D_Z v|\\
	&\le\|u^\top A\|_1 \cdot \max_{k} 
	\left|\e_k^T\left({1\over n}\Z^\top \Z  -\C + {1\over n}\Z^\top\wt \W + {1\over n}\Z^\top \Z  \D\right)\Omega^{1/2}v\right|
	\end{align*}
	where we used the definition of $\D_1$ in (\ref{def_D_Z}) in the third line.  Noting that $|\e_k^T\Z^\top \Z \D\Omega^{1/2}v| \le \max_a|\Z_{\sbt k}^T\Z_{\sbt a}| \|\D \Omega^{1/2}v\|_1$, invoking  (\ref{rate_zz_c}), (\ref{rate_zw_td}), (\ref{rate_zz}) together with Lemma \ref{lem_deltaA_I} gives 
	\[
	|u^\top H\D_1 v| \lesssim   \|u^\top A\|_1\cdot \delta_n \|v\|_2\cdot T_n
	\]
	with probability greater than $1-c\pn^{-c'}$.
	On the other hand, 
	by the definition of $\D_2$ in (\ref{def_D_W}) and using $\wt \Z = \Z + \Z\D$,  we can upper bound $|u^\top\D_2v| $ by 
	\[
	\left\{{1\over n}|u^\top \W^\top \Z\Omega^{1/2}v| + {1\over n}|u^\top \W^\top \Z\D\Omega^{1/2}v| + \left|u^T\left(
	{1\over n}\W^\top \wt \W - \wh \Gamma_{\sbt \wh I}\wh \Pi\right)\Omega^{1/2}v
	\right|\right\}.
	\]
	Invoking (\ref{rate_wz}) and (\ref{rate_ww}) in Lemma \ref{lem_quad} together with Lemma \ref{lem_deltaA_I} again, gives
	\[
	\|u^\top\D_2v| \lesssim \delta_n \|u\|_2 \|v\|_2 T_n
	\]
	with probability $1-c\pn^{-c'}$.
	This concludes the result of $(d)$. 
	
	The result of $(e)$ follows by 	
	using $\|\wh H - H\|_{{\rm op}}\le \sqrt{pK}\max_{j,k} |\e_j^\top(\wh H - H)\e_k|$ together with choosing $u = \e_j$ and $v = \e_K$ in part $(d)$ and taking a union bound over $1\le j\le p$ and $1\le k\le K$. \qed

	\section{Proof of Theorem \ref{thm_beta}: convergence rate of $\|\wh\beta - \beta\|_2$ }\label{sec_proof_rate_beta}
	\subsection{Main proof of Theorem \ref{thm_beta}}
	Throughout this proof, we work on the event 
	that   $\wh\Theta^\top \wh \Theta$ can be inverted intersected with
	the event $\E$ defined in (\ref{def_event}). This event holds with probability $1-c\pn^{-c'}$ by recalling that $\PP(\E)\ge 1-\pn^{-c}$ and using Lemma \ref{lem_H_op} in Section \ref{sec_auxiliary_lemma}. Further recall that $\E$ implies $\wh K=K$ and $I_k\subseteq \wh I_k \subseteq I_k\cup J_1^k$ with $J_1^k=\{ j\in J: \ |A_{jk}| \ge 1-4\delta/v\}$, for all $k\in [K]$.  Define  
	$$\wh \Theta^+ := [\wh\Theta^\top \wh\Theta]^{-1}\wh\Theta^\top $$ so that $\wh\Theta^+ \wh\Theta = {\bI}_K$.
	In the same way, using $\textrm{rank}(\Theta)=K$, define
	$\Theta^+  := (\Theta^\top \Theta)^{-1}\Theta^\top $.
	Recall that  
	\[
	\wh \beta =  \wh\Theta^+ {1\over n} \X^\top \y,
	\]
	hence
	\begin{align}\nonumber\label{diff_beta_hat_beta}
	\wh\beta - \beta &= \wh\Theta^+ \left(
	{1\over n} \X^\top \y - \wh\Theta \beta 
	\right)\\
	&= \Theta^+ \left(
	{1\over n} \X^\top \y - \wh\Theta \beta 
	\right) + \left(
	\wh\Theta^+ - \Theta^+
	\right) \left(
	{1\over n} \X^\top \y - \wh\Theta \beta 
	\right).
	\end{align}
	Let $\delta_n^2 := \log(p\vee n) /n$ and let $\bar\rho$ be as defined in (\ref{def_rho}). Lemma \ref{lem_Rem2} in Section \ref{sec_lemmas_rate} implies that, 	with probability $1-\pn^{-c}$,
	\begin{align}\label{rate_Rem2}\nonumber
	&	\left\|\left(
	\wh\Theta^+ - \Theta^+
	\right) \left(
	{1\over n} \X^\top \y - \wh\Theta \beta 
	\right)\right\|_2 \\
	&\qquad\lesssim  \cl^{-1/2}\cdot  {\delta_n\sqrt{p} \over\sigma_K (A\C^{1/2} )} \cdot \left\{ 1 + {\sqrt{p}\over \sigma_K(A\C^{1/2} )}\right\} \delta_n\sqrt{K}\left\{1+{\|\beta\|_2 \over \sqrt m} +  \bar{\rho} \|\beta\|_2\right\}.
	\end{align}
	It remains to study the first term on the right  of (\ref{diff_beta_hat_beta}).  
	Recall the definition	of $\wh \Theta$ in (\ref{est_Theta}),
	\[
	\wh \Theta = \left(
	\wh\Sigma_{\sbt \wh I} - \wh \Gamma_{\sbt \wh I}
	\right) \wh \Pi
	\]with \[\wh \Pi:= \wh A_{\wh I\sbt}\left[\wh A_{\wh I\sbt}^\top \wh A_{\wh I\sbt}\right]^{-1}.\]
	Use the fact that $\X=\Z A^T +\W$ 
	to obtain
	\begin{align}\label{disp_Theta_hat}\nonumber
	\wh \Theta  &= \left({1\over n} \X^\top \X_{{\sbt \wh I}} - \wh \Gamma_{{\sbt \wh I}}\right)\wh \Pi\\
	&= \left[A{1\over n}\Z^\top \Z A_{{\wh I\sbt}}^\top + A{1\over n}\Z^\top \W_{{\sbt {\rm \wh I}}}+ {1\over n}\W^\top \Z A_{{\wh I\sbt}}^\top + \left({1\over n}\W^\top \W_{{\sbt \wh I}} - \wh \Gamma_{{\sbt \wh I}}\right)\right]\wh \Pi .
	\end{align}
	Combine this expression with the expansion
	\begin{eqnarray}
	{1\over n} \X^\top \y = A{1\over n} \Z^\top \Z\beta+ A{1\over n}\Z^\top \beps + {1\over n}\W^\top \Z\beta + {1\over n}\W^\top \eps ,
	\end{eqnarray}
	to arrive at
	\begin{align}\label{disp_diff_h}
	{1\over n} \X^\top \y - \wh\Theta \beta & = A \Delta_Z + \Delta_W
	\end{align}
	with
	\begin{align}\label{D_Z}
	\D_Z &:= 
	{1\over n}\Z^\top \eps - {1\over n}\Z^\top \W_{{\sbt \wh I}}\wh \Pi\beta + {1\over n}\Z^\top \Z\left({{\bI}}_K - A_{{\wh I\sbt}}^\top \wh \Pi\right)\beta,\\\label{D_W}
	\D_W &:= {1\over n}\W^\top \eps +{1\over n}\W^\top \Z\left({\bI}_K - A_{{\wh I\sbt}}^\top \wh \Pi\right)\beta  +\left(\wh \Gamma_{{\sbt \wh I}} - {1\over n}\W^\top \W_{{\sbt \wh I}}  \right)\wh \Pi\beta.
	\end{align}
	Recall  that $\Theta = A\C$ and $\Theta^+ \Theta = {\bI}_K$, to arrive at the identity
	\begin{align}\label{disp_Theta+diff}\nonumber
	\Theta^+ \left(
	{1\over n}\X^\top \y - \wh\Theta \beta 
	\right)  &= \Theta^+ A\C \C^{-1} \D_Z + \Theta^+ \D_W\\
	&=  
	\C ^{-1}\D_Z + \Theta^+\D_W.
	\end{align}
	Use the inequality $\|
	\C^{-1}\D_Z\|_{2} \le \|
	\C^{-1/2}\|_{\rm op}\|\C^{-1/2}\D_Z\|_{2} \le \cl^{-1/2} \|\C^{-1/2}\D_Z\|_2$ and invoke Lemma \ref{lem_Delta_Z} and \ref{lem_Delta_W} in Section \ref{sec_lemmas_rate} to obtain the bounds
	\begin{align}\label{rate_D1}
	& \left\|
	\C^{-1}\D_Z\right\|_2 \lesssim  \cl^{-1/2}\delta_n\sqrt{K}\left(1+{\|\beta\|_2 \over \sqrt m} +  \bar{\rho} \|\beta\|_2\right) ,\\\label{rate_HD2}
	& \|\Theta^+\D_W\|_2 \lesssim  \cl^{-1/2}{\delta_n\sqrt{K} \over \sigma_K(A\C^{1/2})}\left(1+{\|\beta\|_2 \over \sqrt m} +\bar{\rho} \|\beta\|_2\right).
	\end{align} These inequalities hold
	with probability $1-C\pn^{-c}$.
	Use the inequality   $$
	\sigma_K(A\C^{1/2}) \ge \cl^{1/2}\sigma_K(A) \ge \cl^{1/2} \sigma_K(A_{I \sbt}) \ge \cl^{1/2} \sqrt{m}\ge \sqrt{2\cl},$$ 
	and
	collect (\ref{rate_Rem2}), (\ref{rate_D1}) and (\ref{rate_HD2}) to conclude that, with probability $1-C\pn^{-c}$, 
	\begin{align*}
	\|\wh\beta - \beta\|_2 &\le  \left\|\Theta^+ \left(
	{1\over n}\X^\top \y - \wh\Theta \beta 
	\right)\right\|_2 + \left\|\left(
	\wh\Theta^+ - \Theta^+
	\right) \left(
	{1\over n}\X^\top \y - \wh\Theta \beta 
	\right)\right\|_2\\
	&\le \cl^{-1/2}\left\| 
	\C ^{-1}\D_Z\right\|_2 + \|\Theta^+\D_W\|_2 + \left\|\left(
	\wh\Theta^+ - \Theta^+
	\right) \left(
	{1\over n}\X^\top \y - \wh\Theta \beta 
	\right)\right\|_2\\
	&\lesssim \cl^{-1/2}\delta_n\sqrt{K}\left\{ 1+{\|\beta\|_2 \over \sqrt m} +  \bar{\rho} \|\beta\|_2\right\}\left\{ 1+{\delta_n\sqrt{p} \over\sigma_K (A\C^{1/2} )} \cdot \left(1+ {\sqrt{p}\over \sigma_K (A\C^{1/2} )} \right) \right\}.
	\end{align*}
	Recall  that $\lambda_K := \lambda_K(A\C A^\top ) = \sigma_K^2(A\C^{1/2})$ and  invoke Lemma \ref{lem_signal} in Section \ref{sec_lemmas_rate} to deduce $p / \lambda_K \ge K / B_z  $. This implies  
	\begin{align}\label{res_cond_H}
	{\delta_n\sqrt{p} \over\sigma_K(A\C^{1/2})} \cdot \left\{ 1 + {\sqrt{p}\over \sigma_K(A\C^{1/2})}\right\}= \delta_n\sqrt{ {p\over \lambda_K}}+ {\delta_n   {p\over \lambda_K}} \lesssim \delta_n {p\over \lambda_K},
	\end{align}
	whence 
	\begin{eqnarray*}
		\|\wh\beta-\beta\|_2 &\lesssim& \cl^{-1/2}\delta_n\sqrt{K}\left(1+{\|\beta\|_2 \over \sqrt m} +  \bar{\rho} \|\beta\|_2\right)\left(1+\delta_n {p\over \lambda_K}\right).
	\end{eqnarray*}
	Finally, 
	invoke Assumptions \ref{ass_J1} and \ref{ass_H} to complete the proof. \qed
	
	\subsection{Main lemmas used in the proof of Theorem \ref{thm_beta}}\label{sec_lemmas_rate}
	
	The  first two lemmas provide upper bounds for the Euclidean norm of $\left(
	\C\right)^{-1}\D_Z$ and $\D_W$ defined in (\ref{D_Z}) and (\ref{D_W}).  The third lemma controls the sup-norm and $\ell_2$ norm of 
	$(\wh\Theta^+-\Theta^+)(\frac1n \X^\top \y - \wh \Theta \beta)$ in  (\ref{diff_beta_hat_beta}). The final lemma gives a lower bound on the ``signal strength'' $\lambda_K( A\C A^\top)$.
	All lemmas are proved in Section \ref{sec_proof_lem_rate}.
	We use throughout the notation
	\[ H:= A [\Sigma_Z]^{1/2} = \Theta [ \Sigma_Z^{-1} ]^{1/2}.\]
	All statements are valid on some events that are subsets of $\E$ and the  probabilities of these events are greater than $1-C\pn^{-\alpha}$ for some positive constants $C,\alpha$. This is an important observation since on the event $\E$, the dimensions  $\wh K$ and $K$ are equal, which ensures that the various quantities in the statements are well-defined. For instance, $\D_Z$, $\D_W$ and $\wh \Theta- \Theta$ are only defined if $\wh K=K$.
	\begin{lemma}\label{lem_Delta_Z}
		Under the conditions of Theorem \ref{thm_beta}, with probability $1-c\pn^{-c'}$, 
		\[
		\left\|\left(
		\C\right)^{-1/2}\D_Z\right\|_2 \lesssim  \delta_n\sqrt{K}\left(1+{\|\beta\|_2 \over \sqrt m} +  \bar{\rho} \|\beta\|_2\right)
		\]
		and 
		\begin{equation}\label{disp_T2}
		\|\D_Z\|_\i \lesssim  \delta_n \left(1+{\|\beta\|_2 \over \sqrt{m}} + \bar{\rho}\|\beta\|_2\right).
		\end{equation}
	\end{lemma}
	
	\begin{lemma}\label{lem_Delta_W}
		Under the conditions of Theorem \ref{thm_beta}, with probability $1-c\pn^{-c'}$,  
		\[
		\|\Theta^+\D_W\|_2 \lesssim  \cl^{-1/2}{\delta_n\sqrt{K} \over \sigma_K(H)}\left(1+{\|\beta\|_2 \over \sqrt m} +\bar{\rho} \|\beta\|_2\right)
		\]
		and 
		\begin{align}\label{disp_T3}
		\|\D_W\|_{\i} \lesssim  \delta_n\left(1+{\|\beta\|_2 \over \sqrt m} +  \bar{\rho} \|\beta\|_2\right).
		\end{align}	
	\end{lemma}
	
	\begin{lemma}\label{lem_Rem2}
		Under the conditions of Theorem \ref{thm_beta}, with probability $1-c\pn^{-c'}$, we have
		\begin{align*}
		&\left|\e_k^\top \left(
		\wh\Theta^+ - \Theta^+
		\right) \left(
		{1\over n}\X^\top \y - \wh\Theta \beta 
		\right)\right|\\ 
		&\qquad \lesssim [\C^{-1}]_{kk}^{1/2}\cdot  {\delta_n\sqrt{p} \over\sigma_K(H)} \cdot \left(\sqrt{K} + {\sqrt{p}\over \sigma_K(H)}\right) \delta_n\left(1+{\|\beta\|_2 \over \sqrt m} +  \bar{\rho} \|\beta\|_2\right) 
		\end{align*}
		for any $1\le k\le K$. Moreover, with the same probability, 
		\begin{align*}
		&\left\|\left(
		\wh\Theta^+ - \Theta^+
		\right) \left(
		{1\over n}\X^\top \y - \wh\Theta \beta 
		\right)\right\|_2 \\
		&\qquad\lesssim  \cl^{-1/2}\cdot  {\delta_n\sqrt{p} \over \sigma_K(H)} \cdot \left(1 + {\sqrt{p}\over \sigma_K(H)}\right) \delta_n\sqrt{K}\left(1+{\|\beta\|_2 \over \sqrt m} +  \bar{\rho} \|\beta\|_2\right) 
		\end{align*}
	\end{lemma}
	
	\begin{lemma}\label{lem_signal}
		Under 
		Assumptions \ref{ass_model} and  \ref{ass_subg}, we have 
		\[
		\lambda_K(A\C A^\top ) \le B_z {p\over K}.
		\]
	\end{lemma}

	\subsection{Proof of lemmas in Section \ref{sec_lemmas_rate}}\label{sec_proof_lem_rate}
	
	\subsubsection{Proof of Lemma \ref{lem_Delta_Z}}
	We work on the event $\E$ defined in (\ref{def_event}) that has probability at least $1-\pn^{-c}$ for some $c>0$. Recall that it implies $\wh K = K$ and $I_k \subseteq \wh I_k \subseteq I_k \cup J_1^k$ for all $k\in [K]$. 
	
	By writing $\Omega = \C^{-1}$, note that
	\begin{align*}
	\|\Omega\D_Z\|_2\le \|\Omega^{1/2}\|_{{\rm op}}\cdot  \|\Omega^{1/2}\D_Z\|_2 \le \cl^{-1/2} \cdot \sqrt{K}\cdot \|\Omega^{1/2}\D_Z\|_\i.
	\end{align*}
	From (\ref{D_Z}), it   suffices to   bound 
	\[
	\max_k\left[ {1\over n}|\Omega_{k\sbt}^{1/2}\Z^\top \eps| + {1\over n}|\Omega_{k\sbt}^{1/2}\Z^\top \W_{\sbt\wh I}\wh \Pi\beta| + 
	{1\over n}\left|\Omega_{k\sbt }^{1/2}\Z^\top \Z\left({\bI}_K - A_{{\wh I\sbt}}^\top \wh \Pi\right)\beta\right|\right].
	\]
	Note that, by definition (\ref{def_deltaA_I}),
	\[
	\left|\Omega_{k\sbt }^{1/2}\Z^\top \Z\left({\bI}_K - A_{{\wh I\sbt}}^\top \wh \Pi\right)\beta\right|  \le \max_a 	\left|\Omega_{k\sbt }^{1/2}\Z^\top \Z_{\sbt a}\right|\left\|\D\beta\right\|_1. 
	\]
	Invoking (\ref{rate_zeps}), (\ref{rate_zw_td}), (\ref{rate_zz}) in Lemma \ref{lem_quad} together with Lemma \ref{lem_deltaA_I} concludes the proof of the first result. The second result follows immediately by the same arguments applied to $\max_k |\e_k^\top \D_Z|$.
	\qed
	
	\subsubsection{Proof of Lemma \ref{lem_Delta_W}}
	We work on the event $\E$ defined in (\ref{def_event}) that has probability at least $1-\pn^{-c}$ for some $c>0$. Recall that it implies $\wh K = K$ and $I_k \subseteq \wh I_k \subseteq I_k \cup J_1^k$ for all $k\in [K]$. 
	
	First recall that $H\C^{1/2} = \Theta$ implies $\Theta^+ = \Omega^{1/2} H^+$. Note that 
	\begin{align*} 
	\|\Theta^+ \D_W\|_2 =  	\|\Omega^{1/2}H^+ \D_W\|_2\le \cl^{-1/2}\sqrt{K}\|H^+\D_W\|_\i.
	\end{align*}
	By definitions (\ref{D_W}) and (\ref{def_deltaA_I}), we can  bound 
	\[
	\max_k \left\{{1\over n}\left|\e_k^\top H^+\W^\top \beps\right| + {1\over n}\left\|\e_k^\top H^+\W^\top \Z\right\|_\i\|\D\beta\|_1 + \left|\e_k^\top H^+\left({1\over n}\W^\top \W_{\sbt\wh I} - \wh \Gamma_{{\sbt \wh I}}\right)\wh \Pi\beta\right|\right\}.
	\]
	Apply (\ref{rate_hweps}), (\ref{rate_hwz}) and (\ref{rate_hww}) in Lemma \ref{lem_quad}, apply Lemma \ref{lem_deltaA_I} and take  an union bound to complete the proof of the first result. To prove (\ref{disp_T3}), since
	\begin{align*}
	\|\D_W\|_\i& \le \max_{j}\left\{{1\over n}\left|\W_{\sbt j}^\top \eps\right| + {1\over n}\left\|\W_{\sbt j}^\top \Z\right\|_\i\|\D\beta\|_1  + \left|\e_j^\top \left({1\over n}\W^\top  \W_{\sbt \wh I} - \wh \Gamma_{{\sbt \wh I}}\right)\wh \Pi\beta\right|\right\},
	\end{align*}
	invoking (\ref{rate_weps}), (\ref{rate_wz}) and (\ref{rate_ww}) in Lemma \ref{lem_quad} and applying Lemma \ref{lem_deltaA_I} completes the proof. 
	\qed

	\subsubsection{Proof of Lemma \ref{lem_Rem2}}
	We work on the event $\E$ defined in (\ref{def_event}) that has probability at least $1-\pn^{-c}$ for some $c>0$. Recall that it implies $\wh K = K$ and $I_k \subseteq \wh I_k \subseteq I_k \cup J_1^k$ for all $k\in [K]$. 
	
	Using the identity (\ref{disp_diff_h}) and $\Theta = A\C$ gives 
	\begin{align}\label{iden_diff}\nonumber
	\left(
	\wh\Theta^+ - \Theta^+
	\right) \left(
	{1\over n}\X^\top \y - \wh\Theta \beta 
	\right)
	&=\left(
	\wh\Theta^+ - \Theta^+
	\right) (A\D_Z +\D_W)\\
	&= \left( [\Theta^\top \Theta]^{-1}[\wh \Theta- \Theta]^\top P_{\wh \Theta}^{\perp} + \Theta^+[\Theta-\wh \Theta]\wh \Theta^+\right)(A\D_Z +  \D_W).
	\end{align}
	In the last line, we  used $P_{\wh \Theta}^{\perp} = {\bI}_p - \wh \Theta\wh \Theta^+$ and the identity
	\begin{equation}\label{eq_theta_hat_inv}
	\wh \Theta^+ -\Theta^+ = (\Theta^\top \Theta)^{-1}(\wh \Theta- \Theta)^\top P_{\wh \Theta}^{\perp} + \Theta^+(\Theta-\wh \Theta)\wh \Theta^+.
	\end{equation}
	%
	To prove the first statement of the lemma,  recall that $\Theta \Omega^{1/2} = H$ and $\wh \Theta\Omega^{1/2} = \wh H$, so that we can write 
	\begin{align*}
	&\left|
	\e_k^\top  \left(
	\wh\Theta^+ - \Theta^+
	\right) \left(
	{1\over n}\X^\top \y - \wh\Theta \beta 
	\right) \right|\\ 
	&= \left|
	\e_k^\top \left( [\Theta^\top \Theta]^{-1}[\wh \Theta- \Theta]^\top P_{\wh \Theta}^{\perp} + \Theta^+[\Theta-\wh \Theta]\wh \Theta^+\right)(A\D_Z +  \D_W)
	\right|\\
	&= \left|
	\e_k^\top \Omega^{1/2}\left([H^\top H]^{-1}[\wh H-H]^\top  P_{\wh \Theta}^{\perp} + H^+[H-\wh H]\wh H^+\right)(A\D_Z + \D_W)
	\right|\\
	&\le \left\|
	\e_k^\top \Omega^{1/2}\left([H^\top H]^{-1}[\wh H-H]^\top  P_{\wh \Theta}^{\perp} + H^+[H-\wh H]\wh H^+\right)
	\right\|_2\left(\|A\D_Z\|_2 + \|\D_W\|_2\right)
	\end{align*}
	by using the Cauchy-Schwarz inequality in the last step. 
	Note that $$
	\|A\D_Z\|_2 + \|\D_W\|_2 \le \sqrt{p}\max_j\left(|A_{j\sbt}^\top \D_Z| + \|\D_W\|_\i\right)\le \sqrt{p}(\|\D_Z\|_\i + \|\D_W\|_\i)
	$$
	from the fact that $\|A_{j\sbt}\|_1\le 1$.
	The first term can be  bounded above via
	\begin{align}\label{disp_T1}\nonumber
	& \left\|\e_k^\top \Omega^{1/2}(H^\top H)^{-1}(\wh H-H)^\top P_{\wh \Theta}^{\perp} \right\|_{2} + \left\|\e_k^\top \Omega^{1/2}H^+(\wh H-H)\wh H^+\right\|_{2}\\\nonumber
	&\le \left\|\e_k^\top \Omega^{1/2}(H^\top H)^{-1}(\wh H-H)^\top\right\|_{2} + \Omega_{kk}^{1/2}{\|H^+(\wh H-H)\|_{{\rm op}}\over \sigma_K(\wh H)}.
	\end{align}
	Apply  Lemma \ref{lem_H_op}  to obtain 
	\begin{equation}\label{disp_H}
	\sigma_K(\wh H) \ge (1-c_1)\sigma_K(H),\qquad \|H^+(\wh H-H)\|_{{\rm op}} \lesssim \delta_n\sqrt{K}
	\end{equation}
	and 
	\begin{align}\label{disp_H_prime}\nonumber
	\left	\|\e_k^\top\Omega^{1/2}(H^\top H)^{-1}(\wh H-H)^\top \right\|_{2} &\le \sqrt{p}\max_{j\in [K]} \left|\e_k^\top \Omega^{1/2}(H^\top H)^{-1}(\wh H-H)^\top \e_j\right| \\\nonumber
	&\lesssim \sqrt{p} \cdot \delta_n  \left\|\e_k^\top \Omega^{1/2}(H^\top H)^{-1}\right\|_2 \\
	&\le   {\delta_n \sqrt{p} \over \lambda_K(H^\top H)}\Omega_{kk}^{1/2}.
	\end{align}
	Hence,	
	\begin{align}\label{disp_T1} \nonumber
	&\left\|\e_k^\top \Omega^{1/2}(H^\top H)^{-1}(\wh H-H)^\top P_{\wh \Theta}^{\perp} \right\|_{2} + \left\|\e_k^\top \Omega^{1/2}H^+(\wh H-H)\wh H^+\right\|_{2}\\
	&\lesssim  \left(
	{\delta_n \sqrt{p} \over \lambda_K(H^\top H)} + {\delta_n\sqrt{K} \over \sigma_K(H)}
	\right)\Omega_{kk}^{1/2},
	\end{align}
	and, using  (\ref{disp_T2}) and (\ref{disp_T3}), we conclude
	\begin{align*}
	&	\left\|
	\e_k^\top \Omega^{1/2}\left([H^\top H]^{-1}[\wh H-H]^\top  P_{\wh \Theta}^{\perp} + H^+(H-\wh H)\wh H^+\right)
	\right\|_2\left(\|A\D_Z\|_2 + \|\D_W\|_2\right)\\
	&\qquad \lesssim \Omega_{kk}^{1/2}\left(
	{\delta_n p\over \lambda_K(H^\top H)} + {\delta_n\sqrt{pK} \over \sigma_K(H)}
	\right)\cdot \delta_n\left(1+{\|\beta\|_2 \over \sqrt m} +  \bar{\rho} \|\beta\|_2\right)
	\end{align*}
	with probability $1-\pn^{-c}$. This proves the first statement.

	We now prove the second statement. From (\ref{iden_diff}), observe that 
	\begin{align*}
	&\left\|\left([ \Theta^\top \Theta]^{-1}[\wh \Theta- \Theta]^\top P_{\wh \Theta}^{\perp} + \Theta^+[\Theta-\wh \Theta]\wh \Theta^+\right)(A\D_Z +  \D_W)\right\|_2\\
	&\le \|\Omega^{1/2}\|_{{\rm op}}\left\|
	\left([H^\top H]^{-1}[\wh H-H]^\top  P_{\wh \Theta}^{\perp}  + H^+[H-\wh H]\wh H^+\right)(A\D_Z + \D_W)
	\right\|_2\\
	&\le \|\Omega^{1/2}\|_{{\rm op}}\left[
	{\|\wh H - H\|_{{\rm op}} \over \lambda_K(H^\top H)} + {\|H^+(\wh H - H)\|_{{\rm op}} \over \sigma_K(\wh H)}\right] \|A\D_Z + \D_W\|_2.
	\end{align*}
	Invoking (\ref{disp_H}), (\ref{disp_T1}), (\ref{disp_T2}) and  Lemma \ref{lem_H_op}  concludes the proof. \qed

	\subsubsection{Proof of Lemma \ref{lem_signal}}
	We argue that
	\begin{align*}
	\lambda_K(A\C A^\top )  &= \lambda_{\min}(\C^{1/2} A^\top A \C^{1/2} )\\ &\le \min_{k\in[K]} \e_k^\top \C^{1/2}A^\top A\C^{1/2}\e_k\\
	&\le  {1\over K}\sum_{k=1}^K\sum_{j=1}^p A_{j\sbt}^\top\C^{1/2}\e_k\e_k^\top \C^{1/2} A_{j\sbt}\\
	&= {1\over K}\sum_{j=1}^p A_{j\sbt}^\top\C A_{j\sbt}.
	\end{align*}
	By $\|A_{j\sbt}\|_1\le 1$, the result   follows from the inequality
	\[
	A_{j\sbt}^\top\C A_{j\sbt}\le \|A_{j\sbt}\|_1^2 \|\C\|_\i \le B_z.
	\]
	\qed

	\section{Proof of Theorem \ref{thm_distr}: Asymptotic normality of $\wh \beta$}\label{sec_proof_asn_beta}
	
	\subsection{Main proof of Theorem \ref{thm_distr}}
	We work throughout this proof on the event $\E$ so that
	$\wh K=K$ and $I_k\subseteq \wh I_k \subseteq I_k\cup J_1^k$ with $J_1^k=\{ j\in J: \ |A_{jk}| \ge 1-4\delta/v\}$, for all $k\in [K]$, and
	the event that the inverse of $\wh\Theta^\top \wh \Theta$ exists.   This event holds with probability $1-c\pn^{-c'}$ by recalling that $\PP(\E)\ge 1-\pn^{-c}$ and  Lemma \ref{lem_H_op}.
	Recall from (\ref{diff_beta_hat_beta}) the identity
	\begin{align}\label{diff_beta_hat_beta_prime}
	\wh\beta - \beta = \Theta^+\left(
	{1\over n}\X^\top \y - \wh \Theta \beta
	\right)  + \left(\wh\Theta^+ - \Theta^+\right)\left(
	{1\over n}\X^\top \y - \wh \Theta \beta
	\right) .
	\end{align}
	The proof consists of four main steps:
	\begin{enumerate}[itemsep = 1pt]
		\item[(1)] We   derive the decomposition of $ \wh\beta_k - \beta_k $ as an average of independent, mean zero random variables   $n^{-1}\sum_{i=1}^n \xi_{ik}$  and two remainder terms;
		\item[(2)] We verify that $\EE[\xi_{ik}]=0$ and calculate the variance  $\EE[\xi_{ik}^2]$;
		\item[(3)] We apply   Lyapunov's central limit  theorem  for triangular arrays to $ \sum_{t=1}^n \xi_{tk}$;
		\item[(4)] We show that the two remainder terms are asymptotically negligible.\\
		
	\end{enumerate}
	
	\paragraph{Step 1:} Set $\Pi:= A_{I\sbt} [A_{I\sbt}^\top A_{I\sbt}]^{-1}$. 
	Recall (\ref{disp_diff_h}), (\ref{D_Z}) and (\ref{D_W}), to obtain the identity 
	\begin{align*}
	{1\over n}\X^\top \y - \wh\Theta \beta & = A \Delta_Z + \Delta_W
	\end{align*}
	and rewrite $\Delta_Z$ and $\Delta_W$
	by adding and subtracting appropriate  terms as
	\begin{align*}
	\D_Z &= 
	\underbrace{{1\over n}\Z^\top \eps - {1\over n}\Z^\top \W_{{\sbt I}}\Pi  \beta}_{\D_{Z,I}} + \underbrace{{1\over n}\Z^\top \left(\W_{{\sbt I}}\Pi  -\W_{{\sbt \wh I}}\wh \Pi\right)\beta + {1\over n}\Z^\top \Z\left({\bI}_K - A_{{\wh I\sbt}}^\top \wh \Pi\right)\beta}_{\D_{Z, J_1}},\\
	\D_W &= \underbrace{{1\over n}\W^\top \eps  +\left( \wt \Gamma_{\sbt I} - {1\over n}\W^\top \W_{{\sbt I}} \right)\Pi \beta}_{\D_{W,I}}  \\
	&\quad + \underbrace{{1\over n}\W^\top \Z\left({\bI}_K - A_{{\wh I\sbt}}^\top \wh \Pi\right)\beta+ \left[\left(\wh \Gamma_{{\sbt \wh I}}  - {1\over n}\W^\top \W_{{\sbt \wh I}} \right)\wh \Pi - \left(\wt \Gamma_{\sbt I}  - {1\over n}\W^\top \W_{{\sbt I}} \right)\Pi \right] \beta}_{\D_{W, J_1}}.
	\end{align*}
	Here, in view of (\ref{est_Gamma_I}), the matrix $\wt\Gamma$ is a diagonal $p\times p$  matrix with
	\begin{align}\label{def_bar_Sigmaw}
	\wt\Gamma_{ii} &:= \wh\Sigma_{ii} - {1\over |I_a|(|I_a|-1)}\sumsum_{j, \ell \in I_a,\ j\ne \ell} \wh\Sigma_{j\ell}\quad \forall i\in I_a, ~ a\in [K].
	\end{align}
	Plug   (\ref{disp_Theta+diff}) into (\ref{diff_beta_hat_beta_prime}), 
	use the above expressions of $\D_Z$ and $\D_W$, and find, for each $k\in [K]$,	\begin{eqnarray*}
		\wh\beta_k - \beta_k &=&
		\e_k^\top \Theta^+\left(
		{1\over n}\X^\top \y - \wh \Theta \beta
		\right)  +\e_k^\top \left(\wh\Theta^+ - \Theta^+\right)\left(
		{1\over n}\X^\top \y - \wh \Theta \beta
		\right) \\
		&=&
		\e_k^\top\left(	\C ^{-1}\D_Z + \Theta^+\D_W \right) 
		+\e_k^\top \left(\wh\Theta^+ - \Theta^+\right)\left(
		{1\over n}\X^\top \y - \wh \Theta \beta
		\right) \\
		&=&
		\e_k^\top \left( \C^{-1} \D_{Z,I} + \Theta^+\D_{W,I}\right) +  [{\rm Rem}_1]_k +   [{\rm Rem}_2]_k 
	\end{eqnarray*}
	with 
	\begin{align}\label{def_Rem_2}
	{\rm Rem}_1  = \left(\wh\Theta^+ - \Theta^+\right)\left(
	{1\over n}\X^\top \y - \wh \Theta \beta
	\right), \qquad 
	{\rm Rem}_2 = \C^{-1}\D_{Z, J_1}+ \Theta^+\D_{W, J_1}.
	\end{align}
	We now write 
	$ \e_k^\top \left( \C^{-1} \D_{Z,I} + \Theta^+\D_{W,I}\right)$ as a sum of independent variables.
	First observe that
	\begin{eqnarray*}
		\e_k^\top \C^{-1} \D_{Z,I} &=&\e_k^\top  \C^{-1}
		\left( \frac1n \Z^\top \eps - \frac1n \Z^\top \W_{{\sbt I}} \Pi  \beta \right)\\
		&=& \e_k^\top  \C^{-1}   \frac1n \sum_{t=1}^n \Z_{t\sbt} \left( \eps_t -   \W_{{tI}} \Pi  \beta \right)
	\end{eqnarray*}
	and 
	\begin{eqnarray*}
		\e_k^\top \Theta^+ \D_{W,I} &=& \e_k^\top \Theta^+ \left\{
		\frac1n \W^\top \eps + \left(\widetilde \Gamma_{{\sbt I}}-\frac1n \W^\top \W_{{\sbt I}} \right)\Pi \beta\right\}\\
		&=& \e_k^\top \Theta^+ \left\{
		\frac1n \sum_{t=1}^n\W_{t\sbt} \eps_t +   \frac1n \sum_{t=1}^n \widetilde \Gamma_{{\sbt I}}^{(t)}\Pi \beta -\frac1n\sum_{t=1}^n  \W_{t\sbt} \W_{{t I}}  \Pi \beta\right\}
	\end{eqnarray*}
	with, for any $t\in [n]$,  $\wt \Gamma_{ij}^{(t)} := 0$ for all $i\ne j $ and  
	\begin{equation}\label{def_bar_Sigmaw_t}
	\wt\Gamma_{ii}^{(t)} := \X_{ti}^2- {1\over |I_a|(|I_a|-1)}\sumsum_{j,\ell \in I_a,\ j\ne \ell} \X_{tj}  \X_{t\ell},~ \text{for all }i\in I_a, ~ a\in [K],
	\end{equation}
	so that $\wt\Gamma_{{\sbt I}}= n^{-1} \sum_{t=1}^n \wt \Gamma_{\sbt I}^{(t)}$.
	Finally, define 
	\begin{eqnarray}\label{def_B_tk}
	\xi_{tk} &:=&  \e_k^\top\C^{-1} \Z_{t\sbt}  \left[\eps_t -  \W_{{t\sbt}}\Pi  \beta\right]  + \e_k^\top \Theta^+\left[ \W_{t\sbt}\eps_t  +  \left(\wt \Gamma_{\sbt I}^{(t)} - \W_{t\sbt}\W_{{t I}}	\right)\Pi \beta
	\right]
	\end{eqnarray}
	and conclude that  
	\begin{align}\label{disp_clt_term}
	\wh\beta_k - \beta_k  = {1\over  {n}}\sum_{t=1}^n \xi_{tk} +   [{\rm Rem}_1]_k +[{\rm Rem}_2]_k .
	\end{align}
	It is easily verified that, for each $k$,  the random variables $\xi_{1k},\ldots,\xi_{nk}$ are  independent.
	
	\paragraph{Step 2:} 
	Next we calculate the first two moments of $\xi_{tk}$. Lemma \ref{lem_moment}  in Section \ref{sec_lemmas_distr} shows that 
	$
	\EE[\xi_{tk}] = 0
	$
	and 
	\begin{align}\label{def_V_k_general}
	\EE[\xi_{tk}^2] &= \left\{ \sigma^2 + \sum_{\ell=1}^K \beta_\ell^2 \left( \frac{1}{|I_\ell|^2} \sum_{i\in I_\ell} \tau_i^2 \right)  \right\}
	\left\{ [\C^{-1}]_{kk} + \e_k^\top (\Theta^\top \Theta)^{-1}\Theta^\top \Gamma\Theta (\Theta^\top \Theta)^{-1}\e_k\right\}\\\nonumber
	&\quad  + \sum_{\ell=1}^K\sum_{a \in I_\ell } \left(\e_k^\top (\Theta^\top \Theta)^{-1}\Theta^\top \e_a\right)^2{\beta_\ell^2 \over m_\ell}\sum_{i\in I_\ell}\t_i^2\left(
	{1\over (|I_\ell|-1)^2}\sum_{j\in I_\ell \setminus \{i\}}\t_j^2 - {1\over |I_\ell|^2}\sum_{j\in I_\ell}\t_j^2
	\right),
	\end{align}
	for all $1\le t\le n$.
	In case	
	$|I_1| = \cdots = |I_K| = m$ and $\tau_1^2= \cdots = \tau_p^2 = \tau^2$, the above expression simplifies to
	\begin{equation}\label{def_V_k_simp}
	\left(\sigma^2 + \t^2{\|\beta\|_2^2\over m} \right)\left( [\C^{-1}]_{kk} + \tau^2
	\e_k^\top  [\Theta^\top \Theta ]^{-1}\e_k\right) + {\tau^4\over m(m-1)}\sum_{\ell=1}^K\beta_\ell^2 \sum_{i\in I_\ell} \left(\e_k^\top [\Theta^\top \Theta]^{-1}\Theta^\top \e_i\right)^2 .
	\end{equation}
	This corresponds to the expression for $V_k$ in (\ref{def_Vk}).
	Furthermore, since 
	\begin{equation}\label{disp_ThetaTheta_kk}
	\tau^2\e_k^\top [\Theta^\top \Theta]^{-1}\e_k \le   {\tau^2 [\C^{-1}]_{kk}\over \lambda_K(A\C A^\top )}
	\end{equation}
	and the second term in (\ref{def_V_k_simp}) is always smaller than the first term (as shown below), it follows that
	$V_k$ reduces to (\ref{def_Vk_simp}), provided  $\lambda_K(A\C A^\top ) / \tau^2 \to \i$. To appreciate why the first term dominates the second term in (\ref{def_V_k_simp}), observe that 
	\begin{align}\label{bd_minor_var}\nonumber
	\sum_{\ell=1}^K{\beta_\ell^2 \over m}\sum_{i\in I_\ell} \left(\e_k^\top [\Theta^\top \Theta]^{-1}\Theta^\top \e_i\right)^2  &= \sum_{\ell=1}^K \left(\e_k^\top [\Theta^\top \Theta]^{-1}\C \e_\ell\right)^2 \beta_\ell^2\\\nonumber
	&\le \|\beta\|_2^2 \max_\ell \left|\e_k^\top (\Theta^\top \Theta)^{-1} \C \e_\ell \right|^2\\\nonumber
	&\le \|\beta\|_2^2 \cdot 
	\e_k^\top (\Theta^\top \Theta)^{-1} \e_k
	\cdot 
	\max_\ell \left|
	\e_\ell^{\top}\C (\Theta^\top\Theta)^{-1}\C \e_\ell
	\right|\\\nonumber
	& = \|\beta\|_2^2 \cdot 
	\e_k^\top( \Theta^\top \Theta)^{-1} \e_k
	\cdot 
	\max_\ell \left|
	\e_\ell^{\top} (A^\top A)^{-1} \e_\ell
	\right|\\\nonumber
	&\le \|\beta\|_2^2 \cdot 
	\e_k^\top( \Theta^\top \Theta)^{-1} \e_k
	\cdot \lambda_{\min}^{-1}(A^\top A)\\
	&\le {\|\beta\|_2^2 }
	\e_k^\top (\Theta^\top \Theta)^{-1} \e_k /m.
	\end{align}
	We used the identity $\Theta = A\C$, $A_{I\sbt} = |\e_\ell|$ for all $i\in I_\ell$ and $|I_\ell| = m$ in the first line, the Cauchy-Schwarz inequality in the second line, $\Theta = A\C$ in the third line and the inequality  $\lambda_{\min}(A^\top A) \ge \lambda_{\min}(A_{I\sbt}^\top  A_{I\sbt}) \ge m$ in the last line. \\

	\paragraph{Step 3:} 
	Next we  apply   Lyapunov's central limit theorem (see, for instance, \cite{Shorack2017}) to  $\sum_{t=1}^n \xi_{tk}$. 
	The independent $\xi_{1k},\ldots,\xi_{nk}$ form a triangular array, with variances possibly changing in $n$, due to the dependence on $p$ and $K$.
	We verify  Lyapunov's condition 
	\begin{equation}\label{cond_lyapunov}
	\lim_{n\to\i}{\sum_{i=1}^n\EE[|\xi_{ik}|^3] \over (\sum_{j=1}^n \EE [ \xi_{jk}^2 ])^{3/2}} =
	\lim_{n\to\i}{\sum_{i=1}^n\EE[|\xi_{ik}|^3] \over (n \EE [ \xi_{1k}^2 ])^{3/2}} = 0
	\end{equation}
	for the simplified case with $  \EE[ \xi_{1k}^2] := V_k$ in (\ref{def_V_k_simp}) only. The  general case using the  variance in (\ref{def_V_k_general}) can be verified in a similar way. 
	Invoke Lemma \ref{lem_third_moment} in Section \ref{sec_lemmas_distr} to obtain 
	\[
	\frac{1}{n\sqrt{n}}\sum_{i=1}^n\EE\left[|\xi_{ik}|^3\right] \lesssim  {1 \over \sqrt n} \left\{
	\left(\gamma_{\eps} +  {\t \|\beta\|_2 \over \sqrt m} \right)^3[\C^{-1}]_{kk}^{3/2} + \left(\gamma_w{\|\beta\|_2 \over  \sqrt{m}} + \gamma_{\eps}\right)^3\gamma_w^3\left(\e_k^\top (\Theta^\top \Theta)^{-1}\e_k\right)^{3/2}
	\right\}.
	\]
	Here $\gamma_\eps$ and $\gamma_w$ are the sub-Gaussian constants of $\eps$ and $W$.
	Compare this expression with $V_k=\EE[\xi_{1k}^2]$ in  (\ref{def_V_k_simp}),
	\[
	V_k^{3/2} \ge \left(\sigma^2 + {\t^2\|\beta\|_2^2/ m}\right)^{3/2}\left( [\C^{-1}]_{kk} + \t^2 \e_k^\top (\Theta^\top \Theta)^{-1}\e_k\right)^{3/2}
	\]
	and conclude 
	\[
	\lim_{n\to\i}{\sum_{i=1}^n\EE[|\xi_{ik}|^3] \over (nV_{k})^{3/2}} \lesssim \lim_{n\to\i} {1\over \sqrt{n}} = 0.
	\]
	The last step used $\gamma_{\eps}\lesssim\Sigma  $ and $\gamma_w \lesssim \tau$. By Lyapunov's central limit theorem, the standardized sum $\sum_{i=1}^n \xi_{ik} / \sqrt{n V_k}$ 
	converges weakly to  $N(0,1)$, as $n\to\infty$.\\

	\paragraph{Step 4:} 
	Finally, we show that both remainder terms in (\ref{def_Rem_2}) are $o_p( \sqrt{ V_k/n})$. For ${\rm Rem}_1$, apply Lemma \ref{lem_Rem2}   in Section \ref{sec_lemmas_rate} and invoke Assumption \ref{ass_J1_prime} to find
	\begin{align*}
	&\sqrt{n}\left|[{\rm Rem}_1]_k\right|\\ 
	&= O_p\left(
	[\C^{-1}]_{kk}^{1/2}\left\{
	1+{\|\beta\|_2 \over \sqrt m} + \bar{\rho}\|\beta\|_2\right\}
	{\delta_n\sqrt{p\log\pn} \over \sigma_K(A\C^{1/2})}\left\{
	\sqrt{K} + 
	{\sqrt{p} 
		\over \sigma_K(A\C^{1/2})}\right\}
	\right)\\
	&=O_p\left([\C^{-1}]_{kk}^{1/2}\left\{
	1+{\|\beta\|_2 \over \sqrt m}\right\}
	{\delta_n\sqrt{p\log\pn} \over \sigma_K(A\C^{1/2})}\left\{
	\sqrt{K} + 
	{\sqrt{p} 
		\over \sigma_K(A\C^{1/2})}\right\}
	\right).
	\end{align*}
	Invoke Assumption \ref{ass_H_prime} and write $\lambda_K := \sigma_K^2(A\C^{1/2})$, to get 
	\[
	{\delta_n\sqrt{p\log\pn} \over \sigma_K(A\C^{1/2})}\left\{
	\sqrt{K} + 
	{\sqrt{p} 
		\over \sigma_K(A\C^{1/2})}\right\} = \sqrt{{K\log^2\pn \over n}\cdot {p\over \lambda_K}} + {p\log\pn\over \lambda_K \sqrt n} = o(1),
	\]
	as $n\to\i$, where we also use $\lambda_K \lesssim (p/K)$ from Lemma \ref{lem_signal} in conjunction with Assumption \ref{ass_H_prime} to deduce $K = o(\sqrt n/\log\pn)$.
	This concludes 
	\begin{align*}
	\sqrt{n}\left|[{\rm Rem}_1]_k\right| =o_p\left(
	[\C^{-1}]_{kk}^{1/2}\left(1+{\|\beta\|_2 \over \sqrt m}\right)\right) = o_p\left(\sqrt{V_k}\right).
	\end{align*}
	Finally,   invoke Lemma \ref{lem_Rem_2} in Section \ref{sec_lemmas_distr} and Assumption \ref{ass_J1_prime}, to obtain
	\begin{align*}
	\sqrt{n}\left| [{\rm Rem}_2]_k\right| 
	& = O_p\left(
	\bar\rho\|\beta\|_2 \sqrt{\log\pn}\sqrt{[\C^{-1}]_{kk}}
	\right)= o_p\left(\sqrt{V_k}\right).
	\end{align*}
	This completes the proof.\qed

	\subsection{Lemmas used in the proof of Theorem \ref{thm_distr}}\label{sec_lemmas_distr}
	Let
	$\xi_{tk}$ 
	be defined in (\ref{def_B_tk}) for all $t\in [n]$ and $k\in [K]$. The first lemma states its first and second moments while the second lemma provides upper bounds for its third absolute moment.  The third lemma studies the rate of ${\rm Rem}_2$ defined in (\ref{def_Rem_2}). These lemmas are proved in Section \ref{sec_proof_lem_distr}.

	\begin{lemma}\label{lem_moment}
		Under 
		Assumption \ref{ass_subg}, we have 
		$
		\EE[\xi_{tk}] = 0
		$
		and 
		\begin{align*}
		\EE[\xi_{tk}^2] &= \left(\sigma^2 + \sum_{a=1}^K \beta_a^2 \bt_a \right)\left([\C^{-1}]_{kk} + \e_k^\top [\Theta^\top \Theta]^{-1}\Theta^\top \Gamma\Theta [\Theta^\top \Theta]^{-1}\e_k\right)\\
		&\quad  + \sum_{a=1}^K {\D_{\tau}^{(a)}\over m_a}\sum_{i\in I_a}\left(\e_k^\top [\Theta^\top \Theta]^{-1}\Theta^\top \e_i\right)^2 
		\end{align*}
		for all $t\in [n]$ and $k\in [K]$, with $\bt_a := \sum_{i \in I_a}\tau_i^2 / m_a^2$ and 
		\begin{equation*}
		\D_{\tau}^{(a)} := {\beta_a^2}\sum_{i\in I_a}\t_i^2\left(
		{1\over (m_a-1)^2}\sum_{j\in I_a \setminus \{i\}}\t_j^2 - {1\over m_a^2}\sum_{j\in I_a}\t_j^2
		\right).
		\end{equation*}		
	\end{lemma}
	
	\begin{lemma}\label{lem_third_moment}
		Under the conditions of Theorem \ref{thm_distr}, we have
		\[
		\EE[|\xi_{tk}|^3] \lesssim 
		\left(\gamma_{\eps} +  {\t \|\beta\| \over \sqrt m} \right)^3[\C^{-1}]_{kk}^{3/2} + \left(\gamma_w{\|\beta\| \over  \sqrt{m}} + \gamma_{\eps}\right)^3\gamma_w^3\left(\e_k^\top [\Theta^\top \Theta]^{-1}\e_k\right)^{3/2}.
		\]
	\end{lemma}
	
	\begin{lemma}\label{lem_Rem_2}
		Let ${\rm Rem}_2$ be defined in (\ref{def_Rem_2}) on the event $\E$ defined in (\ref{def_event}). Under the conditions of Theorem \ref{thm_distr}, we have
		\[
		\sqrt{n}\left|[{\rm Rem}_2]_k\right| \cdot \1_\E = O_p\left(
		\bar{\rho} \|\beta\|_2 \sqrt{\log\pn}\sqrt{[\C^{-1}]_{kk}}
		\right).
		\]
	\end{lemma}

	\subsection{Proof of lemmas in Section \ref{sec_lemmas_distr}}\label{sec_proof_lem_distr}
	
	\subsubsection{Proof of Lemma \ref{lem_moment}}
	Fix any $t\in [n]$ and $k\in [K]$.
	Recall that 
	\begin{eqnarray*}
		\xi_{tk} &:=&  \e_k^\top\C^{-1} \Z_{t\sbt}  \left[\eps_t -  \W_{{t\sbt}}\Pi  \beta\right]  + \e_k^\top \Theta^+\left[ \W_{t\sbt}\eps_t  +  \left(\wt \Gamma_{\sbt I}^{(t)} - \W_{t\sbt}\W_{{t I}}	\right)\Pi \beta
		\right]
	\end{eqnarray*}
	Since $\Theta = A\C$, we have
	$$
	\Theta^+ \left(\wt \Gamma_{\sbt I}^{(t)} - \W_{t\sbt} \W_{tI}	\right)\Pi\beta = (\Theta^\top \Theta)^{-1}\C A^\top \left(\wt \Gamma_{\sbt I}^{(t)} - \W_{t\sbt}\W_{tI}	\right)\Pi\beta.
	$$
	Write $U^{(t)} := (\wt \Gamma_{\sbt I}^{(t)} - \W_{t\sbt}\W_{tI}	)\Pi\in \R^{p\times K}$, and observe that  definition   (\ref{def_bar_Sigmaw_t}) yields
	\begin{alignat*}{2}
	U_{ia}^{(t)} &= \W_{ti}^\top \oW_{ta},\quad  &&\forall i\in J,\ a\in [K];\\
	U_{ib}^{(t)} &=  \W_{ti}^\top \oW_{tb},\quad &&\forall i\in I_a,\ a,b\in [K],\ b\ne a;\\
	U_{ia}^{(t)} &=  \W_{ti}^\top \oW_{ta} - {1\over m_a}\wt\Gamma_{ii}^{(t)},\qquad &&\forall i\in I_a,\ a\in [K].
	\end{alignat*}
	Lemma \ref{lem_WW}, stated immediately after this proof, gives
	\[
	{1\over m_a}\sum_{i \in I_a}U_{ia}^{(t)} = \oW_{ta}^\top \oW_{ta} - {1\over m_a^2}\sum_{i\in I_a}\wt\Gamma_{ii}^{(t)} = {1\over m_a(m_a-1)}\sum_{j\ne \ell \in I_a} \W_{tj}\W_{t\ell}.
	\]
	Hence,	we have $$	A_{J\sbt}^\top U_{J\sbt}^{(t)} = A_{J\sbt}^\top \W_{tJ}\oW_t \in \R^{K\times K}$$  and 
	\begin{align}\label{disp_A_IU_I}
	\left[A_{I\sbt}^\top U^{(t)}_{I\sbt}\right]_{ab} &= {1\over m_b}\sum_{j \in I_a, \ell \in I_b} \W_{tj}\W_{t\ell},\qquad \forall a\ne b,\ a,b\in [K],\\
	\left[A_{I\sbt}^\top U^{(t)}_{I\sbt}\right]_{aa} &= {1\over m_a-1}\sum_{j \ne \ell \in I_a} \W_{tj}\W_{t\ell},\quad \forall a\in [K].
	\end{align}
	It   follows that
	\begin{align*}
	\xi_{tk} &=
	\Omega_{k\sbt  }^\top \Z_{t\sbt} \eps_t - \Omega_{k\sbt }^\top \Z_{t\sbt} \oW_{t\sbt}^\top \beta + \e_k^\top \Theta^+ \W_{t\sbt}\eps_t \\\nonumber
	&\quad - \e_k^\top (\Theta^\top \Theta)^{-1}\C A_{J\sbt }^\top \W_{tJ}\oW_{t\sbt}^\top \beta - \e_k^\top (\Theta^\top \Theta)^{-1}\C A_{I\sbt}^\top U_{I\sbt}^{(t)}\beta.
	\end{align*}
	Since all terms in the above display have zero mean,  $\EE[\xi_{tk}] = 0$.  We use  Assumption \ref{ass_subg} to verify  that any two terms are uncorrelated, and we find $\EE[\xi_{tk}^2]$ after summing up the second moments for each individual term.  
	We have 
	\begin{align*}
	&\EE\left[\Omega_{k\sbt  }^\top \Z_{t\sbt} \eps_t \right]^2 = \sigma^2 \Omega_{kk}, \\ 
	&\EE\left[\Omega_{k\sbt }^\top \Z_{t\sbt} \oW_{t\sbt}^\top \beta \right ]^2 = \Omega_{kk} \sum_{a=1}^K {\beta_a^2 \over m_a^2}\sum_{i \in I_a} \tau_i^2 = \Omega_{kk} \sum_{a=1}^K \beta_a^2 \bt_a\\
	&\EE\left[ \e_k^\top \Theta^+W_{t\sbt}\eps_t \right]^2 = \sigma^2 \sum_{i =1}^p\tau_i^2 \left( \e_k^\top (\Theta^\top \Theta)^{-1}\Theta_{I\sbt}\right)^2\\
	&\EE\left[\e_k^\top (\Theta^\top \Theta)^{-1}\C A_{J\sbt}^\top \W_{tJ}\oW_{t\sbt}^\top  \beta \right]^2 = \sum_{j\in J}\tau_i^2 \left(\e_k^\top (\Theta^\top \Theta)^{-1}\Theta_{i\sbt}\right)^2\sum_{a=1}^K  \beta_a^2\bt_a,
	\end{align*}
	by writing $\bt_a = \sum_{i \in I_a}\tau_i^2 / m_a^2$. Finally, we have 
	\begin{align*}
	&\EE\left[ \e_k^\top (\Theta^\top \Theta)^{-1}\C A_{I\sbt}^\top U_{I\sbt}^{(t)}\beta
	\right]^2\\ &=\sum_{a=1}^K \left(\e_k^\top (\Theta^\top \Theta)^{-1}[\C]_{a\sbt}\right)^2\sum_{b=1}^K\beta_b^2\EE\left[(A_{I \sbt}^\top U_{I\sbt})_{ab}\right]^2\\
	&= \sum_{a=1}^K \left(\e_k^\top (\Theta^\top \Theta)^{-1}[\C]_{a\sbt}\right)^2\left[
	\sum_{b=1}^K{\beta_b^2 \over m_b^2}\sumsum_{i \in I_a,\ j\in I_b}\tau_i^2\tau_j^2 + \D_{\tau}^{(a)}
	\right]\\
	&= \sum_{a=1}^K\sum_{i\in I_a}\tau_i^2 \left(\e_k^\top (\Theta^\top \Theta)^{-1}[\C]_{a\sbt}\right)^2
	\sum_{b=1}^K\beta_b^2\bt_b  + \sum_{a=1}^K \left(\e_k^\top (\Theta^\top \Theta)^{-1}[\C]_{a\sbt}\right)^2 \D_{\tau}^{(a)}
	\end{align*}
	with 
	\begin{align*}
	\D_{\tau}^{(a)} &= {\beta_a^2 \over (m_a-1)^2}\sumsum_{i,j\in I_a, i\ne j}\tau_i^2\tau_j^2 - {\beta_a^2 \over m_a^2}\sum_{i, j\in I_a}\tau_i^2\tau_j^2\\
	& = \beta_a^2\sum_{i\in I_a}\t_i^2\left(
	{1\over (m_a-1)^2}\sum_{j\in I_a \setminus \{i\}}\t_j^2 - {1\over m_a^2}\sum_{j\in I_a}\t_j^2
	\right).
	\end{align*}
	Since  $A_{i\sbt} = \e_a$ for any $i\in I_a$ and   $\Theta = A\C$, we further find
	\begin{align*}
	&\EE\left[\e_k^\top (\Theta^\top \Theta)^{-1}\C A_{I\sbt}^\top U_{I\sbt}^{(t)}\beta
	\right]^2\\
	& = \sum_{i\in I}\tau_i^2 \left(\e_k^\top (\Theta^\top \Theta)^{-1}\Theta_{I\sbt}\right)^2
	\sum_{b=1}^K\beta_b^2\bt_b  + \sum_{a=1}^K {\D_{\tau}^{(a)} \over m_a} \sum_{i\in I_a}\left(\e_k^\top (\Theta^\top \Theta)^{-1}\Theta^\top \e_i\right)^2 .
	\end{align*}
	Collecting all second moments yields (\ref{def_V_k_general}) and completes the proof. \qed
	
	\bigskip

	\begin{lemma}\label{lem_WW}
		Let  $\wt \Gamma_{\sbt I}$ be as  defined in (\ref{def_bar_Sigmaw}). For any $a\in [K]$, we have
		\[
		{1\over n}\oW_{\sbt a}^\top \oW_{\sbt a} - {1\over m_a^2}\sum_{i\in I_a}\wt\Gamma_{ii} = {1\over m_a(m_a-1)}\sumsum_{j\ne \ell \in I_a}{1\over n}\W_{\sbt j}^\top \W_{\sbt\ell}.
		\]
	\end{lemma}
	\begin{proof}[Proof of Lemma \ref{lem_WW}]
		Fix any $a\in [K]$. Definition (\ref{def_bar_Sigmaw}) yields
		\begin{align*}\nonumber
		&{1\over m_a^2}\sum_{i\in I_a}\wt\Gamma_{ii} - 	{1\over n}\oW_{\sbt a}^\top \oW_{\sbt a}\\\nonumber 
		&= {1\over m_a^2}\sum_{i\in I_a}{1\over n}\X_{\sbt i}^\top \X_{\sbt i} - {1\over m_a^2(m_a-1)}\sumsum_{i,j\in I_a, i\ne j}{1\over n}\X_{\sbt i}^\top \X_{\sbt j} - {1\over m_a^2}\sum_{i,j\in I_a}{1\over n}\W_{\sbt i}^\top \W_{\sbt j}\\\nonumber
		&=  {1\over m_a^2}\sum_{i\in I_a}{1\over n}\left(\Z_{\sbt a}^\top \Z_{\sbt a} + 2\Z_{\sbt a}^\top \W_{\sbt i}  + \W_{\sbt i}^\top \W_{\sbt j} \right)- {1\over m_a^2}\sum_{i\in I_a}{1\over n}\W_{\sbt i}^\top \W_{\sbt i} -  {1\over m_a^2}\sumsum_{i,j\in I_a, i\ne j}{1\over n}\W_{\sbt i}^\top \W_{\sbt j}\\\nonumber
		&\quad  -{1\over m_a^2(m_a-1)}\sumsum_{i,j\in I_a, i\ne j}{1\over n}\left(\Z_{\sbt a}^\top \Z_{\sbt a} + \Z_{\sbt a}^\top \W_{\sbt j} + \Z_{\sbt a}^\top \W_{\sbt i } + \W_{\sbt i }^\top \W_{\sbt j} \right)\\\nonumber
		&=- {1\over m_a(m_a-1)}\sumsum_{i,j\in I_a, i\ne j}{1\over n}\W_{\sbt i}^\top \W_{\sbt j},
		\end{align*}
		as desired.
	\end{proof}

	\subsubsection{Proof of Lemma \ref{lem_third_moment}}

	By using the inequality $|x+y|^3 \le 4(|x|^3 + |y|^3)$ and (\ref{def_B_tk}), we obtain 
	\begin{align*}	
	\EE\left[|\xi_{tk}|^3\right]&\le 4 
	\EE\left[\left|\e_k^\top (\Theta^\top \Theta)^{-1}\C A_{I \sbt}^\top U_{I\sbt}^{(t)}\beta \right|^3\right] +
	4\EE\left[\left|\e_k^\top (\Theta^\top \Theta)^{-1}\C A_{J \sbt}^\top \W_{tJ}\oW_{t\sbt}^\top \beta\right|^3 \right]\\
	&\qquad \quad + 4\EE\left[\left|\e_k^\top \Theta^+  \W_{t\sbt}\eps_t\right|^3 \right] + 4\EE\left[\left|\Omega_{k\sbt }^\top \Z_{t\cdot} (\eps_t-\langle \oW_{t\sbt}, \beta\rangle\right|^3 \right]
	.
	\end{align*}
	The upper bound of the first term is established in Lemma \ref{lem_var_UI}, stated immediately after this proof. For the second term,  by using  the inequality $\|XY\|_{\psi_1} \le \|X\|_{\psi_2}\|Y\|_{\psi_2}$\footnote{For any random variable $X$, we write $\|X\|_{\psi_1} = \sup_{p\ge 1}p^{-1}(\EE[|X|^p])^{1/p}$ and $\|X\|_{\psi_2} = \sup_{p\ge 1}p^{-1/2}(\EE[|X|^p])^{1/p}$.} for any two sub-Gaussian random variables $X$ and $Y$, we find
	\[
	\left\|\e_k^\top (\Theta^\top \Theta)^{-1}\C A_{J \sbt}^\top \W_{tJ}\oW_{t\sbt}^\top \beta\right\|_{\psi_1}\le
	\left\|\e_k^\top (\Theta^\top \Theta)^{-1}\C A_{J \sbt}^\top \W_{tJ}\right\|_{\psi_2}\left\|\oW_{t\sbt}^\top \beta\right\|_{\psi_2}.
	\]
	Lemma \ref{lem_W} implies  that $\oW_{t\sbt}^\top \beta$ is $\gamma_w\|\beta\|/\sqrt{m}$-sub-Gaussian, hence $\|\oW_{t\sbt}^\top \beta\|_{\psi_2} \le c\w\|\beta\|/\sqrt{m}$ for some constant $c>0$. Similarly, 
	\[
	\left\|\e_k^\top (\Theta^\top \Theta)^{-1}\C A_{J \sbt}^\top \W_{tJ}\right\|_{\psi_2} \le c\gamma_w\sqrt{\e_k^\top (\Theta^\top \Theta)^{-1}\Theta_{J \sbt}^\top \Theta_{J \sbt}(\Theta^\top \Theta)^{-1}\e_k} \le c\gamma_w \sqrt{\e_k^\top (\Theta^\top \Theta)^{-1}\e_k}
	\]
	and
	\[
	\left\|\e_k^\top (\Theta^\top \Theta)^{-1}\C A_{J \sbt}^\top \W_{tJ}\oW_{t\sbt}^\top \beta\right\|_{\psi_1} \le c \gamma_w^2 {\|\beta\|\over \sqrt m}\sqrt{\e_k^\top (\Theta^\top \Theta)^{-1}\e_k}.
	\]
	This inequality and the definition of $\|\cdot\|_{\psi_1}$ imply that
	\begin{equation}\label{eq_T1}
	\EE\left[\left|\e_k^\top (\Theta^\top \Theta)^{-1}\C A_{J \sbt}^\top W_{tJ}\oW_{t\sbt}^\top \beta\right| ^3\right] \le c\gamma_w^6\left(\|\beta\|^2 \over m\right)^{3/2}\left(
	\e_k^\top (\Theta^\top \Theta)^{-1}\e_k
	\right)^{3/2}.
	\end{equation}
	Similarly,   observe that 
	$\e_k^\top \Theta^+W_{t\sbt}$ is $\gamma_w\sqrt{e_k^\top (\Theta^\top \Theta)^{-1}e_k}$-sub-Gaussian by Lemma \ref{lem_W}, and deduce
	\[ \left\|\e_k^\top \Theta^+ \W_{t\sbt}\eps_t \right\|_{\psi_1}  \le
	\left\|\e_k^\top \Theta^+ \W_{t\sbt}\right\|_{\psi_2}\|\eps_t\|_{\psi_2} \le c\gamma_{\eps}\gamma_w\left(\e_k^\top (\Theta^\top \Theta)^{-1}\e_k\right)^{3/2}
	\]
	so that
	\begin{equation}\label{eq_T2}
	\EE\left[\left|\e_k^\top \Theta^+\W_{t\sbt}\eps_t\right|^3\right] \le c\gamma_{\eps}^3\gamma_w^3\left(\e_k^\top (\Theta^\top \Theta)^{-1}\e_k\right)^{3/2}.
	\end{equation}
	Finally, we bound the third term. The independence of $Z$, $W$ and $\eps$ guarantees
	\begin{equation*}
	\EE\left[\left|\Omega_{k\sbt }^\top \Z_{t\sbt}(\eps_t-\langle \oW_{t\sbt}, \beta\rangle\right|^3\right] \le \EE\left[\left|\Omega_{k\sbt }^\top \Z_{t\sbt}\right|^3\right]\EE\left[\left|\eps_t-\langle \oW_{t\sbt}, \beta\rangle\right|^3\right].
	\end{equation*}
	Observe that $\|\eps_t - \langle \oW_{t\sbt}, \beta\rangle\|_{\psi_2}\le c(\gamma_{\eps} + \|\beta\|\w/\sqrt{m})$ and part (1) of Lemma \ref{lem_Z} imply that
	$\langle \Omega_{k\sbt }, \Z_{t\sbt}\rangle$ is $(\z \sqrt{\Omega_{kk}})$-sub-Gaussian.  The definition of the Orlicz $\psi_2$ norm   implies 
	\begin{equation}\label{eq_T3}
	\EE\left[\left|\Omega_{k\sbt }^\top \Z_{t\sbt}\right|^3\right] \EE\left[\left|\eps_t-\langle \oW_{t\sbt}, \beta\rangle\right|^3\right] \le c \left(\gamma_{\eps} +  {\t \|\beta\| \over \sqrt m} \right)^3\Omega_{kk}^{3/2}.
	\end{equation}
	Collecting (\ref{eq_T1}) -- (\ref{eq_T3}) and invoking Lemma \ref{lem_var_UI} concludes the proof. \qed
	
	\bigskip

	\begin{lemma}\label{lem_var_UI}
		Under conditions of Theorem \ref{thm_distr}, we have
		\[
		\EE\left[\left| \e_k^\top (\Theta^\top \Theta)^{-1}\C A_{I \sbt}^\top U_{I\sbt}^{(t)}\beta \right|^3\right] \lesssim \gamma_w^6 \left(
		\|\beta\|^2_2 \over m
		\right)^{3/2} \left(
		\e_k^\top (\Theta^\top \Theta)^{-1}\e_k
		\right)^{3/2}
		\]
		for any $1\le t\le n$.
	\end{lemma}
	
	\begin{proof}
		From (\ref{disp_A_IU_I}), we have  
		\begin{align*}
		&\e_k^\top (\Theta^\top \Theta)^{-1}\C A_{I \sbt}^\top U_{I\sbt}^{(t)}\beta\\ 
		&= \sumsum_{a,b} \e_k^\top (\Theta^\top \Theta)^{-1}[\C]_{a\sbt} \beta_b[A_{I \sbt}^\top U_{I\sbt}^{(t)}]_{ab}\\
		&=  \sum_{a=1}^K \e_k^\top (\Theta^\top \Theta)^{-1}[\C]_{a\sbt} \left( {\beta_a\over m_a-1}\sumsum_{j \ne \ell \in I_a} \W_{tj}\W_{t\ell } +  \sumsum_{b\ne a} {\beta_b\over m_b}\sumsum_{j \in I_a,  \ell \in I_b} \W_{tj}\W_{t\ell }\right)\\
		& =  \sum_{a=1}^K \e_k^\top (\Theta^\top \Theta)^{-1}[\C]_{a\sbt} \left( {\beta_a\over m_a-1}\sumsum_{j \ne \ell \in I_a} \W_{tj}\W_{t\ell } - {\beta_a\over m_a}\sumsum_{j,\ell \in I_a} \W_{tj}\W_{t\ell}+  \sum_{b=1}^K \beta_b\sum_{j \in I_a} \W_{tj}\oW_{tb}\right)\\
		& =\sum_{a=1}^K \e_k^\top (\Theta^\top \Theta)^{-1}[\C]_{a\sbt} \beta_a\left( {1\over m_a(m_a-1)}\sumsum_{j \ne \ell \in I_a} \W_{tj}\W_{t\ell } - {1\over m_a}\sum_{j\in I_a} \W_{tj}^2\right)\\
		&\quad + \sum_{a=1}^K \e_k^\top (\Theta^\top \Theta)^{-1}[\C]_{a\sbt}\sum_{j \in I_a} \W_{tj} \beta^\top \oW_{t\cdot }.
		\end{align*}
		By using 
		\[
		\sum_{a=1}^K \e_k^\top (\Theta^\top \Theta)^{-1}[\C]_{a\sbt}\sum_{j \in I_a} \W_{tj}  = \e_k^\top (\Theta^\top \Theta)^{-1}\C A_{I\sbt}^\top \W_{tI} = \e_k^\top (\Theta^\top \Theta)^{-1}\Theta_{I\sbt}^\top \W_{tI},
		\]
		after a bit algebra, we obtain
		\begin{align*}
		&\e_k^\top (\Theta^\top \Theta)^{-1}\C A_{I \sbt}^\top U_{I\sbt}^{(t)}\beta\\ 
		&=\sum_{a=1}^K \e_k^\top (\Theta^\top \Theta)^{-1}[\C]_{a\sbt} \beta_a\left( {1\over m_a(m_a-1)}\sum_{j, \ell \in I_a} \W_{tj}\W_{t\ell } - {m_a-2\over m_a(m_a-1)}\sum_{j\in I_a} \W_{tj}^2\right)\\
		&\quad +  \e_k^\top (\Theta^\top \Theta)^{-1}\Theta_{I\sbt}^\top \W_{tI} \beta^\top \oW_{t\cdot }
		\end{align*}
		such that
		\begin{align}\label{disp_AIUI}
		&\left|\e_k^\top (\Theta^\top \Theta)^{-1}\C A_{I \sbt}^\top U_{I\sbt}^{(t)}\beta\right|\\\nonumber
		&\le \left|\sum_{a=1}^K \e_k^\top (\Theta^\top \Theta)^{-1}[\C]_{a\sbt} \beta_a {m_a\over m_a-1}\oW_{ta}^2\right| + \left|\sum_{a=1}^K \e_k^\top (\Theta^\top \Theta)^{-1}[\C]_{a\sbt} \beta_a {m_a-2\over m_a(m_a-1)}\sum_{j\in I_a} \W_{tj}^2\right|\\\nonumber
		&\quad +  |\e_k^\top (\Theta^\top \Theta)^{-1}\Theta_{I\sbt}^\top \W_{tI} \beta^\top \oW_{t\sbt }|.
		\end{align}
		We now   bound the third moment of each term on the right. For the last term,  using the inequality $\|XY\|_{\psi_1} \le \|X\|_{\psi_2} \|Y\|_{\psi_2}$ yields
		\[
		\left\|\e_k^\top (\Theta^\top \Theta)^{-1}\Theta_{I\sbt}^\top \W_{tI}\beta^\top \oW_{t\cdot }\right\|_{\psi_1} \le \left\|\e_k^\top (\Theta^\top \Theta)^{-1}\Theta_{I\sbt}^\top \W_{tI}\right\|_{\psi_2}\left\|\beta^\top \oW_{t\cdot }\right\|_{\psi_2}.
		\]
		Recall that, by Lemma \ref{lem_W}, 
		\[
		\left\|\beta^\top \oW_{t\cdot }\right\|_{\psi_2} \le c{\gamma_w\|\beta\|_2 \over \sqrt{m}},\quad \left\|\e_k^\top (\Theta^\top \Theta)^{-1}\Theta_{I\sbt}^\top \W_{tI}\right\|_{\psi_2}\le c\gamma_w\sqrt{\e_k^\top (\Theta^\top \Theta)^{-1}\e_k} 
		\]
		after using $\e_k^\top (\Theta^\top \Theta)^{-1}\Theta_{I \sbt}^\top \Theta_{I \sbt}(\Theta^\top \Theta)^{-1}\e_k \le \e_k^\top (\Theta^\top \Theta)^{-1}\e_k$. The definition of the Orlicz $\psi_1$-norm gives 
		\begin{equation}\label{eq_T0_1}
		\EE\left[\left|\e_k^\top (\Theta^\top \Theta)^{-1}\Theta_{I\sbt}^\top \W_{tI}\beta^\top \oW_{t\cdot }\right|^3\right] \le c \gamma_w^6 \left(
		\|\beta\|_2^2 \over m
		\right)^{3/2} \left(
		\e_k^\top (\Theta^\top \Theta)^{-1}\e_k.
		\right)^{3/2}
		\end{equation}
		It remains to study the first two terms in (\ref{disp_AIUI}).
		By using the independence between $\oW_{ta}$, $a\in [K]$ and $\W_{tj}$,  $j\in I$, and by further writing, for each $a\in [K]$,
		\begin{align*}
		\alpha_a &= \e_k^\top (\Theta^\top \Theta)^{-1}[\C]_{a\sbt} \beta_a {m_a\over m_a-1},\\
		\wt \alpha_a & = \e_k^\top (\Theta^\top \Theta)^{-1}[\C]_{a\sbt} \beta_a {m_a-2\over m_a(m_a-1)},
		\end{align*}
		two applications of Rosenthal's inequality \citep{rosenthal} give
		\begin{align*}
		&\EE\left[
		\left|\sum_{a=1}^K\alpha_a \oW_{ta}^2\right|^3
		\right] \le c'\left\{
		\sum_{a=1}^K |\alpha_a|^3 \EE[\oW_{ta}^6] + \left(
		\sum_{a = 1}^K \alpha_a^2 \EE[\oW_{ta}^4]
		\right)^{3/2}
		\right\},\\
		& \EE\left[
		\left|\sum_{a=1}^K \sum_{j \in I_a} \wt\alpha_a \W_{tj}^2\right|^3
		\right] \le c'\left\{
		\sum_{a = 1}^K\sum_{j \in I_a}|\wt\alpha_a|^3 \EE[\W_{tj}^6] + \left(
		\sum_{a = 1}^K\sum_{j \in I_a}\wt\alpha_a^2 \EE[\W_{tj}^4] 
		\right)^{3/2}
		\right\}
		\end{align*}
		for some absolute constant $c'>0$. 
		The definition of the Orlicz $\psi_2$-norm and the inequalities  $\|\oW_{ta}\|_{\psi_2} \le c\gamma_w / \sqrt{m}$ and $\| \W_{tj}\|_{\psi_2}\le c\gamma_w$ from Lemma \ref{lem_W} imply 
		\[
		\EE[\oW_{ta}^6] \le c\gamma_w^6 / m^3,\quad \EE[\oW_{ta}^4] \le c\gamma_w^4 / m^2, \quad \EE[ \W_{tj}^6]\le c\gamma_w^6,\quad \EE[ \W_{tj}^4]\le c\gamma_w^4.
		\]
		Furthermore, we have
		\begin{align}\label{disp_square}\nonumber
		\sum_{a=1}^K \left(\e_k^\top (\Theta^\top \Theta)^{-1}[\C]_{a\sbt}\right)^2 &\le  \sum_{a=1}^K \left(\e_k^\top (\Theta^\top \Theta)^{-1}[\C]_{a\sbt}m_a [\C]_{a\sbt}^\top (\Theta^\top \Theta)^{-1}\e_k\right)^2\\\nonumber
		&=\e_k^\top (\Theta^\top \Theta)^{-1}\C A_{I \sbt}^\top A_{I \sbt} \C(\Theta^\top \Theta)^{-1}
		\e_k \\\nonumber
		& = \e_k^\top (\Theta^\top \Theta)^{-1}\Theta_{I \sbt}^\top \Theta_{I \sbt} (\Theta^\top \Theta)^{-1}\e_k\\ 
		&\le \e_k^\top (\Theta^\top \Theta)^{-1}\e_k,
		\end{align}
		and  we obtain 
		\begin{align*}
		\sum_{a=1}^K |\alpha_a|^3 \EE[\oW_{ta}^6]  & \le
		c{\gamma_w^6\over m^3} \left(\sum_{a = 1}^K |\alpha_a|\right)^3 \\
		&\le  	c{\gamma_w^6\over m^3} \left(\sum_{a = 1}^K |\beta_a|\cdot  \left|\e_k^\top (\Theta^\top \Theta)^{-1}[\C]_{a\sbt}\right|\right)^3\\
		&\le c{\gamma_w^6} \left(
		\|\beta\|_2^2 \over m
		\right)^{3/2}\left(
		\e_k^\top (\Theta^\top \Theta)^{-1}\e_k
		\right)^{3/2}
		\end{align*}
		where we use $m\ge 2$, the Cauchy-Schwarz inequality and (\ref{disp_square}) in the last line.  Moreover, by the same reasoning, we find 
		\begin{align*}
		\sum_{a = 1}^K \alpha_a^2 \EE[\oW_{ta}^4] &\le c{\gamma_w^4 \over m^2}\sum_{a = 1}^K \beta_a^2 \left(\e_k^\top (\Theta^\top \Theta)^{-1}[\C]_{a\sbt}\right)^2 \le c\gamma_w^4 {\|\beta\|_2^2 \over m}\e_k^\top [\Theta^\top \Theta]^{-1}\e_k.
		\end{align*}
		Combine the previous displays to conclude that 
		\begin{equation}\label{eq_T0_2}
		\EE\left[
		\left|\sum_{a=1}^K\alpha_a \oW_{ta}^2\right|^3
		\right]  \le c'c \gamma_w^6 \left(
		\|\beta\|_2^2 \over m
		\right)^{3/2} \left(
		\e_k^\top (\Theta^\top \Theta)^{-1}\e_k
		\right)^{3/2}.
		\end{equation}
		By similar arguments, it is easy to show that
		\begin{align*}
		\sum_{a = 1}^K\sum_{j \in I_a}|\wt\alpha_a|^3 \EE[\W_{tj}^6]   & \le
		c{\gamma_w^6} \left(\sum_{a = 1}^K|\e_k^\top (\Theta^\top \Theta)^{-1}[\C]_{a\sbt}\beta_a|\right)^3 \\
		&\le  	c{\gamma_w^6} \left(\sum_{a = 1}^K {|\beta_a| \over \sqrt m_a}\cdot  |\e_k^\top (\Theta^\top \Theta)^{-1}[\C]_{a\sbt}\sqrt{m_a}|\right)^3\\
		&\le c{\gamma_w^6} \left(
		\|\beta\|_2^2 \over m
		\right)^{3/2}\left(
		\e_k^\top (\Theta^\top \Theta)^{-1}\e_k
		\right)^{3/2}
		\end{align*}
		and 
		\begin{align*}
		\sum_{a = 1}^K\sum_{j \in I_a}\wt\alpha_a^2 \EE[\W_{tj}^4]  &\le c\gamma_w^4\sum_{a = 1}^K \beta_a^2 \left(\e_k^\top (\Theta^\top \Theta)^{-1}[\C]_{a\sbt}\right)^2 \le c\gamma_w^4 {\|\beta\|_2^2 \over m}\e_k^\top (\Theta^\top \Theta)^{-1}\e_k.
		\end{align*}
		Consequently,
		\begin{equation}\label{eq_T0_3}
		\EE\left[
		\left|\sum_{a=1}^K \sum_{j \in I_a} \wt\alpha_a \W_{tj}^2\right|^3
		\right]  \le c\gamma_w^6 \left(
		\|\beta\|_2^2 \over m
		\right)^{3/2} \left(
		\e_k^\top (\Theta^\top \Theta)^{-1}\e_k
		\right)^{3/2}.
		\end{equation}
		Combination of the bounds (\ref{eq_T0_1}), (\ref{eq_T0_2}) and (\ref{eq_T0_3}) concludes the proof. 
	\end{proof}

	\subsubsection{Proof of Lemma \ref{lem_Rem_2}}
	We work on the event $\E$ defined in (\ref{def_event}) that has probability at least $1-\pn^{-c}$ for some $c>0$. Recall that it implies $\wh K = K$ and $I_k \subseteq \wh I_k \subseteq I_k \cup J_1^k$ for all $k\in [K]$. 
	
	Fix any $k\in K$. Definition  (\ref{def_Rem_2}) yields
	\begin{align*}
	\left|[{\rm Rem}_2]_k \right| &\le {1\over n}\left|\e_k^\top \Omega \Z^\top \Z\D \beta\right|+	{1\over n}\left|\e_k^\top \Theta^+
	\W^\top \Z\D \beta\right|+ {1\over n}\left|\e_k^\top \Omega \Z^\top (\wt \W-\oW)\beta\right| \\\nonumber
	& \quad + \left|
	\e_k^\top \Theta^+\left(
	\wh\Gamma_{{\sbt \wh I}} - 	\wt\Gamma_{{\sbt \wh I}}\Pi- {1\over n}\W^\top (\wt \W- \bar \W)
	\right)\beta\right|.
	\end{align*}
	First, recall $\Theta^+ = \Omega^{1/2}H^+$, and observe that Lemma \ref{lem_quad} with $u=\Omega^{1/2}\e_k$ and Lemma \ref{lem_deltaA_I} imply
	\[
	{1\over n}\left|\e_k^\top \Omega \Z^\top \Z\D \beta \right|+	{1\over n}\left|\e_k^\top \Theta^+
	\W^\top \Z\D \beta\right| \lesssim 
	\left(\sqrt{\Omega_{kk}} + \sqrt{\e_k^\top (\Theta^\top \Theta)^{-1}\e_k}\right) \|D_\rho\beta\|_1 \delta_n
	\]
	with probability tending to one.
	Next, we bound  $n^{-1}|\e_k^\top \Omega \Z^\top (\wt \W-\oW)\beta|$. From the identity (\ref{eq_diff_W_td_W_bar}) for $\wt\W-\oW$, we need to   bound  
	\begin{align*}
	&{1\over n}\left|\e_k^\top \Omega \Z^\top \W_{\sbt L}\wh A_{L\sbt}D_{\wh m}\beta \right| + {1\over n}\left|\e_k^\top \Omega \Z^\top  \oW D_{\rho}\beta\right| \\
	&\le \max_{i\in J_1}{1\over n}\left|\e_k^\top \Omega \Z^\top \W_{\sbt i} \right|\cdot \|\wh A_{L\sbt}D_{\wh m}\beta\|_1 + {1\over n}\left\|\e_k^\top \Omega \Z^\top  \oW\right\|_\i
	\|D_{\rho}\beta\|_1\\
	&\le \bar{\rho}\|\beta\|_2\left(\max_{i\in J_1}{1\over n} \left| \e_k^\top \Omega \Z^\top \W_{\sbt i}\right|+ {1\over n}\left\|\e_k^\top \Omega \Z^\top  \oW\right\|_\i \right)
	\end{align*}
	The last inequality uses  (\ref{disp_D_rho}) with $v = \beta$. Observe that   $\e_k^\top \Omega \Z_{t\sbt}$ is $\z\sqrt{\Omega_{kk}}$-sub-Gaussian by Lemma \ref{lem_Z}, and apply Lemmas \ref{lem_W} and \ref{lem_bernstein} together with a union bound to conclude 
	\begin{align*}
	\max_{i\in J_1}{1\over n}
	\left| \e_k^\top \Omega \Z^\top \W_{\sbt i}\right|+ {1\over n}\left\|\e_k^\top \Omega \Z^\top  \oW\right\|_\i \lesssim  \sqrt{\Omega_{kk}}\delta_n.
	\end{align*}
	Finally, from using identity (\ref{eq_diff_W_td_W_bar}) again, we find that
	\begin{align*}
	&\left|\e_k^\top \Theta^+\left(
	\wh\Gamma_{{\sbt \wh I}} - 	\wt\Gamma_{{\sbt \wh I}}\Pi- n^{-1}\W^\top (\wt \W- \bar \W)
	\right)\beta\right|\\
	&=\left|\e_k^\top \Theta^+\left(
	\wh\Gamma_{{\sbt \wh I}} - 	\wt\Gamma_{{\sbt \wh I}}\Pi- {1\over n}\W^\top \W_{\sbt L}\wh A_{L\sbt}D_{\wh m} +   {1\over n}\W^\top \oW D_{\rho}
	\right)\beta\right|\\
	&\le  \left|\e_k^\top \Theta^+\left(
	\wh\Gamma_{{\sbt \wh I}} - 	\wt\Gamma_{{\sbt \wh I}}\Pi- \Gamma_{\sbt L}
	\wh A_{L\sbt}D_{\wh m} + \Gamma_{\sbt I}\Pi D_{\rho}
	\right)\beta\right|\\
	&\quad + \left|\e_k^\top \Theta^+\left({1\over n}\W^\top \oW- \Gamma_{\sbt I}\Pi\right)D_{\rho}\beta\right| + \left|\e_k^\top \Theta^+\left(
	{1\over n}\W^\top \W_{\sbt L} - \Gamma_{\sbt L}
	\right)\wh A_{L\sbt}D_{\wh m} \beta\right|.
	\end{align*}
	Close inspection of the proof of (\ref{rate_ww}), by taking $\alpha = \e_k^\top \Theta^+$ and $v = \beta$, we can bound the last two terms  above by 
	\[
	\bar{\rho}\|\beta\|_2 \delta_n \sqrt{\e_k^\top (\Theta^\top \Theta)^{-1}\e_k}
	\]
	with probability $1-\pn^{-c}$. It remains to   bound $|\e_k^\top \Theta^+ \D' \beta|$ with 
	\[
	\D' = 
	\wh\Gamma_{{\sbt \wh I}} - 	\wt\Gamma_{{\sbt \wh I}}\Pi- \Gamma_{\sbt L}
	\wh A_{L\sbt}D_{\wh m} + \Gamma_{\sbt I}\Pi D_{\rho}.
	\]
	Observe that 
	\begin{align*}
	\D'_{ik} &= {\wh \Gamma_{ii} - \Gamma_{ii} \over \wh m_k},\quad   i\in L_k\\
	\D'_{ik} &= {|L_k| \over m_k \wh m_k}\left( \Gamma_{ii} - \wh \Gamma_{ii}\right) + {\wh \Gamma_{ii} -\wt \Gamma_{ii} \over m_k},\quad   i\in I_k
	\\
	\D'_{ik} &= 0, \quad \text{otherwise.}
	\end{align*}
	On the event $\E_W$ defined in (\ref{def_event_E_w}), we find
	\begin{align*}
	|\e_k^\top \Theta^+ \D' \beta| &\le \left|
	\sum_{a=1}^K\sum_{i \in I_a \cup L_a} \e_k^\top (\Theta^\top \Theta)^{-1}\Theta_{i\sbt} \D'_{ia}\beta_a
	\right|	\\
	&
	\le \max_{i\in I\cup J_1} \left|\e_k^\top (\Theta^\top \Theta)^{-1}\Theta_{i\sbt} \right| \cdot  \sum_{a=1}^K\left(
	\sum_{i\in L_a} |\D'_{ia}\beta_a| + \sum_{i\in I_a} |\D'_{ia}\beta_a|
	\right)
	\\
	&\lesssim \sqrt{\e_k^\top (\Theta^\top \Theta)^{-1}\e_k}\left(
	\|D_\rho \beta\|_1 \delta_n + \sum_{a=1}^K |\beta_a|\max_{i\in I_a} |\wh \Gamma_{ii} -\wt \Gamma_{ii}|
	\right).
	\end{align*}
	Since Lemma \ref{lem_diag_gamma_hat_td} gives 
	\[
	\max_{i\in I_a} |\wh \Gamma_{ii} -\wt \Gamma_{ii}| = O_p\left({|J_1^a| \over |J_1^a| + m_a} \delta_n\right),
	\]
	for all $a\in [K]$, 
	we conclude 
	\[
	|\e_k^\top \Theta^+ \D' \beta| = O_p\left(\bar{\rho}\|\beta\|_2 \delta_n \sqrt{\e_k^\top (\Theta^\top \Theta)^{-1}\e_k}\right).
	\]
	Finally, we complete the proof by  collecting all bounds,   using the bound (\ref{disp_ThetaTheta_kk}) for $\e_k^\top (\Theta^\top \Theta)^{-1}\e_k$ and the inequality  $\lambda_K(H^\top H) \ge \lambda_K(A_{I \sbt}^\top A_{I \sbt})\cl \gtrsim 1$. \qed
	
	\bigskip

	\begin{lemma}\label{lem_diag_gamma_hat_td}
		Let	$\wh\Gamma_{ii}$ and $\wt \Gamma_{ii}$ be defined in (\ref{est_Gamma_I}) and (\ref{def_bar_Sigmaw}), respectively.	Under the conditions of Theorem \ref{thm_distr}, we have
		\[
		\max_{i\in I_a} |\wh \Gamma_{ii} -\wt \Gamma_{ii}| = O_p\left({|J_1^a| \over |J_1^a| + m_a}\delta_n\right),\quad \text{for all }a\in [K].
		\]
	\end{lemma}
	\begin{proof}
		We work on the event $\E$ defined in (\ref{def_event}) that has probability at least $1-\pn^{-c}$ for some $c>0$. Recall that it implies $\wh K = K$ and $I_k \subseteq \wh I_k \subseteq I_k \cup J_1^k$ for all $k\in [K]$. 
		
		First recall the definitions of $\whC$ and $\bar{\Sigma}^z$. 	For any $a,b\in [K]$,
		\begin{equation}\label{Chat_simp}
		\left[\whC\right]_{aa} = 
		\frac{1}{\wh m_a(\wh m_a-1)}\sumsum_{i, j\in \wh I_a,\ i\ne j}{1\over n}\X_{\sbt i}^\top \X_{\sbt j},
		\quad \left[\whC\right]_{ab} = 
		\frac{1}{\wh m_a\wh m_b}\sum_{i\in \wh I_a} \sum_{ j\in \wh I_b} {1\over n}\X_{\sbt i}^\top \X_{\sbt j},
		\end{equation}
		\begin{equation}\label{Cbar}
		\left[\bar \Sigma_Z\right]_{aa} = 
		\frac{1}{m_a(m_a-1)}\sumsum_{i, j\in I_a,\ i\ne j}{1\over n}\X_{\sbt i}^\top \X_{\sbt j},
		\quad \left[\bar \Sigma_Z\right]_{ab} = 
		\frac{1}{m_am_b}\sum_{i\in I_a} \sum_{ j\in I_b} {1\over n}\X_{\sbt i}^\top \X_{\sbt j}
		\end{equation}
		
		Choose any $a\in [K]$ and $i\in I_a$. Since $\wh\Gamma_{ii} - \wt\Gamma_{ii} = [\whC]_{aa} - [\bar \Sigma_Z]_{aa}$, we have
		\begin{align*}
		&	\left[\whC\right]_{aa} - \left[\bar \Sigma_Z\right]_{aa} =\\ &{1\over \wh m_a (\wh m_a-1)}\sumsum_{i,j\in \wh I_a,\ i\ne j}{1\over n}\left(
		A_{i\sbt}^\top \Z^\top \Z A_{j\sbt} + A_{i\sbt}^\top \Z^\top \W_{\sbt j} + \W_{\sbt i}^\top \Z A_{j\sbt}  + \W_{\sbt i}^\top \W_{\sbt j}
		\right)\\
		&\quad - {1\over m_a (m_a-1)}\sumsum_{i,j\in  I_a,\ i\ne j}{1\over n}\left(
		A_{i\sbt}^\top \Z^\top \Z A_{j\sbt} + A_{i\sbt}^\top \Z^\top \W_{\sbt j} + \W_{\sbt i}^\top \Z A_{j\sbt}  + \W_{\sbt i}^\top \W_{\sbt j}
		\right).
		\end{align*}
		Use $\Z A_{i \sbt} = \Z_{\sbt a}$ for all $i\in I_a$ and $\wh m_a = m_a + |L_a|$, to argue that 
		\begin{align*}
		&\sumsum_{i, j\in \wh I_a,\ i\ne j}{1\over n}
		A_{i\sbt}^\top \Z^\top \Z A_{j\sbt}  - \sumsum_{i, j\in  I_a,\ i\ne j}{1\over n}
		A_{i\sbt}^\top \Z^\top \Z A_{j\sbt}\\
		&= \left(
		2m_a\sum_{j\in L_a} + \sumsum_{i, j\in L_a,\ i\ne j}
		\right)
		\left\{{1\over n}A_{i \sbt}^\top \Z^\top \Z A_{j \sbt} - \left[\wh m_a(\wh m_a - 1) - m_a(m_a-1)\right]{1\over n}\Z_{\sbt a}^\top \Z_{\sbt a}	\right\}\\
		&= \left(
		2m_a\sum_{j\in L_a} + \sumsum_{i, j\in L_a, \ i\ne j}
		\right)
		\left\{{1\over n}(A_{i \sbt}^\top \Z^\top \Z A_{j \sbt} - \Z_{\sbt a}^\top \Z_{\sbt a})\right\}\\
		&\le
		2\left\{ 2m_a|L_a| + |L_a|(|L_a|-1)\right\} \sum_{i\in I_a \cup J_1^a}\left|{1\over n} \Z_{\sbt a}^\top \Z (A_{i\sbt} - \e_a)\right|.
		\end{align*}
		This implies 
		\begin{align*}
		&\left|{1\over \wh m_a (\wh m_a-1)}\sumsum_{i, j\in \wh I_a,\ i\ne j}{1\over n}
		A_{i\sbt}^\top \Z^\top \Z A_{j\sbt}  - {1\over m_a (m_a-1)}\sumsum_{i, j\in  I_a, i\ne j}{1\over n}
		A_{i\sbt}^\top \Z^\top \Z A_{j\sbt}\right|\\ 
		&\le  {2|J_1^a| \over |J_1^a| + m_a}\sum_{i\in I_a \cup J_1^a}\left|{1\over n} \Z_{\sbt a}^\top \Z (A_{i\sbt} - \e_a)\right|\\
		&\le{2|J_1^a| \over |J_1^a| + m_a} \frac{ 8 B_z}{\nu} \delta_n
		\end{align*}
		The last line follows from inequality (\ref{disp_1_1}), and holds with probability $1-C(\pn)^{-c}$.
		For the other two terms, after expanding $\wh I_a = I_a\cup L_a$ to $I_a \cup J_1^a$, we find that
		\begin{align*}
		&\left|{1\over \wh m_a (\wh m_a-1)}\sumsum_{i, j\in \wh I_a,\ i\ne j}{1\over n}
		A_{i\sbt}^\top \Z^\top \W_{\sbt j}  - {1\over m_a (m_a-1)}\sumsum_{i, j\in  I_a,\ i\ne j}{1\over n}
		A_{i\sbt}^\top \Z^\top \W_{\sbt j}\right|\\ 
		&\lesssim  
		{{|J_1^a| \over |J_1^a| + m_a}\max_{i,j\in I_a \cup J_1^a, i\ne j}\left|{1\over n}A_{i \sbt}^\top \Z_{\sbt a}^\top \W_{\sbt j}\right| }
		\end{align*}
		and 
		\begin{align*}
		&\left|{1\over \wh m_a (\wh m_a-1)}\sumsum_{i, j\in \wh I_a,\ i\ne j}{1\over n}
		\W_{\sbt i}^\top \W_{\sbt j}  - {1\over m_a (m_a-1)}\sumsum_{i,j\in  I_a,\ i\ne j}{1\over n}\W_{\sbt i}^\top \W_{\sbt j}\right|\\ 
		&\lesssim { {|J_1^a| \over |J_1^a| + m_a}\max_{i,j\in I_a \cup J_1^a, i\ne j}\left|{1\over n}\W_{\sbt i}^\top \W_{\sbt j}\right|.}
		\end{align*}
		Apply the exponential inequality Lemma \ref{lem_bernstein} and use the fact that
		$\|A_{i\sbt}\|_1\le 1$ to arrive at the desired result. 
	\end{proof}

	\section{Proof of Proposition \ref{prop_est_V}: consistent estimation of the asymptotic variance $V_k$}\label{sec_proofs_prop_V}
	
	\subsection{Main proof of  Proposition \ref{prop_est_V}}
	
	We only need to show $\wh V_{k}^{1/2}/V_k^{1/2} = 1+ o_p(1)$ as $n\to\infty$ since the rest of the proof is a direct consequence of   Theorem \ref{thm_distr} and  Slutsky's lemma. For ease of presentation, we only prove the result for simplified $V_k$ in (\ref{def_V_k_simp}) when $|I_1| = \cdots = |I_K| = m$ and $\tau_1^2 = \cdots = \tau^2_p = \tau^2$. The  general case with the complicated expression for $V_k$ in (\ref{def_V_k_general} can be proved by similar arguments. 
	
	We work on the event $\E$ defined in (\ref{def_event}) intersected with $\wh\Theta^\top \wh \Theta$ and $\whC$ are invertible. This event holds with probability $1-c\pn^{-c'}$ by Lemma \ref{lem_H_op} and Lemma \ref{lem_op_C}. Also recall that $\E$ implies $\wh K = K$ and $I_k \subseteq \wh I_k \subseteq I_k \cup J_1^k$ for all $k\in [K]$. 
	
	Write $\Omega :=\C^{-1}$ and $\wh \Omega = \whC^{-1}$.
	By a Taylor expansion, we have 
	\[
	\left|\wh V_{k}^{1/ 2} V_k^{-1/ 2} - 1\right| = V_k^{-1}\left|\wh V_{k} - V_k\right|(1+o_p(1)),
	\]
	and it suffices to show $V_k^{-1}|\wh V_{k} - V_k|= o_p(1)$. To this end, we first recall   
	\begin{equation}\label{def_Ds}
	V_k = \underbrace{\left(\sigma^2 + {\t^2\|\beta\|_2^2\over m}\right)}_{\Delta_a}\underbrace{\left( \Omega_{kk} + \tau^2\e_k^\top \left(\Theta^\top \Theta\right)^{-1}\e_k\right)}_{\Delta_b} +\underbrace{{\tau^4\over m-1}\sum_{a=1}^K{\beta_a^2 \over m}\sum_{i \in I_a} \left(\e_k^\top \Theta^+ \e_i\right)^2}_{\D_c}.
	\end{equation}
	The corresponding estimator is $\wh V_k = \wh \D_a \wh \D_b + \wh \D_c$ with 
	\[
	\wh \D_a := \wh\sigma^2 + \wh \tau^2{ \|\wh \beta\|_2^2 \over \wh m},\qquad \wh \D_b := \wh \Omega_{kk} + \wh \t^2 \e_k^\top (\wh \Theta^\top \wh \Theta)^{-1}\e_k
	\]
	and 
	\begin{equation}\label{disp_Delta_c}
	\wh \D_c := {\wh \tau^4\over \wh m-1}\sum_{a=1}^K{\wh \beta_a^2 \over \wh m}\sum_{i \in I_a} \left(\e_k^\top \wh \Theta^+ \e_i\right)^2
	\end{equation}
	writing $\wh m = |\wh I_k|$ and $\wh\tau^2 = \wh \tau^2_i$ for any   $k\in [K]$ and $i\in [p]$. Observe that
	\begin{equation}\label{disp_13}
	{|\wh V_{k} - V_k| \over V_k} \le  {|\wh \D_a - \D_a|\over \D_a} + {|\wh \D_b - \D_b|\over \D_b} + {|\wh \D_c - \D_c|\over \D_a\D_b},
	\end{equation} and
	we will bound each term on the right separately. Lemma \ref{lem_Delta_c} in Section \ref{sec_lemmas_var} guarantees that $|\wh \D_c - \D_c| = o_p(\D_a \D_b)$.
	From the definition of $\bar\rho$ in (\ref{def_rho}) and the inequality $|\wh m - m| / \wh m \le \bar{\rho}$, we have
	\begin{align*}
	&|\wh \D_a - \D_a |\\ 
	&\le |\wh \sigma^2 - \sigma^2| + \wh\tau^2\|\wh\beta\|_2^2{|m - \wh m| \over m\wh m} + {\|\beta\|_2^2 \over m}|\wh \t^2 - \t^2| + {\wh \t^2 \over m}(\|\wh \beta\|_2 + \|\beta\|_2)\|\wh \beta-\beta\|_2\\
	&\le |\wh \sigma^2 - \sigma^2| +  \wh\tau^2{\|\wh\beta\|_2^2 \over m}\bar{\rho} + {\|\beta\|_2^2 \over m}|\wh \t^2 - \t^2| + {\wh \t^2 \over m}(\|\wh \beta\|_2 + \|\beta\|_2)\|\wh \beta-\beta\|_2.
	\end{align*}
	Since  Lemma \ref{lem_signal} and Assumption \ref{ass_H_prime} imply 
	\[
	K\log\pn = o(\sqrt n),
	\]
	the rate of $\|\wh \beta-\beta\|_2$ in (\ref{rate_beta}) and the rate of $\max_{1\le i \le p}|\wh\tau_i^2 -\tau_i^2|$ in Lemma \ref{lem_tau} in Section \ref{sec_lemmas_var}  guarantee that, with probability tending to one, 
	\begin{equation}\label{eq_beta_tau}
	\|\wh\beta\|_2= O\left(1\vee \|\beta\|_2\right), \qquad  \wh \t^2 = \t^2 + o(1) = O(\t^2).
	\end{equation}
	Combine  Lemmas \ref{lem_tau} and \ref{lem_sigma} in Section \ref{sec_lemmas_var}   with the rate of $\|\wh \beta - \beta\|_2$ from (\ref{rate_beta}) to find
	\begin{align*}
	|\wh \D_a - \D_a |  = O\left( \left(1 \vee {\|\beta\|_2^2 \over m}\right)\left(\bar{\rho}+ \cl^{-1/2}
	\delta_n\sqrt{K} \right)\right)
	\end{align*}
	with probability tending to one.
	This bound with $K\log\pn = o(n)$ and Assumption \ref{ass_J1} and the definition  $\D_a= \sigma^2 + {\t^2\|\beta\|_2^2 /m} $ yields
	\begin{align}\label{bd_delta1}
	\D_a^{-1}|\wh \D_a - \D_a|& = O\left(  \bar\rho + \cl^{-1/2}
	\delta_n\sqrt{K}\right) = o(1)
	\end{align} 
	with probability tending to one.
	We proceed to bound $\D_b^{-1}|\wh \D_b -\D_b|$ by
	\begin{align*}
	|\wh \D_b - \D_b | &\le |\wh \Omega_{kk} - \Omega_{kk}| +|\wh \t^2 - \t^2|\cdot \e_k^\top (\Theta^\top \Theta)^{-1}\e_k + \wh \t^2\left|\e_k^\top \left[(\wh\Theta^\top \wh\Theta)^{-1} - (\Theta^\top \Theta)^{-1}\right]\e_k\right|,
	\end{align*}
	and study each term on the right separately. Write    $H := A \C^{1/2} = \Theta\Omega^{1/2}$ and
	$\wh H = \wh\Theta \Omega^{1/2}$.
	For the first term on the right, observe that the identity 
	\begin{equation}\label{eq_Omega_hat_Omega}
	\wh \Omega - \Omega = \Omega^{1/2}\left[\Omega^{1/2}(\whC - \C)\Omega^{1/2}\right](\Omega^{1/2}\whC\Omega^{1/2})^{-1}\Omega^{1/2},
	\end{equation}
	implies
	\begin{align*}
	|\wh \Omega_{kk} - \Omega_{kk}| &\le \Omega_{kk}{\|\Omega^{1/2}(\C- \whC)\Omega^{1/2}\|_{{\rm op}} \over \lambda_K(\Omega^{1/2}\whC \Omega^{1/2})}.
	\end{align*}
	By Weyl's inequality  
	\begin{equation}\label{disp_lb_OmegaCOmega}
	\lambda_K\left(\Omega^{1/2}\whC \Omega^{1/2}\right) = \lambda_K\left({\bI}_K - \Omega^{1/2}(\C - \whC)\Omega^{1/2}\right) \ge 1 - \|\Omega^{1/2}(\whC - \C) \Omega^{1/2}\|_{{\rm op}},
	\end{equation}
	and Lemma \ref{lem_op_C} in Section \ref{sec_lemmas_var} yields 
	\begin{equation*}
	|\wh\Omega_{kk} - \Omega_{kk}| = O\left(\Omega_{kk}\delta_n\sqrt{K}\right) = o(\D_b)
	\end{equation*}
	with probability tending to one.
	Next, invoke  Lemma \ref{lem_tau} in Section \ref{sec_lemmas_var} and the definition of $\Delta_b$ to obtain for 
	the second term on the right 
	\begin{equation*}
	|\wh \t^2 - \t^2|\cdot \e_k^\top (\Theta^\top \Theta)^{-1}\e_k  = o(\Delta_b) 
	\end{equation*}
	with probability tending to one.
	Finally, 
	obtain
	the identity 
	\begin{align*}
	&(\wh \Theta^\top \wh \Theta)^{-1} - (\Theta^\top \Theta)^{-1}\\ 
	&= (\Theta^\top \Theta)^{-1}(\Theta^\top \Theta - \wh\Theta^\top \wh \Theta)(\wh \Theta^\top \wh \Theta)^{-1} \\
	&= \Theta^+(\Theta-\wh \Theta)(\wh\Theta^\top \wh\Theta)^{-1} + (\Theta^\top \Theta)^{-1}(\Theta - \wh \Theta)^\top \wh\Theta(\wh\Theta^\top \wh\Theta)^{-1}\\
	&=  \Omega^{1/2}\left[H^+(H-\wh H)(\wh H^\top \wh H)^{-1} + (H^\top H)^{-1}(H - \wh H)^\top \wh H(\wh H^\top \wh H)^{-1}\right]\Omega^{1/2},
	\end{align*}
	and invoke Lemma \ref{lem_H_op} in Section \ref{sec_auxiliary_lemma} to conclude
	\begin{align*}
	&\left|\e_k^\top \left[ (\wh \Theta^\top \wh \Theta)^{-1} - (\Theta^\top \Theta)^{-1}\right] \e_k\right| 
	= O\left( {\Omega_{kk} \over \lambda_K(H^\top H)} \left(\delta_n\sqrt{K}+ {\delta_n\sqrt{pK} \over \sigma_K(H)}
	\right)\right) = o(\D_b)
	\end{align*}
	with probability tending to one.
	The last step uses  Assumption \ref{ass_H_prime}.  From (\ref{eq_beta_tau}), we get 
	\begin{align*}
	&\wh\tau^2 \left|\e_k^\top \left[ (\wh \Theta^\top \wh \Theta)^{-1} - (\Theta^\top \Theta)^{-1}\right] \e_k\right| 
	= o(\D_b)
	\end{align*} 
	with probability tending to one.
	Collecting all three bounds, we find 
	\begin{equation*}
	\D_b^{-1}|\wh \D_b - \D_b|  = o(1)
	\end{equation*} 
	with probability tending to one.
	This concludes 
	completes the proof.\qed
	
	\subsection{Lemmas used in the proof of Proposition \ref{prop_est_V}}\label{sec_lemmas_var}
	
	We state the lemmas used for proving Proposition \ref{prop_est_V}. Their proofs are deferred to Section \ref{sec_proof_lem_var}.
	All statements are valid on some events that are subsets of $\E$ and the  probabilities of these events are greater than $1-C\pn^{-\alpha}$ for some positive constants $C,\alpha$. This is an important observation since on the event $\E$, the dimensions  $\wh K$ and $K$ are equal, which ensures that the various quantities in the statements, for instance, the difference  $\C-\C$, are well-defined.

	\begin{lemma}\label{lem_tau}
		Let $\wh \tau_i^2$ be defined in (\ref{est_Gamma}). Under the   conditions of Theorem \ref{thm_I}, with probability greater than $1-\pn^{-c}$ for some constant $c>0$, we have 
		\begin{align*}
		\max_{1\le i \le p}|\wh \tau_i^2 - \tau_i^2| \lesssim \delta_n.
		\end{align*} 
	\end{lemma}

	\begin{lemma}\label{lem_sigma}
		Let $\wh\sigma^2$ be defined as in (\ref{est_sigma}). Under the conditions of Theorem \ref{thm_beta}, with probability $1-c\pn^{-c'}$,
		$$
		|\wh \sigma^2 - \sigma^2| \lesssim \cl^{-1}\left(1\vee {\|\beta\|_2^2 \over m}\right)\delta_n\sqrt{K}.
		$$
	\end{lemma}

	\begin{lemma}\label{lem_op_C}
		Under the conditions of Theorem \ref{thm_beta}, with probability  $1- c\pn^{-c'}$, 
		\[
		\left \|\C^{-1/2}(\whC-\C)\C^{-1/2}\right\|_{{\rm op}} \lesssim \delta_n\sqrt{K}.
		\]
	\end{lemma}

	\begin{lemma}\label{lem_Delta_c}
		Let $\D_a$, $\D_b$ and $\D_c$ be defined in (\ref{def_Ds}). Let $\wh \D_c$ be defined in (\ref{disp_Delta_c}). 
		Under the conditions of Proposition \ref{prop_est_V}, we have
		\begin{align*}
		|\wh \D_c - \D_c| = o_p(\D_a \D_b).
		\end{align*}
	\end{lemma}

	\subsection{Proof of lemmas in Section \ref{sec_lemmas_var}}\label{sec_proof_lem_var}
	
	\subsubsection{Proof of Lemma \ref{lem_tau}}
	We work on the event $\E$ defined in (\ref{def_event}) that has probability at least $1-\pn^{-c}$ for some $c>0$. Recall that it implies $\wh K = K$ and $I_k \subseteq \wh I_k \subseteq I_k \cup J_1^k$ for all $k\in [K]$. 
	
	Fix any $i\in [p]$. The decomposition  $\X_{\sbt i} = \Z A_{i\sbt}+\W_{\sbt i}$ readily gives
	\begin{align*}
	|\wh \tau_i^2 - \tau_i^2|			&\le \left|{1\over n}A_{i\sbt}^\top \Z^\top \Z A_{i\sbt} -\wh A_{i\sbt}^\top \whC \wh A_{i\sbt} \right| + \left|{1\over n}\W_{\sbt i}^\top \W_{\sbt i} - \tau_i^2 \right| + \left|{2\over n}A_{i\sbt}^\top \Z^\top \W_{\sbt i}\right|
	\end{align*}
	Since $\EE[n^{-1}\W_{\sbt i}^\top \W_{\sbt i}] = \tau_i^2$ and $|A_{i\sbt}^\top \Z^\top \W_{\sbt i}|\le \|\Z^\top \W_{\sbt i}\|_\i$ by $\|A_{i\sbt}\|_1\le 1$, an application of the exponential inequality in Lemma \ref{lem_bernstein} followed by  the union bound yields
	\begin{equation}\label{eq_R2R3}
	\PP\left\{\left|{1\over n}\W_{\sbt i}^\top \W_{\sbt i} - \tau_i^2 \right| +\left|{2\over n}A_{i\sbt}^\top \Z^\top \W_{\sbt i}\right| \lesssim \delta_n\right\}\ge 1-\pn^{-c}.
	\end{equation}
	For  the first term, we argue that, for any $i\in I_k$ and $k\in[K]$, $\wh A_{i\sbt} = A_{i\sbt}$ (on the event $\E$) so that
	\begin{align*}
	\left|{1\over n}A_{i\sbt}^\top \Z^\top \Z A_{i\sbt} -\wh A_{i\sbt}^\top \whC \wh A_{i\sbt} \right|= \left|{1\over n}\Z_{\sbt k}^\top \Z_{\sbt k}-[\C]_{kk}\right|.
	\end{align*}
	This can be bounded by applying (\ref{rate_zz_c}) in Lemma \ref{lem_quad} with $u = v = \e_k$. Taking the union bound concludes the proof of $\wh \tau_i^2 - \tau_i^2$ for $i\in I$.\\  For any $i\in J$, 
	\begin{align*}
	|\tau_i^2-\wh \tau_i^2 |= \left|{1\over n}A_{i\sbt}^\top \Z^\top \Z A_{i\sbt} -\wh A_{i\sbt}^\top \whC \wh A_{i\sbt} \right|&\le \left|(\wh A_{i\sbt}-A_{i\sbt})^\top \whC\wh A_{i\sbt} \right| +\left|(\wh A_{i\sbt} -A_{i\sbt})\whC A_{i\sbt} \right|\\
	&\quad  +\left| A_{i\sbt}^\top \left({1\over n}\Z^\top \Z-\whC\right) A_{i\sbt} \right|. 
	\end{align*}
	By the Cauchy-Schwarz inequality, $\|\wh A_{I\sbt}\|_1\le 1$,
	$\|A_{i\sbt}\|_1\le1$ and Lemma \ref{lem_A}, we obtain
	\[
	\left|(\wh A_{i\sbt} -A_{i\sbt})\whC A_{i\sbt} \right|+ \left|(\wh A_{i\sbt}-A_{i\sbt})^\top \whC\wh A_{i\sbt} \right| \le 2\|\whC(\wh A_{i\sbt} - A_{i\sbt})\|_\i \lesssim \delta_n
	\]
	with probability $1-\pn^{-c}$. 
	Regarding the third term, apply Lemma \ref{lem_1} to find that
	\begin{align*}
	\left| A_{i\sbt}^\top \left({1\over n}\Z^\top \Z-\whC\right) A_{i\sbt} \right| &\le 
	\|n^{-1} \Z^\top \Z-\whC\|_\i \\
	&\le \left\| n^{-1}\Z^\top \Z-\C \right\|_\i + \|\whC - \C\|_\i \lesssim \delta_n
	\end{align*}
	with probability $1-\pn^{-c}$.
	Combining these two displays completes our proof.
	\qed

	\bigskip
	
	\begin{lemma}\label{lem_A}
		Under the same conditions of Theorem  \ref{thm_I}, let $\whC$ be constructed as (\ref{Chat}) and $\wh A$ as in display (\ref{est_AJ}) together with (\ref{est_AI_a}) -- (\ref{est_AI_b}). Let $s_j = \|A_{j\sbt}\|_0$ for any $j\in[p]$. With probability greater than $1-\pn^{-c}$ for some constant $c>0$,
		\begin{equation*}
		\|\wh A_{j\sbt } - A_{j\sbt }\|_2  \lesssim \cl^{-1}\delta_n\sqrt{s_j},\quad  \|\whC(\wh A_{j\sbt } - A_{j\sbt })\|_\i \lesssim  \delta_n
		\end{equation*}
		hold for all $1\le j\le p$.
	\end{lemma}
	\begin{proof}
		Write $s:= s_j$. The rate of $\wh A_{j\sbt } - A_{j\sbt }$ follows immediately from 
		\cite[display (36) in Theorem 5]{LOVE} by observing that 
		\begin{align*}
		\kappa_2(\C, s) := \inf_{|S|\le s} \inf_{\substack{\|v\|=1\\v \in \mathcal {C}_S}} \| \C v \|_\i  
		&\ge \inf_{|S|\le s} \inf_{\substack{\|v\|=1\\\textrm{supp}(v) = S}} \| \C v \|_\i\\ 
		&\ge \inf_{|S|\le s} \inf_{\substack{\|v\|=1\\\textrm{supp}(v) = S}} \left\| [\C]_{SS}\cdot v_S\right\|_\i\\
		& \ge \inf_{|S|\le s} \inf_{\substack{\|v\|=1\\\textrm{supp}(v) = S}} \left\| [\C]_{SS}\cdot v_S\right\| / \sqrt{s}\\
		& \ge \cl/\sqrt{s}
		\end{align*}
		with $\mathcal {C}_S:=\{v\in \R^{K}:\| v_{S^c}\|_1\le \| v_{S}\|_1 \}$ and $S\subseteq [K]$ with $|S|\le s$. The proof of the bound for  $\|\whC(\wh A_{j\sbt } - A_{j\sbt })\|_\i$ follows from, writing $\wh \Pi:= 	\wh A_{{\wh I\sbt}} [\wh A_{{\wh I\sbt}}^\top \wh A_{{\wh I\sbt}}]^{-1}
		$, 
		\begin{align*}
		\|\whC(\wh A_{j\sbt } - A_{j\sbt })\|_\i &\le \|\whC\wh A_{j\sbt } - \wh \Pi^\top \wh \Sigma_{\wh I j}\|_\i +\|\whC A_{j\sbt } - \wh \Pi^\top \wh \Sigma_{\wh I j}\|_\i
		= O(\delta_n)
		\end{align*}
		by using the feasibility of both $\wh A_{j\sbt }$ and $A_{j\sbt }$.
	\end{proof}

	\subsubsection{Proof of Lemma \ref{lem_sigma}}
	We work on the event $\E$ defined in (\ref{def_event}) that has probability at least $1-\pn^{-c}$ for some $c>0$. Recall that it implies $\wh K = K$ and $I_k \subseteq \wh I_k \subseteq I_k \cup J_1^k$ for all $k\in [K]$. 
	
	Recall that, from the definition of $\wh h$ in (\ref{est_h}),
	$$\wh h ={1\over n}(\wh A_{{\wh I\sbt}}^\top \wh A_{{\wh I\sbt}})^{-1}\wh A_{{\wh I\sbt}}^\top \\X_{\sbt \wh I}^\top \y \overset{(\ref{def_tildes})}{=}{1\over n}\wt \X^\top \y.$$
	We observe that the definition of $\whC$ in (\ref{Chat}) yields
	\begin{equation}\label{eq_C_hat}
	\left[\whC\right]_{aa} = {1\over n}\wt \X_{\sbt a}^\top \wt \X_{\sbt a} - d_a,~  \forall a\in [K];\qquad \left[\whC\right]_{ab} = {1\over n}\wt \X_{\sbt a}^\top \wt \X_{\sbt b},~ \forall a\ne b\in [K]
	\end{equation}
	with 
	\begin{equation}\label{def_D}
	d_a= {1\over \wh m_a^2}\sum_{i\in \wh I_a} {1\over n}\X_{\sbt i}^\top \X_{\sbt i} - {1\over \wh m_a^2 (\wh m_a-1)}\sumsum_{i,j\in \wh I_a, i\ne j}{1\over n}\X_{\sbt i}^\top \X_{\sbt i}.
	\end{equation}
	Let $D = \textrm{diag}(d_1, \ldots, d_K)$ so that $\whC = n^{-1}\wt \X^\top \wt \X - D$, and define $\D_\beta:= \Z\beta - \wt \Z\wh \beta$.  
	Since $\y = \Z\beta + \eps$ and  $\wt \X = \wt \Z + \wt \W$, we find
	\begin{align*}
	\wh \sigma^2 &= {1\over n}\y^\top \y - {2\over n}\wh\beta^\top \wt \X^\top \y+ \wh \beta^\top \whC\wh \beta\\ 
	&= {1\over n}\|\y- \wt \X\wh \beta\|_2^2 - \wh \beta^\top \left({1\over n}\wt \X^\top \wt \X - \whC \right) \wh \beta\\
	&= {1\over n}\|\Z\beta + \eps - (\wt \Z+\wt \W)\wh \beta\|_2^2 - \wh \beta^\top D \wh \beta\\
	&= {1\over n}\|\D_\beta\|_2^2 + {1\over n}\eps^\top \eps + {2\over n}\eps^\top \D_\beta - {2\over n}\eps^\top  \wt \W\wh \beta - {2\over n}\wh\beta^\top \wt \W^\top \D_\beta+ \wh\beta^\top \left({1\over n}\wt \W^\top \wt \W - D\right)\wh \beta ,
	\end{align*}
	and consequently we have
	\begin{align*}
	|\wh \sigma^2 - \sigma^2| &\le \left|{1\over n}\eps^\top \eps-\sigma^2\right|  + \left\|{1\over n}\wt \W^\top \wt \W - D\right\|_{{\rm op}}\|\wh \beta\|^2 + {2\over n}|\eps^\top \wt \W^\top \wh\beta|\\
	&\quad + {1\over n}\|\D_\beta\|_2^2+ {2\over n}|\eps^\top  \D_\beta| + {2\over n}|\wh\beta^\top \wt \W^\top \D_\beta|.
	\end{align*}
	The rate of $\|\wh \beta -\beta\|_2$ in Theorem \ref{thm_beta} together with $K\log\pn = o(n)$ implies 
	\begin{equation}\label{wh_beta_2norm}
	\|\wh\beta\| \lesssim 1\vee \|\beta\|_2
	\end{equation}
	with probability $1-\pn^{-c}$.
	From Lemma \ref{lem_bernstein} and Lemma \ref{lem_wt_WW_D} with the same discretization arguments in the proof of Lemma \ref{lem_op_C}, we obtain
	\begin{equation}\label{disp_7}
	\left|{1\over n}\eps^\top \eps-\sigma^2\right| \lesssim \delta_n,\qquad \left\|{1\over n}\wt \W^\top \wt \W - D\right\|_{{\rm op}} \lesssim {\delta_n\sqrt{K} \over m},
	\end{equation}
	with probability greater than $1-\pn^{-c}.$ 
	Using similar arguments as in the proof of Lemma \ref{lem_Z_wt_W_wt}, by substituting  $Z_k$ by $\eps$, we find, with probability $1-\pn^{-c}$,
	\[
	{1\over n}|\eps^\top \wt \W v| \lesssim \left({\|v\|_2\over \sqrt{m}} + \rho \|v\|_1\right)\delta_n\sqrt{K}
	\]
	for any fixed $v\in \R^K$. By choosing $v = \beta$ and $v = \e_k$ for $k\in [K]$ and using the bound
	\[
	{1\over n}\left|\eps^\top \wt \W\wh\beta \right| \le {1\over n}\left|\eps^\top \wt \W\beta \right| + {1\over n}\|\eps^\top \wt \W\|_\i \sqrt{K}\|\wh\beta - \beta\|_2,
	\]
	we have, with probability $1-c\pn^{-c'}$,
	\begin{equation}\label{disp_8}
	{1\over n}\left|\eps^\top \wt \W\wh\beta\right| \lesssim \cl^{-1}\left(1\vee {\|\beta\|_2 \over \sqrt m}\right)\delta_n\sqrt{K}.
	\end{equation}
	We proceed to bound the remaining three terms involving  $\D_\beta$. Recall that 
	\[
	\D_\beta = \Z(\beta - \wh \beta) + (\Z - \wt \Z) \wh\beta = \Z(\beta - \wh \beta) + \Z\D\wh\beta
	\]
	with $\D$ defined in (\ref{def_deltaA_I}). 
	Observe that
	\begin{align*}
	{1\over n}\|\D_\beta\|_2^2 &\le {1\over n}\|\Z(\wh\beta - \beta)\|_2^2 + {1\over n}\|\Z^\top \Z\|_\i \|\D\|_{1,\i}^2\|\wh \beta\|_1^2,\\
	{1\over n}|\eps^\top  \D_\beta| & \le {1\over n}\|\eps^\top  \Z\|_\i \sqrt{K}\|\wh\beta-\beta\|_2 + {1\over n}\|\eps^\top \Z\|_\i \|\D\wh\beta\|_1,
	\end{align*}
	and  $\|\wh\beta\|_1 \le \sqrt{K}\|\wh \beta\|_2$.
	Hence, display (\ref{wh_beta_2norm}),
	Lemmas \ref{lem_bernstein}, \ref{lem_deltaA_I} and \ref{lem_fit} and Assumptions \ref{ass_J1} and \ref{ass_H} yield
	\begin{equation}\label{disp_9}
	{1\over n}\|\D_\beta\|_2^2 + {2\over n}|\eps^\top \D_\beta| \lesssim \cl^{-1}\left(1\vee {\|\beta\|_2^2 \over m}\right)\delta_n\sqrt{K}
	\end{equation}
	with probability at least $1-c\pn^{-c'}.$
	Finally, we obtain the bound 
	\begin{align*}
	{1\over n}|\wh\beta^\top \wt \W^\top \D|&\le {1\over n}|\beta^\top \wt \W^\top \D| + {1\over n}|(\wh\beta - \beta)^\top \wt \W^\top \D|\\
	&\le {1\over n}\|\beta^\top \wt \W^\top \Z\|_\i \sqrt{K}\|\wh\beta-\beta\|_2 + {1\over n}\|\beta^\top \wt \W^\top \Z\|_\i \|\D\wh \beta\|_1\\
	&\quad  + {1\over n}\|\wt \W^\top  \Z\|_\i \sqrt{K}\|\wh\beta - \beta\|_2\|\D\wh \beta\|_1+ {1\over n}|(\wh\beta - \beta)^\top \wt \W^\top  \Z(\wh\beta - \beta)|\\
	& = O_p\left(\cl^{-1/2}\left(1\vee {\|\beta\|_2^2 \over  m}\right)\delta_n \sqrt{K} \right)+{1\over n}|(\wh\beta - \beta)^\top \wt \W^\top  \Z(\wh\beta - \beta)|,
	\end{align*}
	using the same arguments as above, in combination  with Lemma \ref{lem_Z_wt_W_wt}.
	We control the final term by
	\begin{align}\label{disp_12}\nonumber
	{1\over n}|(\wh\beta - \beta)^\top \Z^\top \wt \W(\wh\beta - \beta)| &\le  \left\|\C^{1/2}(\wh\beta - \beta)\right\|_2^2 \sup_{v\in \S^{K-1}}{1\over n}\left|v^\top\Omega^{1/2} \Z^\top \wt \W \Omega^{1/2} v\right|\\
	&= O_p\left(\cl^{-1/2}\left(1\vee {\|\beta\|_2^2 \over  m}\right)\delta_n \sqrt{K} \right)
	\end{align}
	using   inequality (\ref{eq_T3'}), Lemma \ref{lem_fit} and  $K\log\pn = o(n)$.
	Collecting (\ref{disp_7}) -- (\ref{disp_12}) gives the desired result. \qed
	
	\bigskip
	The following lemmas are used in the proof of Lemma \ref{lem_sigma}.
	
	\begin{lemma}\label{lem_fit}
		Under conditions of Theorem \ref{thm_beta}, with probability $1-\pn^{-c}$ for some constant $c>0$, 
		\begin{align*}
		{1\over n}\|\Z \wh \beta - \Z\beta\|_2^2 \lesssim \left(1\vee {\|\beta\|_2^2 \over m}\right)\delta_n^2 K,\qquad \left\|\C^{1/2}(\wh\beta - \beta)\right\|_2^2 \lesssim \left(1\vee {\|\beta\|_2^2 \over m}\right)\delta_n^2 K .
		\end{align*}
	\end{lemma}
	
	\begin{proof}
		From (\ref{D_Z}), (\ref{D_W}) and  (\ref{diff_beta_hat_beta}),   $\wh\Theta\Omega^{1/2} = \wh H$ and $\Theta\Omega^{1/2} = H$, it follows that, on the event $\E$, 
		\begin{align*}
		{1\over n}\|\Z(\wh \beta - \beta)\|_2^2 &\le \|\D_1\|_2 + \|H^+\D_2\|_2 + \|H^+(\wh H-H)\D_1\|_2 + \|(\C)^{1/2}(\wh\Theta^+-\Theta^+)\D_2\|_{{\rm op}}.
		\end{align*}
		The first result follows immediately by invoking the upper bounds of each term in the proof of Theorem \ref{thm_beta}, except that there is no $\cl^{-1/2}$. The second result follows from Lemma \ref{lem_op_C} and the inequality, on the event $\E$,
		\[
		\|\Z(\wh \beta - \beta)\|_2 \le \|\Z \Omega^{1/2}\|_{\rm op}\|\C^{1/2}(\wh \beta-\beta)\|_2.
		\]
	\end{proof}
	
	\medskip
	
	\begin{lemma}\label{lem_Z_wt_W_wt}
		Let $u, v\in \R^K$ be any fixed vectors. Under the conditions of Theorem \ref{thm_I}, with probability $1-\pn^{-c}$, one has
		\begin{align}\label{rate_Zwt_Wwt}
		{1\over n}\left|u^\top \wt \Z^\top \wt \W v\right| &~\lesssim~ \left(\sqrt{u^\top \C u} + \bar\rho\|u\|_2\right)\left({\|v\|_2 \over \sqrt m} +   \bar\rho \|v\|_2\right)\delta_n.
		\end{align}
	\end{lemma}
	\begin{proof}
		We work on the event $\E$ defined in (\ref{def_event}) that has probability at least $1-\pn^{-c}$ for some $c>0$. Recall that it implies $\wh K = K$ and $I_k \subseteq \wh I_k \subseteq I_k \cup J_1^k$ for all $k\in [K]$. 
		
		Since 
		$
		\wt \Z = \Z + \Z\D.
		$
		from (\ref{def_deltaA_I}),
		we have
		\begin{align}\label{disp_ztd_wtd}\nonumber
		{1\over n}|u^\top  \wt \Z^\top \wt \W v| &= 	{1\over n}|u^\top  \Z^\top \wt \W v| + 	{1\over n}|u^\top  \D^\top  \Z^\top \wt \W v|\\
		& \le {1\over n}|u^\top  \Z^\top \wt \W v| + 	\max_{1\le k\le K}{1\over n}|\Z_k^\top \wt \W v| \cdot 
		\|\D u\|_1.
		\end{align}
		Repeated application of   Lemma \ref{lem_quad} in conjunction with Lemma \ref{lem_deltaA_I} gives 
		\begin{align*}
		{1\over n}|u^\top  \wt \Z^\top \wt \W v| &\lesssim \delta_n \sqrt{u^\top \C u} \left(
		{\|v\|_2 \over \sqrt m} + \bar{\rho}\|v\|_2
		\right) + \delta_n \left(
		{\|v\|_2 \over \sqrt m} + \bar{\rho}\|v\|_2
		\right) \rho \|u\|_1 \delta_n
		\\
		&\lesssim \delta_n\left(\sqrt{u^\top \C u} + \bar\rho\|u\|_2\right) \left(
		{\|v\|_2 \over \sqrt m} + \bar{\rho}\|v\|_2
		\right) .
		\end{align*}
		The last inequality uses $\delta_n \lesssim 1$.
	\end{proof}
	
	\bigskip
	
	\begin{lemma}\label{lem_wt_WW_D}
		Let $\wt \W$  and $D = \textrm{\rm diag}(d_1, \ldots d_K)$ be defined in (\ref{def_tildes}) and (\ref{def_D}), respectively. Under the conditions of Theorem \ref{thm_I}, for any fixed vectors $u,v\in\R^K$, with probability $1-\pn^{-c}$,
		\[
		\left|u^\top \left({1\over n}\wt \W^\top \wt \W - D\right) v\right| \lesssim\left({1\over m} +  \bar\rho^2\right)\|u\|_2\|v\|_2\delta_n.
		\]
	\end{lemma}
	
	\begin{proof}
		We work on the event $\E$ defined in (\ref{def_event}) that has probability at least $1-\pn^{-c}$ for some $c>0$. Recall that it implies $\wh K = K$ and $I_k \subseteq \wh I_k \subseteq I_k \cup J_1^k$ for all $k\in [K]$. 
		
		Recall that $\oW:= \W A_{I \sbt}[A_{I \sbt}^\top A_{I \sbt}]^{-1} =\W \Pi$. For any $a\in[K]$, we have
		\begin{equation}\label{eq_wt_W}
		\wt \W_{\sbt a} = {1\over \wh m_a}\sum_{i\in  I_a}\W_{\sbt i} +  {1\over \wh m_a}\sum_{i\in  L_a}\W_{\sbt i} = \oW_{\sbt a} + \underbrace{{|L_a|\over \wh m_a}\left({1\over |L_a|}\sum_{i\in  L_a}\W_{\sbt i} -\oW_{\sbt a}  \right)}_{R_a}.
		\end{equation}
		Expand the quadratic term as 
		\begin{align*}\nonumber
		u^\top \left({1\over n}\wt \W^\top \wt \W - D\right)v &= \sum_a u_av_a\left({1\over n}\wt \W_{\sbt a}^\top \wt \W_{\sbt a} - d_a\right)\\
		&\quad  + {1\over n}\sumsum_{a,b\in [K],a\ne b}u_av_b\left(\oW_{\sbt a}^\top \oW_{\sbt b} + \oW_{\sbt a}^\top R_b + \oW_{\sbt b}^\top R_a + R_a^\top R_b\right).
		\end{align*}
		Since $\X_{\sbt i} = \Z A_{i\sbt} + \W_{\sbt i}$, by  the definition of $d_a$ in (\ref{def_D}), we find
		\begin{align*}
		d_a& = {1\over \wh m_a^2}\sum_{i\in \wh I_a} {1\over n}\X_{\sbt i}^\top \X_{\sbt i} - {1\over \wh m_a^2 (\wh m_a-1)}\sumsum_{i,j\in \wh I_a, i\ne j}{1\over n}\X_{\sbt i}^\top \X_{\sbt i}\\
		& = {1\over \wh m_a^2}\sum_{i\in \wh I_a} \left({1\over n}A_{i\sbt}^\top \Z^\top \Z A_{i\sbt} + {1\over n}\W_{\sbt i}^\top \W_{\sbt i} + {2\over n}A_{i\sbt}^\top \Z^\top \W_{\sbt i} \right)\\
		&\quad - {1\over \wh m_a^2 (\wh m_a-1)}\sumsum_{i,j\in \wh I_a, i\ne j}\left[{1\over n}A_{i\sbt}^\top \Z^\top \Z A_{j\sbt} + {1\over n}\W_{\sbt i}^\top \W_{\sbt j} + {2\over n}A_{i\sbt}^\top \Z^\top \W_{\sbt j}\right].
		\end{align*}
		Rearranging terms yields
		\begin{align}\label{disp_d_a}
		d_a
		&= {1\over \wh m_a^2}\sumsum_{i,j\in \wh I_a} {1\over n}\W_{\sbt i}^\top \W_{\sbt j} - \underbrace{{1\over \wh m_a (\wh m_a-1)}\sumsum_{i,j\in \wh I_a, i\ne j}{1\over n}\W_{\sbt i}^\top \W_{\sbt j}}_{T_1^a}\\\nonumber
		&\quad  + \underbrace{{1\over \wh m_a^2}\sum_{i\in \wh I_a} {1\over n}A_{i\sbt}^\top \Z^\top \Z A_{i\sbt} - {1\over \wh m_a^2 (\wh m_a-1)}\sumsum_{i,j\in \wh I_a, i\ne j}{1\over n}A_{i\sbt}^\top \Z^\top \Z A_{j\sbt}}_{T_2^a}\\\nonumber
		&\quad + \underbrace{{2\over \wh m_a^2}\sum_{i\in \wh I_a}{1\over n}A_{i\sbt}^\top \Z^\top \W_{\sbt i}  - {2\over \wh m_a^2 (\wh m_a-1)}\sumsum_{i,j\in \wh I_a, i\ne j}{1\over n}A_{i\sbt}^\top \Z^\top \W_{\sbt j}}_{T_3^a}.
		\end{align}
		The first term on the right equals $n^{-1}\wt \W_{\sbt a}^\top \wt \W_{\sbt a}$. 
		We further have
		\begin{align}\label{eq_vWWv}\nonumber
		\left|u^\top \left(n^{-1}\wt \W^\top \wt \W - D\right)v\right|
		&\le  \|u\|_2\|v\|_2\max_a(|T_1^a| + |T_2^a| + |T_3^a|)\\
		&\quad + {1\over n}\left|\sumsum_{a,b\in [K],a\ne b}u_av_b\left(\oW_{\sbt a}^\top \oW_{\sbt b}+\oW_{\sbt a}^\top R_b +\oW_{\sbt b}^\top R_a+ R_a^\top R_b\right)\right|\\\nonumber
		&:= \D_1 + \D_2
		\end{align} 
		
		\noindent\textbf{To bound $\D_1$:}
		We first study $T_1^a$, $T_2^a$ and $T_3^a$, separately. For $T_1^a$, expand  $\wh I_a = I_a \cup L_a$ to get the bound
		\begin{align}\label{disp_15}
		&{1\over \wh m_a(\wh m_a - 1)}\left\{\left|\sumsum_{i,j\in I_a, i\ne j}{1\over n}\W_{\sbt i}^\top \W_{\sbt j}\right|+2\left|\sumsum_{i\in I_a, j\in L_a}  {1\over n}\W_{\sbt i}^\top \W_{\sbt j}\right| + \left|\sumsum_{i\ne j\in L_a}  {1\over n}\W_{\sbt i}^\top \W_{\sbt j}\right|\right\}\\\nonumber
		&\le ~{1\over \wh m_a(\wh m_a - 1)}\left|\sumsum_{i,j\in I_a, i\ne j}{1\over n}\W_{\sbt i}^\top \W_{\sbt j}\right| + {2|L_a| \over \wh m_a(\wh m_a-1)}\max_{j\in L_a}\left|\sum_{i\in I_a}{1\over n}\W_{\sbt i}^\top \W_{\sbt j}\right|\\
		&\qquad + {|L_a|(|L_a| - 1) \over \wh m_a (\wh m_a-1)}\max_{i,j\in \wh I_a, i\ne j} \left|{1\over n}\W_{\sbt i}^\top \W_{\sbt j}\right|.\nonumber
		\end{align}
		The first term is no greater than
		\[
		{1\over n}\left|\sum_{t=1}^n {1\over m_a}\sum_{i\in I_a}\W_{ti}{1\over m_a-1} \sum_{j\in I_a\setminus \{i\}}   \W_{tj}\right| = {1\over n}\left|\sum_{t=1}^n \oW_{ta}{1\over m_a-1} \sum_{j\in I_a\setminus \{i\}}   \W_{tj}\right|.
		\]
		Since $(m_a-1)^{-1}\sum_{j\in I_a\setminus \{i\}} \W_{tj}$ is $(\gamma_w/\sqrt{m_a-1})$-sub-Gaussian by the arguments of Lemma \ref{lem_W} and 
		\[
		\EE\left[\sum_{i\in I_a}\W_{ti}\sum_{j\ne i} \W_{tj}\right] = 0,
		\]
		invoking Lemma \ref{lem_bernstein} gives
		\[
		{1\over \wh m_a(\wh m_a - 1)}\left|\sumsum_{i,j\in I_a, i\ne j}{1\over n}\W_{\sbt i}^\top \W_{\sbt j}\right| \le c\w^2\sqrt{\log \pn \over nm_a(m_a-1)}
		\]
		with probability $1-\pn^{-c'}$.
		Note that $|L_a|/\wh m_a\le \|\rho\|_2 \le \bar{\rho}$ and $\wh m_a-1 \ge m_a$ when $L_a \ne \emptyset$. We further obtain
		\begin{align*}
		|T_1^a|&\le {1\over \wh m_a(\wh m_a - 1)}\left|\sumsum_{i,j\in I_a, i\ne j}{1\over n}\W_{\sbt i}^\top \W_{\sbt j}\right| + 2\bar\rho \max_{j\in L_a}{1\over n}\left|\sum_{t=1}^n \W_{tj}\oW_{ta}\right|+ \bar\rho^2\max_{i,j\in \wh I_a, i\ne j} \left|{1\over n}\W_{\sbt i}^\top \W_{\sbt j}\right|.
		\end{align*}
		Invoking Lemma \ref{lem_bernstein} and taking an union bound yields
		\begin{equation}\label{eq_T_21}
		\max_a|T_1^a|\le c\w^2\left({1\over \sqrt{m(m-1)}} \vee {\bar\rho \over \sqrt{m}}\vee  \bar\rho^2 \right)\delta_n,
		\end{equation}
		with probability greater than $1-c\pn^{-c'}$.
		We bound $T_2^a$ by writing 
		\begin{align*}
		T_2^a &= {m_a \over \wh m_a^2}{1\over n}\Z_{\sbt a}^\top \Z_{\sbt a} + {1\over \wh m_a^2}\sum_{i\in L_a}{1\over n}A_{i\sbt}^\top \Z^\top \Z A_{i\sbt}\\
		&\quad - {1\over \wh m_a^2 (\wh m_a-1)}\bigg\{m_a(m_a-1){1\over n}\Z_{\sbt a}^\top \Z_{\sbt a}+ 2m_a\sum_{i\in L_a}{1\over n}\Z_{\sbt a}^\top \Z A_{i\sbt}\\
		&\hspace{3.3cm}+
		\sum_{i\ne j\in L_a}{1\over n}A_{i\sbt}^\top \Z^\top \Z A_{j\sbt} \bigg\}\\
		& = {1\over \wh m_a^2}\sum_{i\in L_a}\left({1\over n}A_{i\sbt}^\top \Z^\top \Z A_{i\sbt} - {1\over n}\Z_{\sbt a}^\top \Z_{\sbt a}\right)- {2m_a\over \wh m_a^2 (\wh m_a-1)}\sum_{i\in L_a}{1\over n}\Z_{\sbt a}^\top \Z(A_{i\sbt}-\e_a)\\
		&\quad -{1\over \wh m_a^2 (\wh m_a-1) }
		\sum_{i\ne j\in L_a}\left({1\over n}A_{i\sbt}^\top \Z^\top \Z A_{j\sbt} - {1\over n}\Z_{\sbt a}^\top \Z_{\sbt a} \right).
		\end{align*}
		From (\ref{rate_Ajek}),  on the event $\E$ defined in (\ref{def_event}), we find
		\begin{equation}\label{disp_1_1}
		\left|{1\over n}\Z_{\sbt a}^\top \Z(A_{I\sbt}-\e_b) \right|\le {8\delta_n\over \nu}\max_\ell \left|{1\over n}\Z_{\sbt a}^\top \Z_\ell\right| \overset{\E}{\le} {8B_z\over \nu}\delta_n,
		\end{equation}
		for any $i\in L_b$ and $a\in [K]$, which in conjunction with $\|A_{j\sbt}\|_1\le 1$ further gives
		\begin{align*}
		\left|{1\over n}A_{i\sbt}^\top \Z^\top \Z A_{j\sbt} - {1\over n}\Z_{\sbt a}^\top \Z_{\sbt a} \right| &\le \left|{1\over n}(A_{i\sbt}-\e_a)^\top \Z^\top \Z A_{j\sbt}\right| +\left|{1\over n}\Z_{\sbt a}^\top \Z(A_{j\sbt}^\top -\e_a) \right|\\ 
		&\le  {16B_z\over \nu}\delta_n.
		\end{align*} 
		We thus obtain, on the event $\E$,
		\begin{align}\label{eq_T_22}\nonumber
		\max_a |T_2^a| &\le \max_a \left\{{|L_a| \over \wh m_a^2}+ {m_a |L_a| \over \wh m_a^2 (\wh m_a - 1)}+ {|L_a|(|L_a| - 1) \over \wh m_a^2 (\wh m_a-1)}\right\}{16B_z\over \nu}\delta_n\\
		& \le  {32B_z\over \nu}{\bar\rho \delta_n\over m}\lesssim  \bar\rho \sqrt{\log \pn\over nm^2}.
		\end{align}
		Regarding $T_3^a$, notice that
		\begin{align*}
		\sum_{i \in I_a}{1\over n}A_{i\sbt}^\top  \Z^\top  \W_{\sbt i} = m_a{1\over n}\Z_{\sbt a}^\top  \oW_{\sbt a} +  \sum_{i \in L_a}{1\over n}A_{i\sbt}^\top  \Z^\top  \W_{\sbt i}
		\end{align*}
		and
		\begin{align*}
		&\sumsum_{j,\ell\in \wh I_a,j\ne \ell } {1\over n}A_{i\sbt}^\top  \Z^\top  \W_{\sbt i}\\
		&= m_a(m_a-1){1\over n}\Z_{\sbt a}^\top \oW_{\sbt a} + m_a \sum_{\ell\in L_a}{1\over n}\Z_{\sbt a}^\top \W_{\sbt \ell} + \sumsum_{j\in L_a, \ell \in \wh I_a\setminus \{j\}}{1\over n}A_{j\sbt}^\top  \Z^\top  \W_{\sbt \ell}.
		\end{align*}
		After a bit algebra, by also using $\|A_{j\sbt}\|_1 \le 1$ for $1\le j\le p$, we can establish that 
		\begin{align*}
		|T_3^a|~&~ \lesssim {|L_a| \over \wh m_a^2}\max_{i\in \wh I_a}{1\over n}\|\Z^\top \W_{\sbt i}\|_\i +  {m_a|L_a| \over \wh m_a^2(\wh m_a - 1)}{1\over n}|\Z_{\sbt a}^\top \oW_{\sbt a}|\\
		~&~\le {|L_a| \over \wh m_a^2}\max_{i\in \wh I_a}{1\over n}\|\Z^\top \W_{\sbt i}\|_\i +  {|L_a| \over \wh m_a^2}{1\over n}|\Z_{\sbt a}^\top \oW_{\sbt a}|. 
		\end{align*}
		By (\ref{def_rho_k}) and (\ref{disp_rho_rho_bar}), 
		invoking Lemmas \ref{lem_W} -- \ref{lem_bernstein} and taking the union bounds  over $a\in [K]$, $i,j\in I_k\cup J_1^k$ yield 
		\begin{equation}\label{eq_T_23}
		\max_a |T_3^a| \lesssim \bar{\rho}\sqrt{\log \pn\over nm^2}
		\end{equation}
		with probability $1 - \pn^{-c}.$
		Collecting displays (\ref{eq_T_21}) - (\ref{eq_T_23}) concludes 
		\begin{equation}\label{bd_D1}
		\D_1 \lesssim \|u\|_2\|v\|_2 \left({1\over \sqrt{m(m-1)}} \vee {\bar\rho \over \sqrt{m}}\vee  \bar\rho^2 \right) \sqrt{\log \pn\over n},
		\end{equation}
		with probability greater than $1-c\pn^{-c'}.$\\

		\noindent \textbf{To bound $\D_2$:} We study the first term
		\[
		{1\over n}\sumsum_{a,b\in [K], a\ne b}u_av_b\oW_{\sbt a}^\top \oW_{\sbt b} = {1\over n}\sum_{t=1}^n\sum_{a}u_a\oW_{ta}\sum_{b\ne a}v_b\oW_{tb}.
		\]
		Since Lemma \ref{lem_W} guarantees $\sum_a u_a \oW_{ta}$ is $(\|u\|_2\w/\sqrt{m})$-sub-Gaussian and, similarly,  $\sum_{b\ne a}v_b\oW_{tb}$ is $(\|v\|_2\w/\sqrt{m})$-sub-Gaussian, invoking Lemma \ref{lem_bernstein} and noting that 
		\[
		\EE\left[\sum_a u_a\oW_{ta}\sum_{b\ne a}v_b\oW_{tb}\right] = \sumsum_{a,b\in [K], a\ne b}u_av_b\EE\left[\oW_{ta}\oW_{tb}\right]=0
		\]
		gives
		\begin{equation}\label{eq_over_W}
		\PP\left\{{1\over n}\sumsum_{a,b\in [K], a\ne b}u_av_b\oW_{\sbt a}^\top \oW_{\sbt b} \lesssim \|u\|_2\|v\|_2\w^2\sqrt{\log \pn\over nm^2} \right\}\ge 1-(p \vee n)^{-c}.
		\end{equation}
		We then bound the rest term by term. By recalling  (\ref{eq_wt_W}), we obtain
		\begin{align*}
		&{1\over n}\left|\sumsum_{a,b\in [K], a\ne b}u_a v_b\oW_{\sbt a}^\top R_b\right|\\
		&= {2\over n}\left|\sum_{b=1}^K{v_b \over \wh m_b}\sum_{t=1}^n \left(|L_b|\oW_{tb} - \sum_{i\in L_b}\W_{ti}\right)\sumsum_{a,b\in [K], a\ne b}u_a\oW_{ta}\right|\\
		&\le  2\|D_\rho v\|_1\max_{b\in [K]} \left|{1\over n}\sum_{t=1}^n \left(\oW_{tb} - {1\over |L_b|}\sum_{i\in L_b}\W_{ti}\right)\sumsum_{a,b\in [K], a\ne b}u_a\oW_{ta}\right|\\
		&\le 2\bar{\rho}\|v\|_2\max_{b\in [K]} \left|{1\over n}\sum_{t=1}^n \oW_{tb}\sumsum_{a,b\in [K], a\ne b}u_a\oW_{ta}\right|+2\bar{\rho}\|v\|_2\max_{
			\substack{b\in [K]\\ i\in L_b}} \left|{1\over n}\sum_{t=1}^n \W_{ti}\sumsum_{a,b\in [K], a\ne b}u_a\oW_{ta}\right|
		\end{align*}
		where we used  (\ref{def_rho_k}) and (\ref{disp_rho_rho_bar}) in the last two lines. 
		Invoke Lemmas \ref{lem_W} and \ref{lem_bernstein} and take union bounds to  conclude
		\begin{equation}\label{eq_WR}
		{1\over n}\left|\sumsum_{a,b\in [K], a\ne b}u_a v_b\oW_{\sbt a}^\top R_b\right| \le c \bar{\rho}\|u\|_2\|v\|_2 \sqrt{\log \pn \over nm}
		\end{equation}
		with probability $1-\pn^{-c'}$.
		A similar bound can be obtained for $|\sumsum_{a,b\in [K], a\ne b}u_a v_b\oW_{\sbt b}^\top R_a|$.
		Finally, to bound the last term, we write 
		\begin{align*}
		&{1\over n}\left|\sumsum_{a,b\in [K], a\ne b}u_a v_b R_b^\top R_a\right|\\\nonumber
		&= {1\over n}\left|\sum_{a\neq b}{u_av_b \over \wh m_b\wh m_a}\sum_{t=1}^n \left(|L_b|\oW_{tb} - \sum_{i\in L_b}\W_{ti}\right)\left(|L_a|\oW_{ta} - \sum_{i\in L_a}\W_{ti}\right)\right|,
		\end{align*}
		which can be  bounded above by
		\begin{align}\label{eq_RR_0}\nonumber
		&\bar{\rho}^2\|u\|_2\|v\|_2\max_{a\ne b}\left|{1\over n}\sum_{t=1}^n \left(\oW_{tb} - {1\over |L_b|}\sum_{i\in L_b}\W_{ti}\right)  \left(\oW_{ta} - {1\over |L_a|}\sum_{i\in L_a}\W_{ti}\right)\right|\\
		&\le \bar{\rho}^2\|u\|_2\|v\|_2\max_{a\ne b} \left|{1\over n}\sum_{t=1}^n \oW_{tb}\oW_{ta}\right|+2\bar{\rho}^2\|u\|_2\|v\|_2\max_{
			\substack{a\ne b\in [K]\\ i\in L_b}} \left|{1\over n}\sum_{t=1}^n \W_{ti}\oW_{ta}\right|\\\nonumber
		&\qquad +\bar{\rho}^2\|u\|_2\|v\|_2 \max_{
			\substack{a\ne b\in [K]\\ i\in L_a, j\in L_b}} \left|{1\over n}\sum_{t=1}^n \W_{ti}\W_{tj}\right|.
		\end{align}
		Invoke Lemmas \ref{lem_W} and \ref{lem_bernstein} and take the union bounds to obtain
		\begin{equation}\label{eq_RR}
		{1\over n}\left|\sumsum_{a,b\in [K], a\ne b}u_a v_bR_a^\top R_b\right| \lesssim \bar{\rho}^2\|u\|_2\|v\|_2\sqrt{\log \pn \over n}
		\end{equation}
		with probability $1-\pn^{-c}.$
		Combining (\ref{eq_over_W}), (\ref{eq_WR}) and (\ref{eq_RR}) concludes 
		\begin{equation}\label{bd_D2}
		\D_2 \lesssim \left({\|u\|_2 \over \sqrt m} + \bar{\rho}\|u\|_2\right)\left({\|v\|_2 \over \sqrt m} + \bar\rho \|v\|_2\right)\sqrt{\log \pn\over n}
		\end{equation}
		with probability $1-c\pn^{-c'}$.
		Finally, combine (\ref{bd_D1}) and (\ref{bd_D2}) to complete the proof.
	\end{proof}

	\subsubsection{Proof of Lemma \ref{lem_op_C}}
	
	We first state and prove the following lemma which is used for proving Lemma \ref{lem_op_C}. 
	
	\begin{lemma}\label{lem_1}
		Under the conditions of Theorem \ref{thm_I}, for any fixed vector $v, u \in \R^K$, we have 
		\[
		\left|u^\top (\whC-\C)v\right| \lesssim \left(\sqrt{u^\top \C u} +  \bar\rho\|u\|_2\right)\left(\sqrt{v^\top \C v} +  \bar\rho\|v\|_2\right)\delta_n
		\]
		with probability greater than $1-c\pn^{-c'}$ for some constant $c>0$.
	\end{lemma}
	\begin{proof}
		We work on the event $\E$ defined in (\ref{def_event}) that has probability at least $1-\pn^{-c}$ for some $c>0$. Recall that it implies $\wh K = K$ and $I_k \subseteq \wh I_k \subseteq I_k \cup J_1^k$ for all $k\in [K]$. 
		
		Recall that, from the proof of Theorem \ref{thm_beta}, $\whC = n^{-1}\wt \X^\top \wt \X - D$ with $D= \textrm{diag}(d_1,\ldots, d_K)$ and $d_k$ defined in (\ref{def_D}). By $\wt \X = \wt \Z + \wt \W$, we obtain
		\begin{align*}
		u^\top (\whC-\C)v
		&= u^\top \left({1\over n} \wt \X^\top \wt \X- \C\right) v - u^\top Dv\\
		&= u^\top \left({1\over n} \wt \Z^\top \wt \Z- \C\right) v + u^\top \left( {1\over n}\wt \W^\top \wt \W - D\right)v\\
		&\qquad  + {1\over n}u^\top \wt \Z^\top \wt \W v+{1\over n}v^\top \wt \Z^\top \wt \W u\\
		&:= T_1 + T_2 + T_3+T_4.
		\end{align*}
		Plugging $\wt \Z = \Z + \Z\D$ into $T_1$ yields
		\begin{align}\label{def_T1}\nonumber
		|T_1| &\le \left|u^\top \left( {1\over n}\Z^\top \Z  - \C\right)v\right| + {1\over n}\left|u^\top \Z^\top \Z \D v\right| + {1\over n}\left|u^\top \D^\top \Z^\top \Z v\right| + {1\over n}\left|u^\top \D^\top \Z^\top \Z \D v\right|\\\nonumber
		&\le \left|u^\top \left( {1\over n}\Z^\top \Z  - \C\right)v\right| + {1\over n}\left\|u^\top \Z^\top \Z \right\|_\i \|\D v\|_1+ {1\over n}\left\|\Z^\top \Z v\right\|_\i \|\D u\|_1\\
		&\qquad  + {1\over n}\max_{a}|\Z_{\sbt a}^\top \Z_{\sbt a}|\cdot \|\D u\|_1 \|\D v\|_1 
		\end{align}
		By using Lemmas \ref{lem_bernstein} and \ref{lem_deltaA_I}, with probability $1-\pn^{-c}$, one has 
		\begin{align}\label{eq_T_1}\nonumber
		|T_1| &\lesssim \delta_n\sqrt{u^\top \C u}\sqrt{v^\top \C v} + \delta_n \bar\rho \left(\|u\|_2\|v\|_2  \bar\rho \delta_n + \|u\|_2\sqrt{v^\top \C v} + \|v\|_2 \sqrt{u^\top \C u}\right)\\
		&\lesssim \delta_n \left(\sqrt{u^\top \C u}+ \bar\rho\|u\|_2
		\right)\left(\sqrt{v^\top \C v}+ \bar\rho\|v\|_2
		\right)
		\end{align}
		where we have used $\delta_n \lesssim 1$ in the second line. 
		The proof is completed by invoking Lemmas \ref{lem_Z_wt_W_wt} -- \ref{lem_wt_WW_D} for $T_2$ -- $T_4$ and using $u^\top \C u\ge \cl\|u\|_2^2$ to simplify expressions.  
	\end{proof}
	\medskip

	\begin{proof}[Proof of Lemma \ref{lem_op_C}]
		We again work on the event $\E$ such that $\wh K = K$ and $I_k \subseteq \wh I_k \subseteq I_k \cup J_1^k$ for all $k\in [K]$.
		
		To prove the upper bound for $\|\Omega^{1/2}(\whC- \C)\Omega^{1/2}\|_{{\rm op}}$, as in the proof of Lemma \ref{lem_1}, we consider the terms $T_1$ -- $T_4$ separately (except here $T_3 = T_4$ since $u=v$). Specifically, we will upper bound 
		\begin{align*}
		T_1' + T_2' + T_3' &:=  \sup_{v\in \S^{K-1}}\left| v^\top \Omega^{1/2}\left({1\over n} \wt \Z^\top \wt \Z- \C\right) \Omega^{1/2}v\right| + \|\Omega\|_{{\rm op}}\sup_{v\in \S^{K-1}}\left|v^\top \left( {1\over n}\wt \W^\top \wt \W - D\right)v\right|\\
		&\quad + \sup_{v\in \S^{K-1}}{2\over n}\left|v^\top \Omega^{1/2}\wt \Z^\top \wt \W \Omega^{1/2}v\ \right|. 
		\end{align*}
		For $T_1'$, (\ref{def_T1}) implies
		\begin{align*}
		T_1' &\le \left\|\Omega^{1/2}\left({1\over n}\Z^\top \Z  - \C\right)\Omega^{1/2}\right\|_{{\rm op}}+ 2\sup_{v\in \S^{K-1}} {1\over n}\left\|v^\top  \Omega^{1/2}\Z^\top \Z \right\|_\i \|\D \Omega^{1/2}v\|_1\\
		& \quad + \sup_{v\in \S^{K-1}} {1\over n}\max_{1\le a\le K}|\Z_{\sbt a}^\top \Z_{\sbt a}|\cdot \|\D \Omega^{1/2}v\|_1^2\\
		&\lesssim \delta_n\sqrt{K}+  \bar\rho \delta_n\sqrt{K/\cl}+  \bar\rho^2\delta_n^2{\sqrt K/\cl}
		\end{align*}
		with probability $1-\pn^{-c} - p^{-cK}$. The last inequality uses Lemma \ref{lem_deltaA_I} with $\|\Omega^{1/2}v\|_2\le \|v\|_2/\sqrt{\cl}$, (\ref{rate_op_zz_c}) in Lemma \ref{lem_op} and (\ref{rate_zz}) in Lemma \ref{lem_quad}. We thus find 
		\begin{equation}\label{eq_T_1'}
		T_1'\lesssim \delta_n\sqrt{K}(1\vee \bar\rho^2/\cl)
		\end{equation}
		with probability $1-\pn^{-c'}.$
		To bound $T_2'$, from (\ref{eq_vWWv}), it suffices to bound $\sup_{v\in \S^{K-1}}\D_2$ only, since the bound of $\D_1$ is uniformly in $v$. 
		Display (\ref{eq_vWWv}) gives 
		\begin{align*}
		\sup_{v\in \S^{K-1}}\D_2 &=   \sup_{v\in \S^{K-1}}  {1\over n}\left|\sumsum_{a,b\in [K], a\ne b}v_av_b\oW_{\sbt a}^\top \oW_{\sbt b}\right|\\
		&\quad +   \sup_{v\in \S^{K-1}} {1\over n}\left|\sumsum_{a,b\in [K], a\ne b}v_av_b\left(2 \oW_{\sbt a}^\top R_b +  R_a^\top R_b\right)\right|.
		\end{align*}
		By repeating a discretization argument similar to the one above, we can show from (\ref{eq_over_W}) that
		\begin{equation}\label{eq_over_W_sup}
		\PP\left\{ \sup_{v\in \S^{K-1}}  {1\over n}\left|\sumsum_{a,b\in [K], a\ne b}v_av_b\oW_{\sbt a}^\top \oW_{\sbt b}\right|\le c\sqrt{K\log \pn \over nm^2} \right\}\ge 1- \pn^{-c'K}
		\end{equation}
		and that, from (\ref{eq_WR}),
		\begin{equation}\label{eq_WR_sup}
		\sup_{v\in \S^{K-1}}{2\over n}\left|\sumsum_{a,b\in [K], a\ne b}v_a v_b\oW_{\sbt a}^\top R_b\right| \le c \bar\rho  \sqrt{K\log \pn \over nm}
		\end{equation}
		with probability $1-\pn^{-c'K}.$
		From (\ref{eq_RR_0}), the last term $n^{-1}\sum v_av_bR_a^\top R_b$ can be upper bounded by
		\begin{align*}
		&\sup_{v\in \S^{K-1}}{1\over n}\left|\sumsum_{a,b\in [K], a\ne b} v_av_bR_a^\top R_b\right|\\
		&\le
		\bar\rho^2\max_{a\ne b} \left|{1\over n}\sum_{t=1}^n \oW_{tb}\oW_{ta}\right|+\bar\rho^2\max_{
			\substack{a\ne b\in [K]\\ i\in L_b}} \left|{1\over n}\sum_{t=1}^n \W_{ti}\oW_{ta}\right|\\
		&\quad + \bar\rho^2\max_{
			\substack{a\ne b\in [K]\\ i\in L_a}} \left|{1\over n}\sum_{t=1}^n \W_{ti}\oW_{tb}\right|+\bar\rho^2\max_{
			\substack{a\ne b\in [K]\\ i\in L_a, j\in L_b}} \left|{1\over n}\sum_{t=1}^n \W_{ti}\W_{tj}\right|. 
		\end{align*}
		Invoking Lemmas \ref{lem_W} - \ref{lem_bernstein} and taking the union bound conclude
		that 
		\begin{equation}\label{eq_RR_sup}
		\sup_{v\in \S^{K-1}}{1\over n}\left|\sumsum_{a,b\in [K], a\ne b} v_av_bR_a^\top R_b\right| \le c\bar\rho^2\delta_n
		\end{equation}
		with probability $1-\pn^{-c'}.$
		Finally, from $\|\Omega\|_{{\rm op}} \le \cl^{-1}$, collecting (\ref{eq_over_W_sup}) -- (\ref{eq_RR_sup}) and invoking the bound of $\D_1$ via (\ref{eq_T_21}), (\ref{eq_T_22}) and (\ref{eq_T_23}) yield, with probability $1-c\pn^{-c'}$,
		\begin{equation}\label{eq_T2'}
		T_2' \le c\left({1\over m} \vee {\bar\rho^2 }\right)\cl^{-1}\delta_n\sqrt{K}.
		\end{equation}
		
		We then proceed to bound $T_3'$. From (\ref{disp_ztd_wtd}), by using Lemma \ref{lem_deltaA_I} and $\|\Omega^{1/2}v\|_2 \le 1/\sqrt{\cl}$ for any $v\in \S^{K-1}$, we know
		\[
		{1\over n}|v^\top  \Omega^{1/2}\wt \Z^\top \wt \W v|  \le {1\over n}|v^\top  \Omega^{1/2}\Z^\top \wt \W \Omega^{1/2}v| + 	\max_k{1\over n}|\Z_k^\top \wt \W\Omega^{1/2} v| \cdot 
		\bar\rho \delta_n/\sqrt{\cl}
		\]
		with probability $1-\pn^{-c}$. 
		By (\ref{rate_op_zwtd}) in Lemma \ref{lem_op} and $K\log\pn = O(n)$, we have
		\begin{equation}\label{eq_T3'}
		T_3' \le c\left({1 \over \sqrt m}\vee \bar{\rho}  \right)\delta_n\sqrt{K/\cl} 
		\end{equation}
		with probability $1-c\pn^{-c'}.$
		Collecting the bounds for $T_1'$, $T_2'$ and $T_3'$ completes the proof.
	\end{proof}

	\subsubsection{Proof of Lemma \ref{lem_Delta_c}}
	We work on the event $\E$ defined in (\ref{def_event}) that has probability at least $1-\pn^{-c}$ for some $c>0$. Recall that it implies $\wh K = K$ and $I_k \subseteq \wh I_k \subseteq I_k \cup J_1^k$ for all $k\in [K]$. 
	
	From definition, by adding and subtracting terms and using the fact $\wh m \ge 2$, we have
	\begin{align*}
	&|\wh \D_c - \D_c|\\ &\le \left|
	{\wh \tau^4\over \wh m-1}\sum_{a=1}^K{\wh \beta_a^2 \over \wh m}\sum_{i \in I_a} \left[\left(\e_k^\top \wh \Theta^+ \e_i\right)^2 -\left(\e_k^\top \Theta^+ \e_i\right)^2 \right]
	\right|\\
	&\quad  + \left|
	{\wh \tau^4\over \wh m-1}\sum_{a=1}^K\left[{\wh \beta_a^2 \over \wh m} - {\beta_a^2 \over m}\right]\sum_{i \in I_a} \left(\e_k^\top \Theta^+ \e_i\right)^2\right|+ 
	\left|{\wh \tau^4\over \wh m-1} - {\tau^4\over m-1}\right|\sum_{a=1}^K{\beta_a^2 \over m}\sum_{i \in I_a} \left(\e_k^\top \Theta^+ \e_i\right)^2\\
	&\le 
	{\wh \tau^4}{\|\wh \beta\|_2^2 \over \wh m}\max_{a\in [K]} \sum_{i \in I_a} \left|\left(\e_k^\top \wh \Theta^+ \e_i\right)^2 -\left(\e_k^\top \Theta^+ \e_i\right)^2 \right|\\
	&\quad  +\wh \tau^4 \max_{a\in [K]}  \left|{\wh \beta_a^2 \over \wh m} - {\beta_a^2 \over m}\right| 
	\sum_{a=1}^K\sum_{i \in I_a} \left(\e_k^\top \Theta^+ \e_i\right)^2+ 
	\left|{\wh \tau^4\over \wh m-1} - {\tau^4\over m-1}\right|\sum_{a=1}^K{\beta_a^2 \over m}\sum_{i \in I_a} \left(\e_k^\top \Theta^+ \e_i\right)^2.
	\end{align*}
	We bound each term separately. First, note that Theorem \ref{thm_I} guarantees $\wh m \ge m$ and (\ref{eq_beta_tau}) yields
	\[
	\wh \tau^4 \le \t^4 + (\t^2 + \wh \t^2)|\wh \t^2 - \t^2| = O_p(\t^4).
	\]
	Recall from (\ref{bd_minor_var}) that
	\begin{align}\label{bd_e_Theta_e}
	|\e_k^\top\Theta^+\e_i| \le \sqrt{\e_k^\top (\Theta^\top \Theta)^{-1}\e_k \over m}.
	\end{align}
	Provided that  
	\begin{align}\label{rate_e_Theta_e}
	\max_{i\in I}\left|\e_k^\top \wh \Theta^+ \e_i -\e_k^\top \Theta^+ \e_i\right| = o_p\left(\sqrt{\e_k^\top (\Theta^\top \Theta)^{-1}\e_k \over m}
	\right),
	\end{align}
	we can conclude 
	\begin{align*}
	&{\wh \tau^4}{\|\wh \beta\|_2^2 \over \wh m}\max_{1\le a\le K}\sum_{i \in I_a} \left|\left(\e_k^\top \wh \Theta^+ \e_i\right)^2 -\left(\e_k^\top \Theta^+ \e_i\right)^2 \right|\\ 
	&\le {\wh \tau^4}{\|\wh \beta\|_2^2 \over \wh m}\cdot m\cdot \max_{i\in I}
	\left(
	\left| \e_k^\top \wh \Theta^+ e_i\right| + \left| \e_k^\top \Theta^+ \e_i\right|
	\right)\left|
	\e_k^\top \wh \Theta^+ \e_i - \e_k^\top \Theta^+ \e_i
	\right|\\
	&=o_p\left(\tau^4{\|\beta\|_2^2 \over m} \e_k^\top (\Theta^\top \Theta)^{-1}\e_k \right) = o_p(\D_a \D_b).
	\end{align*}
	For the other two terms, note that 
	\begin{align*}
	\left|{\wh \beta_a^2 \over \wh m} - {\beta_a^2 \over m}\right|  &\le  (\|\wh \beta\|_2 + \|\beta\|_2){\|\wh \beta - \beta\|_2\over \wh m} + \|\beta\|_2^2 {|m - \wh m|\over \wh m m}\\
	& =O_p\left(
	\left(
	1\vee {\|\beta\|_2^2 \over m}
	\right) \left(\cl^{-1/2}\delta_n\sqrt{K} + \bar{\rho}\right)
	\right)= o_p(\D_a)
	\end{align*}
	by using (\ref{eq_beta_tau}), $|m-\wh m|/\wh m \le \bar{\rho}$, (\ref{rate_beta})  and $\delta_n\sqrt{K} = o(1)$. Moreover, $\wh \t^2 = O_p(\t^2)$, $\wh m \ge 2$, $m \ge 2$, $|m-\wh m|/\wh m \le \bar{\rho}$ and  Lemma \ref{lem_tau} yield
	\[
	\left|{\wh \tau^4\over \wh m-1} - {\tau^4\over m-1}\right| \le {(\wh \t^2 + \t^2)|\wh \tau^2 - \tau^2|\over \wh m-1} +  {\tau^4|m-\wh m| \over (\wh m-1)(m-1)} = o_p(1).
	\]
	By observing
	\[
	\sum_{a=1}^K{\beta_a^2 \over m}\sum_{i \in I_a} \left(\e_k^\top \Theta^+ \e_i\right)^2 \le {\|\beta\|_2^2 \over m}\left(\e_k^\top \Theta^+ \e_i\right)^2 \overset{(\ref{bd_e_Theta_e})}{\le} {\|\beta\|_2^2 \over m}{\e_k^\top (\Theta^\top \Theta)^{-1}\e_k \over m},
	\]
	and also using (\ref{bd_e_Theta_e}), we conclude 
	\begin{align*}
	&\wh \tau^4 \max_a  \left|{\wh \beta_a^2 \over \wh m} - {\beta_a^2 \over m}\right| 
	\sum_{a=1}^K\sum_{i \in I_a} \left(\e_k^\top \Theta^+ \e_i\right)^2 = o_p(\D_a\D_b),\\
	&\left|{\wh \tau^4\over \wh m-1} - {\tau^4\over m-1}\right|\sum_{a=1}^K{\beta_a^2 \over m}\sum_{i \in I_a} \left(\e_k^\top \Theta^+ \e_i\right)^2= o_p(\D_a \D_b).
	\end{align*}
	It then suffices to verify (\ref{rate_e_Theta_e}). By (\ref{eq_theta_hat_inv}), for any $i\in I$, we have 
	\begin{align*}
	\left|\e_k^\top (\wh \Theta^+- \Theta^+) \e_i\right|  &\le 
	\left|\e_k^\top (\Theta^T\Theta)^{-1}(\wh \Theta- \Theta)^TP_{\wh \Theta}^{\perp}\e_i\right| +  \left|\e_k^\top\Theta^+(\Theta-\wh \Theta)\wh \Theta^+\e_i\right|\\
	&\le \left\|\e_k^\top\Omega^{1/2} (H^TH)^{-1}(\wh H- H)^T\right\|_2 +   \left|\e_k^\top\Omega^{1/2}H^+(H-\wh H)\wh H^+\e_i\right|\\
	&\le \sqrt{p} \max_j \left|\e_k^\top\Omega^{1/2}(H^TH)^{-1}(\wh H - H)^T\e_j\right| +   {\Omega_{kk}^{1/2}}{\|H^+(\wh H - H)\|_{op}\over \sigma_K(\wh H)}.
	\end{align*}
	Invoking (\ref{disp_H})  and  part $(d)$ of Lemma \ref{lem_H_op}, we gives
	\begin{align*}
	\left|\e_k^\top (\wh \Theta^+- \Theta^+) \e_i\right| &\lesssim \sqrt{p} \delta_n \sqrt{\e_k^\top\Omega^{1/2}(H^\top H)^{-2}\Omega^{1/2}\e_k}+  \Omega_{kk}^{1/2}{\delta_n\sqrt{K} \over \sigma_K(H)}\\
	&\lesssim  \delta_n\sqrt{p \over \lambda_K(H^\top H)}\sqrt{\e_k^\top (\Theta^\top \Theta)^{-1}\e_k} +   \Omega_{kk}^{1/2}{\delta_n\sqrt{K} \over \sigma_K(H)}
	\end{align*}
	with probability $1-c\pn^{-c'}$. By $\sigma_K^2(H) = \lambda_K(H^\top H) \ge  m \cl $, invoking Assumption \ref{ass_H_prime} concludes the proof. 
	\qed
	
	\bigskip

	\section{Theoretical guarantees of $\wh\beta^{(I)}$: convergence rate and asymptotic normality}\label{sec_thm_beta_I}
	
	We provide statistical guarantees for the estimator defined in (\ref{est_beta_I}) with $\whC$ defined in (\ref{Chat}) and $\wh A_{{\wh I\sbt}}$ obtained from (\ref{est_AI_a}) -- (\ref{est_AI_b}). Their proofs can be found in \cite{bing2019essential}.
	The following theorem states the convergence rate of $\min_{P\in \H_K}\|\wh\beta^{(I)}-P\beta\|_2$. 
	\begin{thm}\label{thm_beta_I}
		Assume 
		Assumptions \ref{ass_model} -- \ref{ass_C} hold. Let $K\log\pn \le cn$ for some sufficiently small constant $c>0$. Then,
		with probability greater than $1-\pn^{-c'}$ for some constant $c'>0$, $\wh K = K$,
		the matrix	$\whC$ is non-singular and  the estimator
		$\wh \beta_{(I)}$ given by (\ref{est_beta_I}) satisfies:
		\begin{equation}\label{init_rate_beta}
		\min_{P\in \H_K}\left\|\wh \beta^{(I)} - P\beta\right\|_2 \lesssim   \left(1\vee {\|\beta\|_2\over \sqrt m}\right)\sqrt{K\log \pn \over n}.
		\end{equation}
	\end{thm}
	
	The following theorem establishes the asymptotic normality of each coordinate of $\wh\beta^{(I)}$. For ease of the presentation, we assume $\Gamma = \tau^2{\bI}_p$ and $|I_1| = \ldots = |I_K| = m$ while the proof holds for the general case. 
	\begin{thm}\label{thm_distr_I}
		Under 
		Assumptions \ref{ass_model}, \ref{ass_subg}, \ref{ass_J1_prime}, \ref{ass_C} and $K\log \pn = o(\sqrt{n})$, assume $\gamma_{\eps} / \sigma = O(1)$ and $\gamma_w / \tau = O(1)$. Then
		$\whC$ is non-singular  with probability tending to one,  and 
		for any $1\le k\le K$,
		\[
		\sqrt{n/U_k}\left(\wh \beta^{(I)}_k - \beta_k\right) \overset{d}{\to} N\left(0,1\right),\qquad \text{as}\quad n\to \i,
		\]
		where
		\[
		U_k =\left(\sigma^2 + {\t^2\over  m}\|\beta\|_2^2\right)\left(\Omega_{kk} + {\tau^2\over m}\|\Omega_{k\sbt }\|_2^2\right) + {\t^4 \over m^2(m-1)}\sum_{a=1}^K\beta_a^2\Omega_{ka}^2.
		\]
	\end{thm}
	
	To estimate the asymptotic variance $U_k$, we can also use a plug-in estimator by using $|\wh I_k|$ for each $1\le k\le K$, $\whC^{-1}$, $\wh\beta^{(I)}$, $\wh\tau_i^2$ for $1\le i\le p$ obtained as (\ref{est_Gamma}) and $\wh\sigma^2$ obtained as (\ref{est_sigma}) by using $\wh\beta^{(I)}$. The following proposition shows that the plug-in estimator $\wh U_k$ consistently estimates 
	the asymptotic variance $U_k$ of $\wh \beta^{(I)}_k$.
	
	\begin{prop}\label{prop_est_Q}
		Under 
		conditions of Theorem \ref{thm_distr_I},  we have 
		\[
		\left|{\wh U_k^{1/2}/ U_k^{1/2}} - 1\right|  = o_p(1).
		\]
		Consequently, we have
		\[
		\sqrt{n/\wh U_k}\left(\wh \beta^{(I)}_k - \beta_k \right)\overset{d}{\to} N(0, 1),\qquad \text{as }n\to\i, \quad k\in [K].
		\]
	\end{prop}

	\section{Data-driven choice of the tuning parameter $\delta$ in Algorithm \ref{alg1}}\label{app_cv}
	A selection procedure of choosing $\delta$ is proposed in \cite[Section 5.1]{LOVE}. For the reader's convenience, we restate it here as well as an illustrative example in \cite{LOVE}.  
	
	Display (\ref{def_delta}) specifies the theoretical  rate of $\delta$, but only up to constants that   depend on the underlying data generating mechanism. We propose below a data-dependent way to select $\delta$, based on data splitting.  Specifically, we split the data set into two independent parts, of equal sizes. On the first set, we calculate the sample covariance matrix $\wh \Sigma^{(1)}$. On the second set, we choose a  fine  grid of values $\delta_\ell = c_\ell\sqrt{\log p/n}$, with $1\le \ell\le M$, for $\delta$, by varying the proportionality constants $c_\ell$. For each $\delta_\ell$, we obtain the estimated number of clusters $\wh K(\ell)$ and the pure variable set $\wh I(\ell)$ with its partition $\wh \I(\ell)$. Then we construct the $|\wh I(\ell)|\times \wh K(\ell)$ submatrix $\wh A_{\wh I(\ell)}$ of $\wh A$, and estimate $\whC(\ell)$ via formula (\ref{Chat}). Finally, we calculate the $|\wh I(\ell)|\times |\wh I(\ell)|$ matrix $W_\ell = \wh A_{\wh I(\ell)}\whC(\ell)\wh A_{\wh I(\ell)}^T$. In the end, we have constructed a family $\mathcal{F} = \{ W_1, \ldots,  W_M\}$ of the fitted matrices  $W_\ell$, each corresponding to different $\wh \I(\ell)$ that depend in turn on $\delta_\ell$, for $\ell\in \{1,\ldots,M\}$. Define
	\begin{equation}\label{cvcrit}
	CV(\wh \I(\ell)) := \frac{1}{\sqrt{| \wh I(\ell)|\bigl(| \wh I(\ell)|-1\bigr)}}\left\|\wh \Sigma_{\wh I(\ell)\wh I(\ell)}^{(1)}-W_\ell\right\|_{\textrm{F-off}},
	\end{equation}
	where $\|B\|_{\textrm{F-off}}:= \| B -\text{diag}(B)\|_F$ denotes the Frobenius norm over the off-diagonal elements of a square matrix $B$. We choose $\delta^{cv}$ as the value $\delta_\ell$ that minimizes $CV(\wh \I(\ell))$ over the grid $\ell\in [M]$.
	To illustrate how the selection procedure works, we provide an example below.\\

	We consider a simple case, when $\C$ is diagonal and the signed permutation matrix $P$ is the identity, to illustrate our cross-validation method.\\
	
	\noindent\textbf{Example 1.}
	Let $\C = \textrm{diag}(\tau,\tau,\tau)$, $\I = \bigl\{\{1,2\},\{3,4\},\{5,6\}\bigr\}$ and
	\begin{eqnarray*}
		A = \begin{bmatrix}
			1 & 0 & 0 \\
			-1 & 0 & 0\\
			0 & 1 & 0\\
			0 & 1  & 0\\
			0 & 0 & -1\\
			0 & 0 & -1\\
			0.4 & 0.6 & 0\\
			-0.5 & 0 & 0.4
		\end{bmatrix},\quad A_I\C A_I^T = \begin{bmatrix}
			* & \tau & 0 & 0 & 0& 0\\
			\tau & * & 0 & 0 & 0& 0\\
			0 & 0 & * & \tau & 0 & 0\\
			0 & 0 & \tau & * & 0& 0\\
			0 & 0 & 0& 0& * & \tau\\
			0 & 0 & 0& 0 & \tau & * \\
		\end{bmatrix},
	\end{eqnarray*}
	where we use $*$ to reflect the fact that our algorithm  ignores the diagonal elements. For the true $I$ and $\I$, we have $\wh A_{I} = A_I$,
	\begin{eqnarray*}
		\left\|\wh \Sigma_{II}^{(1)}-A_I \whC A_I^T\right\|_{\textrm{F-off}} &\le&
		\left\|\wh \Sigma_{II}^{(1)}-\Sigma_{II}\right\|_{\textrm{F-off}}+\left\|A_I \whC A_I^T-\Sigma_{II}\right\|_{\textrm{F-off}}\\
		& \le &\left\|\wh \Sigma_{II}^{(1)}-\Sigma_{II}\right\|_{\textrm{F-off}} + \sqrt{|I|(|I|-1)}\cdot\|\whC-\C\|_\i.
	\end{eqnarray*}
	For
	$$\epsilon = \left(\max_{i\ne j}\left|\wh\Sigma^{(1)}_{ij}-\Sigma_{ij}\right|\right) \vee \left(\max_{i\ne j}\left|\wh\Sigma^{(2)}_{ij}-\Sigma_{ij}\right|\right),$$ 
	we obtain
	\begin{eqnarray*}
		CV(\I) = \frac{1}{\sqrt{| I|\bigl(|I|-1\bigr)}}\left\|\wh \Sigma_{II}^{(1)}-A_I \whC A_I^T\right\|_{\textrm{F-off}}\le 2\epsilon.
	\end{eqnarray*} 
	Suppose that  $\wh \I = \bigl\{\{1,2\},\{3,5\},\{4,6\}\bigr\}$, so $\wh I=I$, yet $\wh\I\ne\I$, we would have 
	\begin{eqnarray*}
		\wh A_{\wh I}\whC\wh A_{\wh I}^T = \begin{bmatrix}
			* & \wh \tau_1 & 0 & 0 & 0& 0\\
			\wh \tau_1 & * & 0 & 0 & 0& 0\\
			0 & 0 & * & 0 & \wh \tau_2 & 0\\
			0 & 0 & 0 & * & 0 & \wh \tau_3\\
			0 & 0 & \wh \tau_2&  0& * & 0\\
			0 & 0 & 0& \wh \tau_3 & 0 & * \\
		\end{bmatrix}
	\end{eqnarray*}
	and 
	\[
	\wh A_{\wh I}\wh C\wh A_{\wh I}^T - \Sigma_{\wh I\wh I} = 
	\begin{bmatrix}
	* &  \Delta \tau_1 & 0 & 0 & 0& 0\\
	\Delta \tau_1 & * & 0 & 0 & 0& 0\\
	0 & 0 & * & \bm{-\tau} & \bm{\wh \tau_2} & 0\\
	0 & 0 & \bm{-\tau} & * & 0& \bm{\wh \tau_3}\\
	0 & 0 &\bm{ \wh \tau_2}& 0& * & \bm{-\tau}\\
	0 & 0 & 0& \bm{\wh \tau_3} & \bm{-\tau}& * \\
	\end{bmatrix}.
	\]
	Here $\Delta\tau_a = \wh\tau_a - \tau_a$, using estimates $\wh \tau_a$ defined in lieu of $[\whC]_{aa}$ from (\ref{Chat}) for each $a\in [\wh K]$. Thus, the cross-validation criterion  in (\ref{cvcrit}) would satisfy 
	\begin{eqnarray*}
		CV(\wh \I) &\ge&  \frac{1}{\sqrt{|\wh I|\bigl(|\wh I|-1\bigr)}} \left\|\wh A_{\wh I}\whC\wh A_{\wh I}^T - \Sigma_{\wh I\wh I}\right\|_{\rm{F-off}} -\frac{1}{\sqrt{| \wh I|\bigl(|\wh I|-1\bigr)}}\left\|\wh \Sigma_{\wh I\wh I}^{(1)}-\Sigma_{\wh I\wh I}\right\|_{\textrm{F-off}}\\
		& \ge & \sqrt{\frac{4\tau^2 + 2\wh\tau_2^2+2\wh \tau_3^2}{|\wh I|\bigl(|\wh I|-1\bigr)}}-2\epsilon.
	\end{eqnarray*}
	From noting that $|\wh \tau_a-\tau|\le \epsilon$, for $a = 2,3$, it gives
	\begin{eqnarray*}
		CV(\wh \I) \ge \sqrt{\frac{4\tau^2-4\tau\epsilon+2\epsilon^2}{15}}-2\epsilon > 2\epsilon \ge CV(\I),
	\end{eqnarray*}
	for $\tau \ge 9\epsilon$. We conclude in this example, with $\wh I=I$, incorrectly specifying $  \I$ will induce a large loss. It is easily   verified that this is also the case when $\wh I = I$ but $\wh K \ne K$ and $\wh\I \ne \I$. 
	
	On the other hand, suppose we mistakenly  included    some non-pure variable in $\wh I$. For instance, suppose we found $\wh \I = \bigl\{\{1,2\},\{3,4\},\{5,6,7\}\bigr\}$. Then we would  have
	\begin{eqnarray*}
		\Sigma_{\wh I'\wh I'} = \begin{bmatrix}
			* & \tau & 0 & 0 & 0& 0 & 0.4\tau\\
			\tau & * & 0 & 0 & 0& 0 & -0.4\tau\\
			0 & 0 & * & \tau & 0 & 0 & 0.6\tau\\
			0 & 0 & \tau & * & 0& 0 & 0.6\tau\\
			0 & 0 & 0& 0& * & \tau & 0\\
			0 & 0 & 0& 0 & \tau & * & 0\\
			0.4\tau & -0.4\tau & 0.6\tau & 0.6\tau & 0 & 0 & *
		\end{bmatrix}
	\end{eqnarray*}
	and 
	\[
	\wh A_{\wh I'}\whC\wh A_{\wh I'}^T =  
	\begin{bmatrix}
	* & \wh\tau_1 & 0 & 0 & 0& 0 & 0\\
	\wh\tau_1 & * & 0 & 0 & 0& 0 & 0\\
	0 & 0 & * & \wh \tau_2 & 0 & 0 & 0\\
	0 & 0 & \wh \tau_2 & * & 0& 0 & 0\\
	0 & 0 & 0& 0& * & \wh\tau_3 & \wh\tau_3\\
	0 & 0 & 0& 0 & \wh\tau_3 & * & \wh\tau_3\\
	0 & 0& 0 & 0& \wh\tau_3 & \wh\tau_3 & *
	\end{bmatrix}.
	\]
	We thus have 
	\begin{eqnarray*}
		\wh A_{\wh I'}\whC\wh A_{\wh I'}^T - \Sigma_{\wh I'\wh I'} = 
		\begin{bmatrix}
			* &  \Delta \tau_1 & 0 & 0 & 0& 0 & \bm{-0.4\tau}\\
			\Delta \tau_1 & * & 0 & 0 & 0& 0 & \bm{0.4\tau}\\
			0 & 0 & * & \Delta \tau_2& 0 & 0 & \bm{-0.6\tau}\\
			0 & 0 & \Delta \tau_2 & * & 0& 0 &\bm{-0.6\tau}\\
			0 & 0 & 0& 0 & * & \Delta \tau_3 & \bm{\wh\tau_3}\\
			0 & 0 & 0& 0 & \Delta \tau_3 & * & \bm{\wh\tau_3}\\
			\bm{-0.4\tau} & \bm{0.4\tau} & \bm{-0.6\tau} & \bm{-0.6\tau} & \bm{\wh\tau_3} & \bm{\wh\tau_3} & *
		\end{bmatrix}
	\end{eqnarray*}
	and, by similar arguments, for $\tau \ge 12\epsilon$, we find
	\begin{eqnarray*}
		CV(\wh \I') \ge 
		\sqrt{\frac{4\wh \tau_3^2 + 4\times 0.36\tau^2+4\times 0.16\tau^2}{42}}-2\epsilon > 2\epsilon.
	\end{eqnarray*}
	Thus, the cross-validation loss in this example will be large even if  only one non-pure variable is mistakenly classified as pure variable. 
	In rare cases, the cross-validation criterion might miss a very small subset of $I$ but this can be rectified in our later estimation of $A_J$.

	\section{Additional simulation results: histograms of the standardized estimates of $\beta_1$}\label{sec_histograms}
	
	We further corroborate the validity of Theorem \ref{thm_distr} and Proposition \ref{prop_est_V}
	in Figure \ref{fig_hist_1}. This figure depicts histograms based on 200 values  of  $\sqrt{n/\wh V_k}(\wh\beta_1 - \beta_1)$.

	\begin{figure}[ht]
		\begin{tabular}{c}
			\includegraphics[width=\textwidth,height=0.25\textheight]{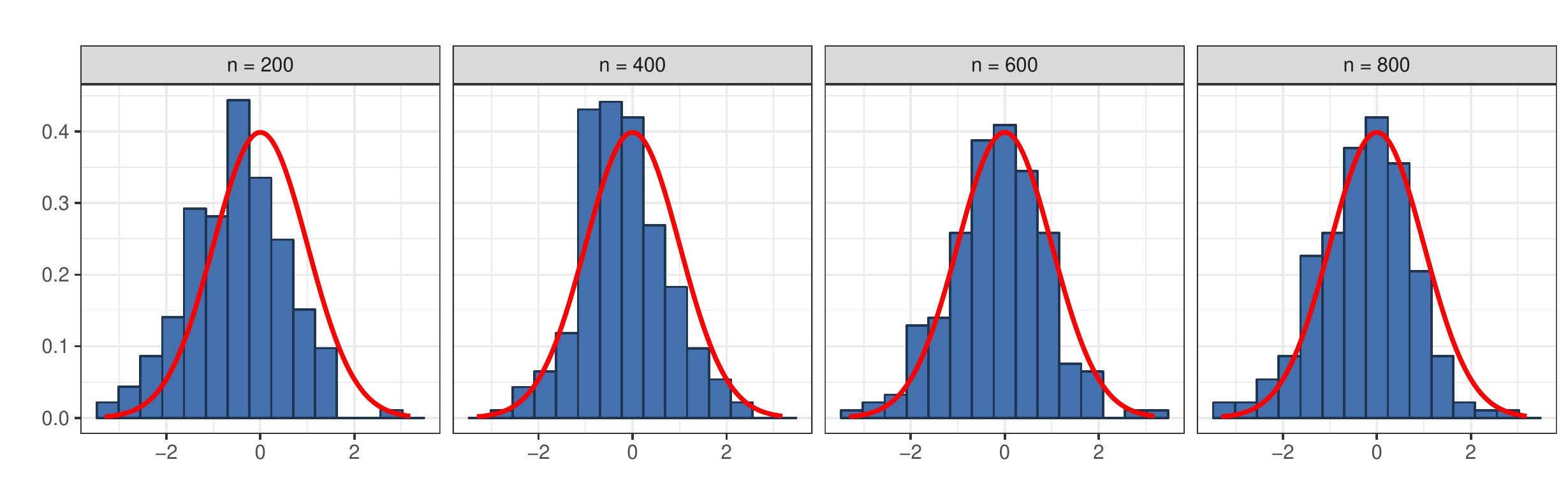}\vspace{-3mm}\\
			\includegraphics[width=\textwidth,height=0.25\textheight]{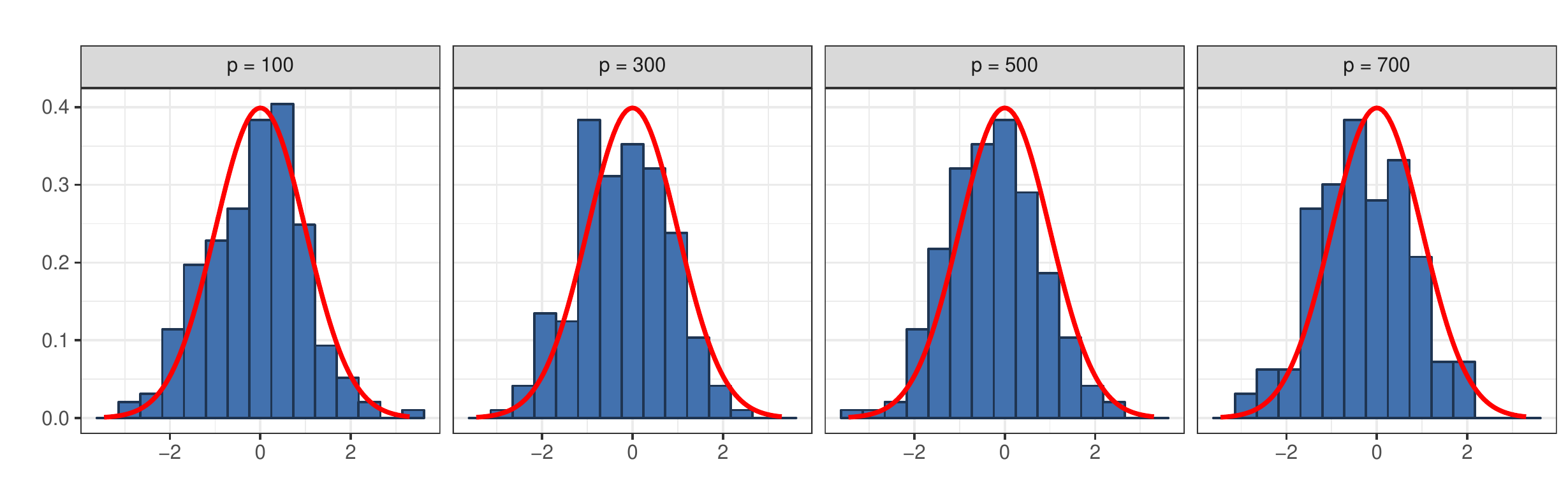}\vspace{-3mm}\\
			\includegraphics[width=\textwidth,height=0.25\textheight]{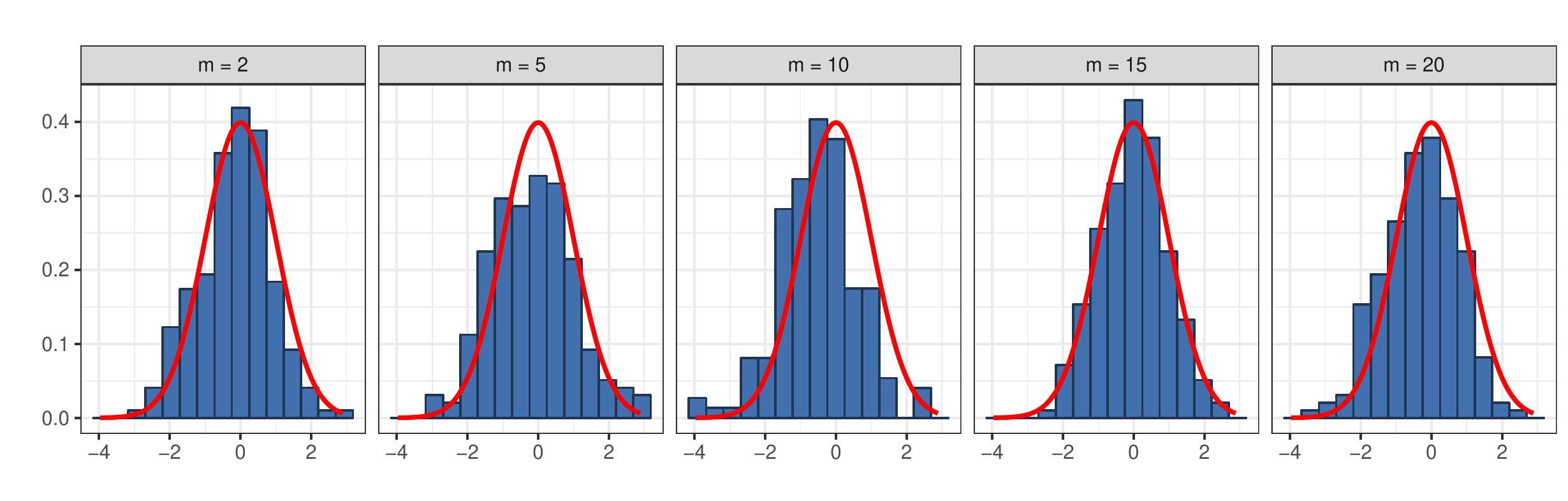}\vspace{-3mm}\\
			\includegraphics[width=\textwidth,height=0.25\textheight]{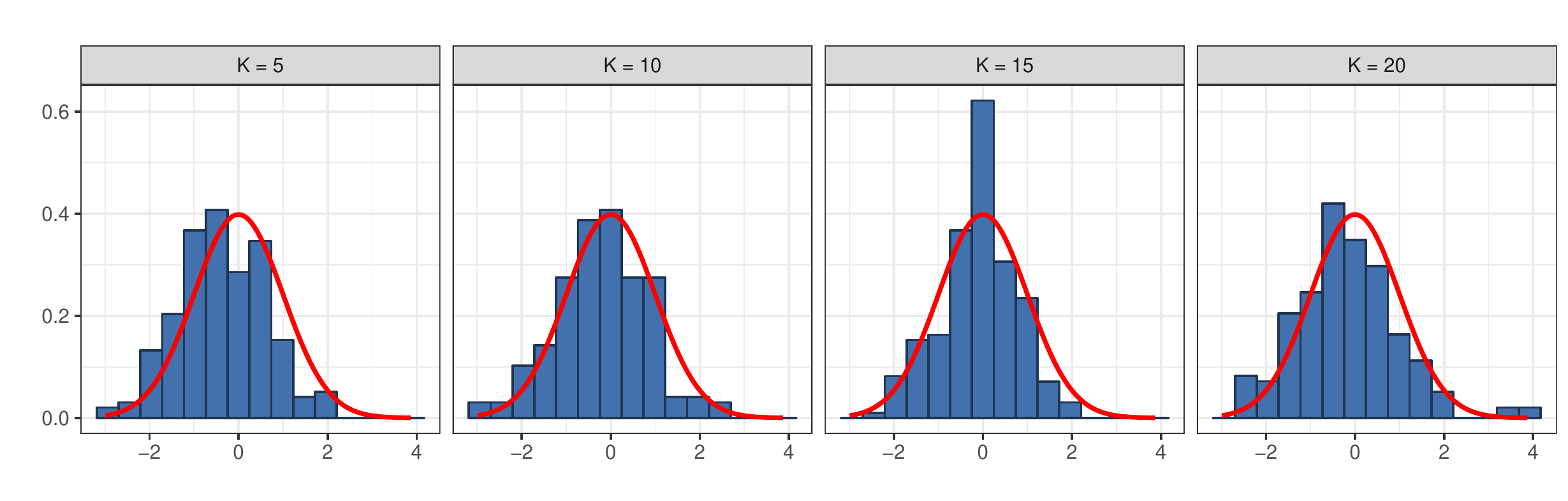}\vspace{-3mm}
		\end{tabular}
		\caption{Histograms of the standardized $\wh\beta_1$. The red curves are the density of $N(0,1)$.}
		\label{fig_hist_1}
	\end{figure}

	\section{Additional simulation results: estimation and inference of $\beta$ when the number of latent factors is not consistently estimated}\label{app_sim}
	
	We discuss the impact of selecting $\wh K$ with $\wh K\ne K$ on the estimation and inference of $\beta$, particularizing to the case when some columns of $A$ have very weak signals such that our estimate $\wh K$ is smaller than the true $K$. We  start by offering some intuition here.  Intuitively, when the submatrix $A_{\cdot S^c}$ contains many zero entries for some index set $S\subseteq [K]$, our procedure of estimating $K$ is likely to miss the latent factors in $S^c$, but may still recover  $Z_S$. 
	As a result, only $A_{\cdot S}$ can be well estimated. Consider the simple case $\C = \bI_K$ and recall that $\beta = (A^TA)^{-1}A^T\Sigma_{XY}$. Replacing $A$ by $A_{\cdot S}$ yields
	\[
	(A_{\cdot S}^T A_{\cdot S})^{-1}A_{\cdot S}^T \Sigma_{XY} = (A_{\cdot S}^T A_{\cdot S})^{-1}A_{\cdot S}^T A\beta = \beta_S + \underbrace{(A_{\cdot S}^T A_{\cdot S})^{-1}A_{\cdot S}^T A_{\cdot S^c}\beta_{S^c}}_{\Delta}.
	\]
	Therefore, when $A_{\cdot S^c}$ is small such that $\Delta$ is small, we still have 
	\[
	(A_{\cdot S}^T A_{\cdot S})^{-1}A_{\cdot S}^T \Sigma_{XY} ~ \approx ~ \beta_S.
	\]
	When $\C$ is not diagonal, the dependence among the latent factors also affects this approximation.\\ 
	
	To empirically investigate the impact of the estimating error of $\wh K$ on the subsequent inferential result, we conduct the following two simulation studies.\\ 
	
	(1) We take the same generating mechanism as described in Section 6 and consider $p = 400$, $n = 300$, $K = 10$ and $m = 5$. The matrix $\C$ is set to be $\sigma_Z^2\bI_K$ with $\sigma_Z^2 = 3$. After generating the matrix $A$, we draw an index set from $p$ i.i.d. Bernoulli$(1-\theta)$ with $\theta \in (0,1)$, and manually change the entries $A_{jK}$ into zero, for all $j$ in this index set.
	The parameter $\theta$ controls the overall sparsity (or signal) of the $K$th column of $A$. Small values of $\theta$ correspond to small signal $A_{\cdot K}$. 
	Note that $A$ no longer necessarily meets our model requirement (A1), since 
	the entries in the $K$th column of $A$ corresponding to pure variables are not necessarily  1, and  allowed to be set to zero.  We allowed for this misspecification since we empirically found that  as long as there exists at least two entries in the $K$th column of $A$ corresponding to pure variables (even though the rest entries of $A_{\cdot K}$ are zero), our algorithm continues to  consistently estimate $K$. 
	
	Our goal is to verify if the resulting estimator $\wh\beta$ estimates $\beta$ well expect for the $K$th coordinate and to examine the empirical coverage of the 95\% confidence intervals of our proposed estimator. 
	The same estimators mentioned in Section 6 of the main paper are considered. For the oracle estimator $\wh\beta_{\rm oracle}$, to illustrate the effect of using a subset of $Z$, we change it to 
	\[
	\wh\beta_{\rm oracle} = \left({\bf Z}_{\cdot \wh K}^T {\bf Z}_{\cdot \wh K}\right)^{-1}{\bf Z}_{\cdot \wh K}^T {\bf Y}
	\] 
	where $\wh K$ is estimated from Algorithm 1. 
	We vary $\theta \in \{0.1, 0.3, 0.5, 0.8\}$ and, for each setting, we calculate the averaged mean squared error of the subvector $[\wh\beta - \beta]_{S(1)}$, with $S(1) := \{1,2,\ldots, K-1\}$, for different estimators, as well as the averaged coverage of the 95\% confidence intervals for $\beta_1$, over 200 repetitions. These results together with the averaged estimated $K$ are reported in Table \ref{tab_new}. 
	
	{\bf Result:} When $\theta$ is small, that is, $A_{\cdot K}$ has very weak signal, our procedure of estimating $K$ is likely to miss the $K$th latent factor, leading to $\wh K = 9$, as expected.  The estimation of the first $K-1$ entries of $\beta$ 
	is only slightly affected. 
	Our proposed estimator $\wh\beta$ has the second best performance, after the oracle estimator. The coverage of the 95\% confidence intervals constructed from $\wh\beta$ is slightly under the nominal 95\% level. As $\theta$ increases ($\theta \ge 0.5$), that is, the signal of $A_{\cdot K}$ gets larger, we   consistently estimate $K$, and consequently, our proposed estimator performs increasingly better in terms of both the estimation error and the coverage of 95\% confidence intervals.\\

	(2) We now consider the case when there are multiple, sparse columns of $A$.
	We take the same setting $p = 400$, $n = 300$, $K = 10$, $m = 5$ and $\C = \sigma_Z^2\bI_K$ with $\sigma_Z^2 =3$. For each  $n_A \in \{1,2,\ldots, 5\}$, we threshold the last $n_A$ columns of $A$ according to the procedure described in the previous paragraph with $\theta = 0.1$. For each $n_A$, let $S(n_A) = \{1,2,\ldots,K-n_A\}$. The  averaged mean squared error of $[\wh\beta - \beta]_{S(n_A)}$, the averaged coverage of the 95\% confidence intervals for $\beta_1$ and the averaged estimated $K$ are reported in Table \ref{tab_new}. 
	
	{\bf Result:} As expected, the estimated number of factors is smaller than the true value, $10$, and is close to
	$K-n_A$, the number of columns of $A$ that have strong signals. As long as $n_A \le 4$, the estimation of $\beta_{S(n_A)}$ seems fairly good, and $\wh\beta$ still has the best performance (after $\wh \beta_{\rm oracle}$). The coverage of the 95\% confidence intervals constructed from $\wh\beta$ is   still close to, though slightly lower than, 95\%. As $n_A$ increases, more latent factors have weak signals. 
	This makes estimation of $K$ more difficult and clearly affects the estimation of $\beta$ and the coverage of confidence intervals.\\

	(3) We further investigate the effect of the correlation of $Z$ on the impact of inconsistently estimating $K$. We choose $\C$ as $[\C]_{ij} = \sigma_Z^2(-1)^{i+j}\rho^{|i-j|}$ for each $i,j\in [K]$ with $\sigma_Z^2 = 3$. We vary $\rho \in \{0.1, 0.2, 0.3,0.4, 0.5\}$ and fix $\theta = 0.1$, $n_A=1$, $p = 400$, $n = 300$, $K=10$ and $m=5$. The averaged mean squared error of the subvector $[\wh\beta - \beta]_{S(1)}$, the averaged coverage of the 95\% confidence intervals for $\beta_1$ and the averaged estimated $K$ are reported in Table \ref{tab_new}. 
	
	{\bf Result:} The estimation errors of all estimators (including $\wh\beta_{\rm oracle}$) increase as $\rho$ gets larger. $\wh\beta^{(I)}$ turns out to be more robust than $\wh\beta$ in the presence of correlated factors. One possible explanation is that $\wh\beta^{(I)}$ targets $\C^{-1}(A_{I\cdot}^TA_{I\cdot})^{-1}A_{I\cdot}^T\Cov(X_I, Y)$ which, due to the diagonal structure of $A_{I\cdot}^TA_{I\cdot}$, is less affected by a non-diagonal $\C$ compared to $\wh\beta$. The coverage of the 95\% confidence intervals constructed from both $\wh\beta$ and $\wh\beta^{(I)}$ is close to 95\%, and slightly decreases as $\rho$ gets larger.

	\begin{table}[ht]
		\centering
		{\renewcommand{\arraystretch}{1.2}
			\begin{tabular}{lccccc|cc|cc|c}
				\hline
				& $\wh\beta$ & $\wh\beta^{(I)}$ &  $\wh\beta_{\rm naive}$ &  $\wh\beta^{(A)}$   & $\wh\beta_{\rm oracle}$ & \multicolumn{2}{c|}{CIs of $\wh\beta$} &  \multicolumn{2}{c|}{CIs of $\wh\beta^{(I)}$} & $\wh K$\\\hline  
				\multicolumn{6}{l|}{
					{\bf	Vary $\theta$ with $n = 300$, $p = 400$, $K = 10$, $m=5$}} & coverage & length & coverage & length & \\
				$\theta = 0.1$ & 0.07 & 0.08 & 0.12 & 0.08 & 0.02 & 0.90  & 0.84 & 0.91 & 0.88 & 9.1\\ 
				$\theta = 0.3$  & 0.06 & 0.06 & 0.11 & 0.07 & 0.03 & 0.90  & 0.84 & 0.92 & 0.87 & 9.0 \\ 
				$\theta = 0.5$  & 0.03 & 0.04 & 0.10 & 0.06 & 0.00 & 0.93  & 0.68 & 0.93 & 0.72 & 10.0 \\ 
				$\theta = 0.8$  & 0.03 & 0.03 & 0.11 & 0.06 & 0.00 & 0.95  & 0.64 &  0.94 & 0.67 & 10.0 \\ \hline
				\multicolumn{6}{l|}{{\bf Vary $n_A$ with $n = 300$, $p = 400$, $K=10$, $m=5$}} & \multicolumn{2}{c|}{}  & \multicolumn{2}{c|}{} & \\
				$n_A = 1$ &	0.07 & 0.07 & 0.12 & 0.08 & 0.02 & 0.92 & 0.83 & 0.92 & 0.87 & 9.1 \\ 
				$n_A = 2$ & 0.10 & 0.11 & 0.13 & 0.11 & 0.05 & 0.91 & 1.05 & 0.93 & 1.09 & 8.1 \\ 
				$n_A = 3$ & 0.12 & 0.13 & 0.14 & 0.13 & 0.07 & 0.91  & 1.14 & 0.93 & 1.19 & 7.1 \\ 
				$n_A = 4$ & 0.13 & 0.15 & 0.13 & 0.13 & 0.07 &  0.93  & 1.16 & 0.93 & 1.21 &  6.1 \\ 
				$n_A = 5$ & 0.18 & 0.20 & 0.18 & 0.18 & 0.10 & 0.89  & 1.32 & 0.89 & 1.37 & 5.2 \\ 
				\hline
				\multicolumn{6}{l|}{{\bf Vary $\rho$ with $n = 300$, $p = 400$, $K=10$, $m=5$}} & \multicolumn{2}{c|}{}  & \multicolumn{2}{c|}{} & \\
				$\rho =0.1$ & 0.07 & 0.07 & 0.16 & 0.10 & 0.03 & 0.94 & 0.83 & 0.94 & 0.87 & 9.1 \\ 
				$\rho =0.2$ & 0.11 & 0.09 & 0.24 & 0.16 & 0.06 & 0.95  & 0.83 & 0.96 & 0.88 & 9.0 \\ 
				$\rho =0.3$ & 0.16 & 0.13 & 0.33 & 0.22 & 0.09 & 0.92  & 0.85 & 0.94 & 0.89 & 9.0 \\ 
				$\rho =0.4$ & 0.24 & 0.19 & 0.46 & 0.31 & 0.15 & 0.92  & 0.85 & 0.93 & 0.91 & 9.0 \\ 
				$\rho =0.5$ & 0.34 & 0.28 & 0.61 & 0.42 & 0.21 & 0.91  & 0.87 & 0.93 & 0.96 & 9.0 \\ 
				\hline
		\end{tabular}}\vspace{2mm}
		\caption{$\ell_2$ error  of various estimators, the coverage and the averaged length of the 95\% CIs of $\beta_1$ and the estimated number of latent factors.}
		\label{tab_new}
	\end{table}

	\end{document}